\documentclass[11pt,a4paper]{article}
\usepackage[T1]{fontenc}
\usepackage[utf8]{inputenc}
\usepackage[english]{babel}
\usepackage{authblk}
\usepackage{amsthm}
\usepackage{color}
\usepackage[colorlinks=true,urlcolor=blue,linkcolor=blue]{hyperref}
\usepackage{graphicx}
\usepackage[bottom]{footmisc}
\usepackage{fullpage}
\usepackage{calc}
\usepackage{amsmath,amssymb}
\usepackage{amsfonts}
\usepackage{mathrsfs}
\usepackage[nottoc]{tocbibind}
\usepackage{graphicx}
\graphicspath{{images/}}
\usepackage{fancyhdr}
\usepackage{tabularx}
\usepackage{array}
\usepackage{lscape}
\usepackage{array,multirow,makecell}
\setcellgapes{1pt}
\makegapedcells
\usepackage[table]{xcolor}
\usepackage{float}
\usepackage{caption}
\usepackage{subcaption}
\usepackage{pifont}
\usepackage[ruled,vlined,linesnumbered,noresetcount]{algorithm2e}
\usepackage{booktabs}
\usepackage{titlesec}
\usepackage[section]{placeins}
\usepackage{stmaryrd}
\usepackage{enumerate}
\usepackage{bbm}
\usepackage{tcolorbox}
\newcolumntype{R}[1]{>{\raggedleft\arraybackslash }b{#1}}
\newcolumntype{L}[1]{>{\raggedright\arraybackslash }b{#1}}
\newcolumntype{C}[1]{>{\centering\arraybackslash }b{#1}}

\newcommand\numberthis{\addtocounter{equation}{1}\tag{\theequation}}
\newtheorem{thm}{Theorem}
\newtheorem{prop}{Proposition}
\newtheorem{prop_appendix}{Proposition}[subsection]
\newtheorem{defi}{Definition}
\newtheorem{coro}{Corollary}
\newtheorem{rem}{Remark}
\newtheorem{ass}{Assumption}

\makeatletter
\renewcommand\paragraph{\@startsection{paragraph}{4}{\z@}
  {-3.25ex \@plus -1ex \@minus -0.2ex}
  {2.25ex \@plus .25ex}
  {\normalfont\normalsize\bfseries}}
\renewcommand\subparagraph{\@startsection{subparagraph}{5}{\z@}
  {-3.25ex \@plus -1ex \@minus -0.2ex}
  {2.25ex \@plus .25ex}
  {\normalfont\normalsize\bfseries}}
\def\toclevel@paragraph{4}
\def\toclevel@paragraph{5}
\def\l@paragraph{\@dottedtocline{4}{7em}{4em}}
\def\l@subparagraph{\@dottedtocline{5}{7em}{4em}}
\makeatother

\title{How Option Hedging Shapes Market Impact}
\author[]{Emilio Said}

\affil[]{Chaire de Finance Quantitative, Laboratoire MICS, CentraleSupélec, Université Paris-Saclay, Gif-Sur-Yvette, France}

\date{October 7, 2019}

\begin{document}

\maketitle

\begin{abstract}
We present a perturbation theory of the market impact based on an extension of the framework proposed by \cite{loeper2018option} -- originally based on \cite{liu2005option} -- in which we consider only \textit{local linear market impact}. We study the execution process of hedging derivatives and show how these \textit{hedging metaorders} can explain some stylized facts observed in the empirical market impact literature. As we are interested in the execution process of hedging we will establish that the arbitrage opportunities that exist in the discrete time setting vanish when the trading frequency goes to infinity letting us to derive a pricing equation. Furthermore our approach retrieves several results already established in the option pricing literature such that the spot dynamics modified by the market impact. We also study the relaxation of our \textit{hedging metaorders} based on the \textit{fair pricing} hypothesis and establish a relation between the \textit{immediate impact} and the \textit{permanent impact} which is in agreement with recent empirical studies on the subject.
\end{abstract}

\textbf{Keywords:} \textit{Market microstructure, option pricing, market impact, metaorders execution, option hedging, metaorders relaxation, fair pricing.}

\section{Introduction}
\label{introduction}

This paper presents a perturbation theory of the market impact. We will consider the framework of covered options. To illustrate our perturbative approach, let us consider an option's hedger who has to deal with a feedback mechanism between the underlying price dynamics and the option's delta-hedging, better know as \textit{market impact}. Let us consider a market impact rule obeying:
\begin{equation}
\label{market impact rule introduction}
S_* \rightarrow S_* + f(S_*, N_*),
\end{equation}
which means that the impact on the stock price of an order to buy $N_*$ stocks at a price $S_*$ is $f(S_*, N_*)$. From the point of view of the option's hedger, if the spot moves from $S$ to $S + \mathrm{d}S_0$, he will buy $\Gamma \mathrm{d}S_0$ stocks, but doing this, due to the market impact rule (\ref{market impact rule introduction}) the spot price will move to $S + \mathrm{d}S_0 + \mathrm{d}S_1$ where $\mathrm{d}S_1$ will be given by $\mathrm{d}S_1 = f(S + \mathrm{d}S_0, \Gamma \mathrm{d}S_0)$. But again to adjust the hedging, he has to buy $\Gamma \mathrm{d}S_1$ stocks which in turns impacts the spot price by $\mathrm{d}S_2 = f(S + \mathrm{d}S_0 + \mathrm{d}S_1, \Gamma \mathrm{d}S_1)$ and so on and so forth. This perturbative approach taking to the limit $n \rightarrow +\infty$, where $n$ represent the number of transactions of the re-hedging procedure, will move the spot:
\begin{equation}
\label{re-hedging procedure rule}
S \rightarrow{S + \displaystyle\sum_{n = 0}^{+\infty}{\mathrm{d}S_n}}.
\end{equation}
From a realistic trading perspective, a re-hedging procedure leading to the divergence of the \textit{market impact series} $\sum_{n \geq 0}{\mathrm{d}S_n}$ cannot be acceptable, therefore we will need to study some convergence properties of the market impact series derived from the market impact rule (\ref{market impact rule introduction}) described above. The convergence of the market impact series $\sum_{n \geq 0}{\mathrm{d}S_n}$ has several main physical consequences that seem necessary. First it implies that the market impact has only a limited effect on the stock price, secondly the incremental impact $\mathrm{d}S_n$ vanish where $n \rightarrow +\infty$. Those properties have already been observed empirically in the U.S. equity market \cite{bucci2019crossover} and in the European equity market \cite{said2018} highlighting the fact that at the end of certain large metaorders the incremental impacts vanish letting the total market impact reaching a plateau. In our framework, as we will only consider a perturbative approach, the initial exogenous spot move $\mathrm{d}S$ will be supposed small enough. In that case, the following re-hedging trades will be considered also small and we will show how this implies to consider at this scale \textit{local linear market impact} rule. Indeed it is well established that the response to individual small orders are linear and it is only the aggregation of those small orders executed consecutively, better known as \textit{metaorders}, which lead to concave market impact in agreement to a \textit{square-root law} in the equity markets \cite{moro2009market} \cite{bacry2015market} \cite{toth2011anomalous} but also in the options market \cite{toth2016square} \cite{said2019market}. The multi-timescale property of the market impact has also been recently adressed in \cite{benzaquen2018market}.

The rest of the paper is organized as follows. Section \ref{Order book modeling and local linear market impact} presents our order book assumptions and the way it leads in our perturbative approach to \textit{local linear market impact}. Section \ref{Context and main results} gives the main results presented in the paper. Section \ref{Market Impact and Hedging: A perturbation theory of the market impact} introduces the main definitions of the perturbation theory and gives the convergence results about the market impact series. Section \ref{Market Impact and Metaorders Execution} presents some results on the execution of metaorders in our setting. Section \ref{Market Impact and Option Pricing} presents the derivation of the pricing equation based on the results of the previous sections. Section \ref{The Market Impact of Hedging Metaorders} studies the shape and the relaxation of our \textit{hedging metaorders}. Section \ref{Conclusion} recalls and discusses our main results and makes some connections with related works.

\section{Order book modeling and local linear market impact}
\label{Order book modeling and local linear market impact}

Let us consider an order book parametrized by a mid price $\overline{S}$ and a supply intensity $\eta(t,s)$ such as the units of stocks available at the instant $t$ on the limit order book between $S$ and $S + \mathrm{d}S$ is equal to $\eta(t,S)\mathrm{d}S$. Of course the order book approach presented here is a trivial simplification of what it is really observed on the markets. More accurate descriptions are given for instance in \cite{biais1995empirical} and \cite{toth2012does}. For a more detailed order book modeling, see for example \cite{cont2010stochastic} and \cite{abergel2013mathematical}.

In our order book perspective, the execution of an amount $A$ (expressed in currency) will consume the order book up to $\overline{S} + \varepsilon(A)$, where $\varepsilon(A)$ is defined such as
$$ A = \int_{\overline{S}}^{\overline{S} + \varepsilon(A)}{s \eta(t,s)\,\mathrm{d}s}, $$
whereas the number $N(A)$ of stocks purchased to execute the amount $A$ must satisfy
$$ N(A) = \int_{\overline{S}}^{\overline{S} + \varepsilon(A)}{\eta(t,s)\,\mathrm{d}s}. $$
In agreement with the intuition, we expect that a small order has only a small impact on the spot price, hence we will assume that $\varepsilon(0) = 0$ and $\varepsilon$ is continuous in $0$. As we will consider only small orders, we will be interested in the behaviour of $\varepsilon$ around $0$ at the leading order in $A$. The function $\varepsilon$ as defined above allows us to capture the liquid character or not of an underlying. Indeed to illustrate this, let us consider two particular cases: The one of a super liquid stock with $A \gg \varepsilon(A)$ and the second one of a very poor liquid equity satisfying $A \sim \varepsilon(A)$ which reads that $A$ and $\varepsilon(A)$ are the same order of magnitude. Those considerations motivate to define the \textit{market depth} $L(t,\overline{S})$ at $\overline{S}$ by
$$ L(t,\overline{S}) := \lim\limits_{A \rightarrow 0}{\frac{A}{\varepsilon(A)}}. $$
We have
$$ \frac{A}{\varepsilon(A)} = \frac{1}{\varepsilon(A)}\int_{\overline{S}}^{\overline{S} + \varepsilon(A)}{s\eta(t,s)\,\mathrm{d}s} \xrightarrow[A \rightarrow 0]{} \eta(t,\overline{S}) \overline{S},$$
and
$$ \frac{N(A)}{\varepsilon(A)} = \frac{1}{\varepsilon(A)}\int_{\overline{S}}^{\overline{S} + \varepsilon(A)}{\eta(t,s)\,\mathrm{d}s} \xrightarrow[A \rightarrow 0]{} \eta(t,\overline{S}).$$
This gives that 
$$ \varepsilon(A) = \frac{\varepsilon(A)}{N(A)\overline{S}} \times N(A)\overline{S} \sim_{A \rightarrow 0} \frac{1}{\eta(t,\overline{S}) \overline{S}} \times N(A)\overline{S}. $$
By setting
\begin{equation}
\label{definition lambda}
\lambda(t,\overline{S}) := \frac{1}{L(t,\overline{S})} = \frac{1}{\eta(t,\overline{S}) \overline{S}},
\end{equation}
we have $\varepsilon(A) \sim_{A \rightarrow 0} \lambda(t,\overline{S}) N(A)\overline{S}$ corresponding to \textit{linear market impact} for $A$ small enough. Therefore we have established the following \textit{local linear market impact} rule:
\begin{equation}
\label{linear market impact rule}
S_* \rightarrow S_* + \lambda(t,S_*)S_* N_*,
\end{equation}
which is valid when the order size $N_*$ is small enough. When $\lambda(t,\overline{S}) \equiv \lambda \in \mathbb{R}_+$ we retrieve the market impact rule given in \cite{abergel2017option} where the impact is given by
\begin{equation}
\label{abergel market impact}
  S_*(e^{\lambda N_*} - 1) \sim_{N_* \rightarrow 0} \lambda S_* N_*.
\end{equation}
The case $\lambda(t,\overline{S}) \equiv \lambda \overline{S}$ with $\lambda \in \mathbb{R}_+$ corresponds to the market impact rule presented in \cite{loeper2018option}:
\begin{equation}
\label{loeper market impact}
  S_* \rightarrow S_* + \lambda S_*^2 N_*.
\end{equation}
In what follows we will consider indifferently (\ref{abergel market impact}) or (\ref{loeper market impact}) by introducing
\begin{equation}
\label{said market impact}
  S_* \rightarrow S_* + \lambda S_*^{1 + \zeta} N_*,
\end{equation}
where $\zeta \in \{0,1\}$. One must notice that depending of the choice of $\zeta$ the dimension of the parameter $\lambda$ can vary.

\section{Context and main results}
\label{Context and main results}

Let us suppose the market impact rule (\ref{said market impact}) hold, and consider an agent who is short of an european style option for instance. Taking in to account market impact, if the spot moves from $S$ to $S + \mathrm{d}S$ the agent is going to try to react to the exogenous market move $\mathrm{d}S$ by adjusting his hedge by purchasing $N$ stocks. This will result in a final state 
$$ S + \mathrm{d}S + (S + \mathrm{d}S)^{1 + \zeta}\lambda N, $$
and as the trader wants to be hedged at the end of the day, $N$ must satisfy the following equation
$$ \Gamma(t,S + \mathrm{d}S)(\mathrm{d}S + (S + \mathrm{d}S)^{1 + \zeta}\lambda N)  = N$$
leading to
$$ N = \frac{\Gamma(t, S + \mathrm{d}S) \mathrm{d}S}{1 - (S + \mathrm{d}S)^{1 + \zeta} \lambda\Gamma(t,S + \mathrm{d}S)}. $$
Let us assume that $x \longmapsto \displaystyle\frac{\Gamma(t, x)}{1 - \lambda x^{1 + \zeta} \Gamma(t,x)}$ has a Taylor series expansion in the neighbourhood of any $S$, hence $N$ can be read
\begin{align*}
N &= \frac{\Gamma(t, S + \mathrm{d}S)}{1 - (S + \mathrm{d}S)^{1 + \zeta} \lambda\Gamma(t,S + \mathrm{d}S)} \\
  &= \frac{\Gamma(t,S)\mathrm{d}S}{1 - \lambda S^{1 + \zeta} \Gamma(t,S)} + c_2(t,S) (\mathrm{d}S)^2 + c_3(t,S) (\mathrm{d}S)^3 + \dots \numberthis \label{N taylor expansion}
\end{align*}
giving $N \approx \displaystyle\frac{\Gamma(t,S)\mathrm{d}S}{1 - \lambda S^{1 + \zeta}\Gamma(t,S)}$ at the leading order in $\mathrm{d}S$. So everything happens as if the market impact has affected the dynamics of the spot in
\begin{equation}
\label{modified spot dynamics}
\mathrm{d}\tilde{S} = \frac{\mathrm{d}S}{1 - \lambda S^{1 + \zeta}\Gamma(t,S)}
\end{equation}
at the first order. This approach was developped in \cite{loeper2018option} by considering (\ref{modified spot dynamics}) as an \textit{ansatz} to derive his pricing equation. Modified spot dynamics generated by market impact has been also discussed in \cite{bouchard2016almost}. As $\mathrm{d}S$ and $\mathrm{d}\tilde{S}$ need to have the same sign, we have necessarily
\begin{equation}
\label{inequality condition}
\sup\limits_{(t,S) \in [0,T] \times \mathbb{R}_+}{\lambda S^{1 + \zeta}\Gamma(t,S)} < 1
\end{equation}
with $T > 0$ fixed. The condition (\ref{inequality condition}) and its variants are admitted in several papers on the topic (see \cite{liu2005option} \cite{abergel2017option} and \cite{loeper2018option} for instance). However this important question has always been left aside and often assumed in order to derive a pricing equation with exact replication of european style options. From a mathematical point of view it is always possible to replace the pricing equation $\mathcal{P} = 0$ by $\max(\mathcal{P}, \lambda S^{1 + \zeta}\Gamma(t,S) - 1) = 0$ as suggested in \cite{loeper2018option}.

At this stage, there is also an other point that needs to be discussed. Most papers on option pricing and hedging with market impact deal with linear impact as done in \cite{bouchard2016almost} \cite{bouchard2017hedging} and \cite{loeper2018option} for exemple. However, although this approach is acceptable for \textit{small} trades is clearly not realistic for \textit{large} orders in terms of market microstructure. Besides when it comes trades with sufficiently large size, it is not realistic to state that a large order can be executed in a single trade. Hence it becomes necessary to study the execution process of hedging.

The main objective of this paper is to show however that the option pricing approaches developped mainly in \cite{liu2005option} \cite{abergel2017option} and \cite{loeper2018option} are compatible with the market impact foundations mainly based on the study of the metaorders. The first step done in this direction has been presented in \cite{abergel2017option}. In their paper the authors integrate in their option pricing model a relaxation factor to illustrate the results obtained of the \textit{permanent impact} in the metaorders. Our paper is a contribution of this strand of research as our main goal is to provide some realistic connections between \textit{option pricing theory} and \textit{market impact empirical findings}. To this end our approach is mainly inspired by the original definition of a metaorder which is nothing less than a large order split into several small orders to be executed incrementally. Besides it has been shown in two companion papers (\cite{said2018} and \cite{said2019market}) that the metaorders can obey to an \textit{algorithmic definition}. This implies that they are not necessarily driven by \textit{execution strategies} but they are more often simply \textit{opportunistic} or \textit{hedging} trades. So as the vast majority of the meatorders executed on the markets appear to be hedging trades they must explain the market impact curves observed in several studies. We will introduce a hedging procedure -- as presented in Section \ref{introduction} based on a local linear market impact rule (\ref{linear market impact rule}) -- composed of successive small orders each of them impacting the price by $\mathrm{d}S_n$. Under a Gamma constant approximation we will show that the market impact series $\sum_{n \geq 0}{\mathrm{d}S_n}$ is convergent if and only if the condition (\ref{inequality condition}) holds. This result gives a physical meaning to the condition (\ref{inequality condition}) as it linked the convergence of the market impact series with the possibility to derive a pricing equation. We will also show that the sum of the market series $\sum_{n \geq 0}{\mathrm{d}S_n}$ can be expressed as in (\ref{N taylor expansion}). Hence what was considered as an \textit{ansatz} (\ref{modified spot dynamics}) in \cite{loeper2018option} is now simply a consequence of the convergence of the market impact series.

\section{Market Impact and Hedging: A perturbation theory of the market impact}
\label{Market Impact and Hedging: A perturbation theory of the market impact}

In this section we give the main results based on our perturbation theory of the market impact which will be useful to derive the pricing equation. For ease of notations, we will not consider any interest rates or dividends in the rest of the paper.

Let us consider that we have sold an european style option whose value is $u(t,s)$ with a fixed maturity $T$. Greeks are given as usual by
\begin{align*}
	\Delta(t,s) &= \partial_s{u}(t,s), \\
	\Gamma(t,s) &= \partial_{ss}{u}(t,s) = \partial_s{\Delta}(t,s), \\
	\Theta(t,s) &= \partial_t{u}(t,s).
\end{align*}
In order to derive a pricing equation, we will be interested in small enough spot moves. We have already established in Section \ref{Order book modeling and local linear market impact} that this leads to consider linear market impact rule in the framework of our perturbation theory. Therefore in what follows we will take a market impact which reads:
\begin{equation}
\label{main permanent market impact rule}
S_* \rightarrow S_* + \lambda S_*^{1 + \zeta} N_*
\end{equation}
i.e. the impact of an order to buy $N_*$ stocks at a price $S_*$ is $\lambda S_*^{1 + \zeta} N_*$ when the size $N_*$ of the order is sufficiently small (\textit{linear market impact}). We set the parameter $\phi$ defined by
\begin{equation}
  \phi \equiv \phi(t,S) := \lambda S^{1 + \zeta} \partial_{ss}{u}(t,S) = \lambda S^{1 + \zeta} \Gamma(t,S).
\end{equation}
We assume an initial spot moves from $S$ to $S + \mathrm{d}S$, $\mathrm{d}S$ supposed to be small and $S > 0$. By following an iterative hedging strategy one has to adjust the hedge by $\Gamma \mathrm{d}S$ stocks after the initial spot move $S \rightarrow S + \mathrm{d}S$, which then again impacts the spot price by $\mathrm{d}S_1 = \lambda (S + \mathrm{d}S_0)^{1 + \zeta} \Gamma \mathrm{d}S_0$ according to the linear market impact rule presented above with $\mathrm{d}S_0 = \mathrm{d}S$. This spot move is then followed by a second hedge adjustment of $\Gamma \mathrm{d}S_1$, which in turn impacts the spot price by $\mathrm{d}S_2 = \lambda (S + \mathrm{d}S_0 + \mathrm{d}S_1)^{1 + \zeta} \Gamma \mathrm{d}S_1$ and so on and so forth. Hence one has
$$\left\{
    \begin{array}{l}
    \mathrm{d}S_0 = \mathrm{d}S \\
    \forall n \in \mathbb{N}, \mathrm{d}S_{n+1} = \lambda \left(S + \displaystyle\sum_{k=0}^{n}{\mathrm{d}S_k}\right)^{1 + \zeta} \Gamma \mathrm{d}S_n = \phi \left(1 + \displaystyle\frac{1}{S}\sum_{k=0}^{n}{\mathrm{d}S_k}\right)^{1 + \zeta} \mathrm{d}S_n
    \end{array}
\right.$$
with the assumptions that $\Gamma \equiv \Gamma(t,S)$ remain constant during the hedging procedure described just above. The Gamma approximation appears also in \cite{almgren2016option} and considerably simplifies the problem by eliminating the dependence of the variable state $S_t$ in the expression of $\Gamma$. Hence the approximation that $\Gamma$ is constant during the hedging procedure allows us to exhibit the essential features of the local hedging without losing ourselves in complexities. Furthermore from the numerical point of view, this hypothesis is clearly justified in the context of the perturbative approach proposed here. In what follows we will show that for $\mathrm{d}S$ small enough, the total market impact $\displaystyle\sum_{n=0}^{+\infty}{\mathrm{d}S_n}$ will be also small enough. In the context of our approach as the market impact is only considered as a perturbation we have that
$$ \Gamma\left(t,S + \sum_{n=0}^{+\infty}{\mathrm{d}S_n}\right) \approx \Gamma(t,S) + \partial_s \Gamma(t,S) \sum_{n=0}^{+\infty}{\mathrm{d}S_n} \approx \Gamma(t,S). $$

Besides considering the fact that the price is bounded below by $0$ it is more relevant to write
\begin{equation}
\label{equation dS}
\left\{
    \begin{array}{l}
    \mathrm{d}S_0 = \mathrm{d}S  \vee (-S) \\
    \forall n \in \mathbb{N}, \mathrm{d}S_{n+1} = \phi \left(1 + \displaystyle\frac{1}{S}\sum_{k=0}^{n}{\mathrm{d}S_k}\right)^{1 + \zeta} \mathrm{d}S_n \vee \left(-\displaystyle\sum_{k=0}^{n}{\mathrm{d}S_k} -S \right)
    \end{array}
\right.
\end{equation}
with for all $a,b \in \mathbb{R}$, $\max(a,b) = a \vee b$. For all $n \in \mathbb{N}$, $\mathrm{d}S_n$ being a function of $\mathrm{d}S$, we will denote by now $x$ for the exogenous spot move $\mathrm{d}S$ and $u_n$ for the $n-th$ impact $\mathrm{d}S_n$. Hence we introduce for ease of notations the sequence of real valued functions $(u_n)_{n \in \mathbb{N}}$ defined by for all $x \in \mathbb{R}$,
\begin{equation}
\label{equation 1}
\left\{
    \begin{array}{l}
    u_0(x) = x \vee (-S)\\
    \forall n \in \mathbb{N}, u_{n+1}(x) = \phi \left(1 + \displaystyle\frac{s_n(x)}{S}\right)^{1 + \zeta} u_n(x) \vee (-s_n(x) - S)
    \end{array}
\right.
\end{equation}
where $s_n := \sum_{k=0}^{n}{u_k}$ denote the partial sums of the series $\sum_{n \in \mathbb{N}}{u_n}$. As $s_n$ represent the cumulative market impact after $n$ successive re-hedging trades, from a realistic trading perspective this quantity cannot explosed, hence we will add the following constraint:

\begin{ass}
\label{ass s_n(x) bounded}
For all $x \in \mathbb{R}$, the sequence $(s_n(x))_{n \in \mathbb{N}}$ is bounded.
\end{ass}

In the context of our perturbation theory of the market impact, it will be necessary to study the convergence properties of the market impact series derived from the local linear market impact rule (\ref{linear market impact rule}) given previously. To this end, we introduce what we will call a \textbf{market impact scenario}:

\begin{defi}
\label{defi market impact scenario}
A sequence $(u_n)_{n \in \mathbb{N}}$ defined by (\ref{equation 1}) is said to be a \textbf{market impact scenario} starting from $x \in \mathbb{R}$.
\end{defi}

A \textit{market impact scenario} can be seen as an extension of the definition given for a metaorder (see Definition \ref{defi metaorder} given in Section \ref{The Market Impact of Hedging Metaorders}). We have this first elementary result:

\begin{prop}
\label{prop uniformly bounded}
Let $(u_n)_{n \in \mathbb{N}}$ an $\phi-$market impact scenario. Then $(s_n)_{n \in \mathbb{N}}$ is uniformly bounded below by $-S$.
\end{prop}

\begin{proof}
Let $(u_n)_{n \in \mathbb{N}}$ an $\phi-$market impact scenario starting from $x \in \mathbb{R}$. We proceed by induction.
\begin{itemize}
    \item $s_0(x) = u_0(x) = x \vee (-S) \geq -S$.
    \item Let $n \in \mathbb{N}$ and suppose that $u_n(x) \geq -S$. We have $u_{n+1}(x) \geq -s_n(x) -S$ which gives $s_{n+1}(x) \geq -S$.
\end{itemize}
Hence we have shown that for all $x \in \mathbb{R}$, $n \in \mathbb{N}$, $s_n(x) \geq -S$.
\end{proof}

Among the market impact scenarios, we must distinguish between those that are acceptable from those that are not. The following definition is here to give an admissibility criterion:

\begin{defi}
\label{defi admissible}
A market impact scenario $(u_n)_{n \in \mathbb{N}}$ is said to be \textbf{admissible from a trading perspective} if there exists $R > 0$ such that for any $x \in (-R,R)$, $\sum_{n \in \mathbb{N}}{u_n(x)}$ converges.
\end{defi}

The following result will give an equivalent criterion on the parameter $\phi$ for a market impact scenario to be admissible. We have that:

\begin{thm}
\label{thm phi}
The $\phi-$market impact scenario $(u_n)_{n \in \mathbb{N}}$ is admissible from a trading perspective if, and only if, $\phi \in (-\infty, 1)$.
\end{thm}

\begin{proof}
see Appendix \ref{proof of thm phi}.
\end{proof}

The previous result establish a straightforward connection between the range values of the parameter $\phi$ and the convergence of the market impact series. Actually the proofs of Theorem \ref{thm phi} given state much more than that. We have in fact the following result:

\begin{thm}
\label{thm lambda gamma 2}
Let $(u_n)_{n \in \mathbb{N}}$ an $\phi-$market impact scenario. The following statements are equivalent:
\begin{enumerate}[(i)]
    \item The market impact scenario $(u_n)_{n \in \mathbb{N}}$ is admissible from a trading perspective.
    \item $\phi \in (-\infty,1)$.
    \item There exists $R > 0$ such that for any $x \in (-R,R)$, $\sum_{n \geq 0}{u_n(x)}$ converges.
    \item There exists $R > 0$ such that for any $x \in (-R,R)$, $\sum_{n \geq 0}{u_n(x)}$ converges absolutely.
\end{enumerate}
\end{thm}

\begin{proof}
This result is an immediate corollary of the proofs given for the Propositions \ref{prop admissible 1}, \ref{prop admissible 2}, \ref{prop admissible 3} and \ref{prop admissible 4}.
\end{proof}

From now on, we consider only market impact scenarios that are admissible from a trading perspective which is, according to Theorem \ref{thm phi}, equivalent to considering values of $\phi$ in $(-\infty,1)$. We give now an expression of the market impact series for admissible market impact scenarios, in particular one can notice that we retrieve the first order term of the taylor expansion \ref{N taylor expansion} introduced above.

\begin{prop}
\label{prop sum series}
Let $(u_n)_{n \in \mathbb{N}}$ an $\phi-$market impact scenario admissible from a trading perspective starting from $x \in \mathbb{R}$ and
\begin{equation}
\label{equation N}
    N(x) := \inf{\left\{n \in \mathbb{N} \,\bigg|\, s_n(x) = -S \right\}}
\end{equation}
with the convention $\inf{\emptyset} = +\infty$. Then there exists $R > 0$ such that for all $x \in (-R,R)$, $N(x) \in \mathbb{N}^* \cup \{+\infty\}$ and
\begin{equation}
\label{equation sum dS_n}
    \sum_{n=0}^{N(x)}{u_n(x)} = \frac{x}{1 - \phi} + \frac{1}{S}\frac{(1 + \zeta)\phi}{1 - \phi}\sum_{n=0}^{N(x)}{s_n(x)u_n(x)} + \frac{1}{S^2}\frac{\zeta\phi}{1 - \phi}\sum_{n=0}^{N(x)}{s_n^2(x)u_n(x)}.
\end{equation}
\end{prop}

\begin{proof}
see Appendix \ref{proof prop sum series}.
\end{proof}

Among the admissible market impact scenarios, three can be distinguished. Firstly the trivial scenario where nothing strictly happens which consists for the price in reaching $0$ after the first exogenous spot move. Namely this corresponds to $u_0(x) = -S$. In this case there is nothing much to say. The second is to reach $0$ after a finite number of re-hedging trades. In that case it is the market impact that brings the price stock to $0$. And finally the third one and more realistic one implying an infinite number of re-hedging trades and putting the stock price at a non-zero different level. The definition given just below presents the three possibilities:

\begin{defi}
\label{defi regular chaotic and null}
Let $(u_n)_{n \in \mathbb{N}}$ an $\phi-$market impact scenario admissible from a trading perspective starting from $x \in \mathbb{R}$ and $N(x)$ as defined in (\ref{equation N}). The market impact scenario $(u_n)_{n \in \mathbb{N}}$ is said to be:
\begin{itemize}
    \item \textbf{null} when $N(x) = 0$.
    \item \textbf{chaotic} when $N(x) \in \mathbb{N}^*$.
    \item \textbf{regular} when $N(x) = +\infty$.
\end{itemize}
\end{defi}

In the rest of the section, we will only consider chaotic or regular market impact scenarios. Especially we will show (see Proposition \ref{prop regular}) that for $\phi \in (-1,1)$, when the initial spot move is small enough, which is clearly the case of interest in our perturbative approach, it is always possible to consider only regular market impact scenarios. Particularly, it will be the case of our \textit{hedging metaorders} discussed later (see Section \ref{The Market Impact of Hedging Metaorders}).

\begin{prop}
\label{prop regular}
Let $(u_n)_{n \in \mathbb{N}}$ an $\phi-$market impact scenario admissible from a trading perspective. For any $\phi \in (-1,1)$, there exists $R > 0$, such that for all $x \in (-R,R)$, the $\phi-$market impact scenario  $(u_n)_{n \in \mathbb{N}}$ is regular.
\end{prop}

\begin{proof}
see Appendix \ref{proof prop regular}.
\end{proof}

The following proposition shows that in the case of regular market impact scenarios with $\phi \in (-1,1)$ the market impact series is nothing less than a power series. In particular we obtain a closed form for the second order term of the market impact series. We recall that the first linear order term has been already given in any case (chaotic or regular) in Proposition \ref{prop sum series}.

\begin{prop}
\label{prop regular power series}
Let $\phi \in (-1,1)$ and $(u_n)_{n \in \mathbb{N}}$ an $\phi-$market impact scenario admissible from a trading perspective. There exists $R > 0$, such that for all $x \in (-R,R)$, the $\phi-$market impact scenario  $(u_n)_{n \in \mathbb{N}}$ is regular. Hence the series function $\sum_{n \geq 0}{u_n}$ can be expressed as a power series on $(-R,R)$, is of class $\mathcal{C}^{\infty}$ on $(-R,R)$ and the first two terms of the decomposition of $\sum_{n \geq 0}{u_n}$ as a power series are given by 
$$ \sum_{n=0}^{+\infty}{u_n(x)} = \displaystyle\frac{1}{1 - \phi} x + \displaystyle\frac{1}{S}\displaystyle\frac{(1 + \zeta) \phi}{(1 - \phi)^3 (1 + \phi)} x^2 + o(x^2) \text{ as } x \rightarrow 0. $$
\end{prop}

\begin{proof}
see Appendix \ref{proof prop regular power series}.
\end{proof}

\section{Market Impact and Metaorders Execution}
\label{Market Impact and Metaorders Execution}

\subsection{\texorpdfstring{$\lambda(t,S) \equiv \lambda$}{zeta = 0}}
\label{zeta = 0}

Let us consider the market impact rule (\ref{main permanent market impact rule}) when $\zeta = 0$
\begin{equation}
\label{main permanent market impact rule zeta = 0}
S_* \rightarrow S_* + \lambda S_* N_*.
\end{equation}
Assume an agent wants to execute incrementally an order of size $N$ with $\mathcal{K} \in \mathbb{N}^*$ child orders of size $n_1, n_2, \dots, n_{\mathcal{K}}$ satisfying
$$ \sum_{k=1}^{\mathcal{K}}{n_k} = N. $$
Without loss of generality we will suppose that $N \in \mathbb{R}_+^*$ and $n_1,\dots,n_{\mathcal{K}} \in \mathbb{R}_+^*$ -- i.e. a buy order, the same holds for a sell order -- such that
\begin{equation}
\label{condition on n_k}
\lim\limits_{\mathcal{K} \rightarrow +\infty}\sup\limits_{1 \leq k \leq \mathcal{K}}{|n_k|} = 0.
\end{equation}
The condition (\ref{condition on n_k}) is needed to ensure that (\ref{main permanent market impact rule zeta = 0}) can be applied to $n_1,\dots,n_{\mathcal{K}}$ for $\mathcal{K}$ large enough. Applying this when $\mathcal{K} = 2$ leads to
$$ S \xrightarrow[]{n_1} S + \lambda S n_1 \xrightarrow[]{n_2} S + \lambda S n_1 + \lambda ( S + \lambda S n_1) n_2, $$
which can be written
$$  S \xrightarrow[]{n_1} S(1 + \lambda n_1) \xrightarrow[]{n_2} S(1 + \lambda (n_1 + n_2) + \lambda^2 n_1 n_2). $$
Let us denote $\xrightarrow[]{n_1, \dots, n_{\mathcal{K}}}$ the contraction of $\xrightarrow[]{n_1} \dots \xrightarrow[]{n_{\mathcal{K}}}$. By a straightforward induction we have for all $\mathcal{K} \in \mathbb{N}^*$,
$$ S \xrightarrow[]{n_1, \dots, n_{\mathcal{K}}} S\left(1 + \sum_{k=1}^{\mathcal{K}}{\lambda^k \sum_{1 \leq i_1 < i_2 < \dots < i_k \leq \mathcal{K}}{n_{i_1}n_{i_2} \dots n_{i_k}}}\right). $$

\begin{defi}
\label{definition discrete execution strategy}
Let $N \in \mathbb{R}_+^*$ and $\mathcal{K} \in \mathbb{N}^*$. A finite sequence $(n_1,\dots,n_{\mathcal{K}}) \in (\mathbb{R}_+^*)^{\mathcal{K}}$ such that $\displaystyle\sum_{k=1}^{\mathcal{K}}{n_k} = N$ is said to be an $(N,\mathcal{K})-$\textbf{execution strategy}. An $(N,\mathcal{K})-$execution strategy is considered \textbf{admissible} if it satisfies (\ref{condition on n_k}).
\end{defi}

\begin{defi}
Let $(n_1,\dots,n_{\mathcal{K}})$ an $(N,\mathcal{K})-$execution strategy. Let us denote by $$ S_{N,k}(n_1,\dots,n_{\mathcal{K}}) :=  S\left(1 + \sum_{i=1}^{k}{\lambda^i \sum_{1 \leq j_1 < j_2 < \dots < j_i \leq k}{n_{j_1}n_{j_2} \dots n_{j_i}}}\right) $$ 
the $k-th$ price of this $(N,\mathcal{K})-$execution strategy for $k \in \llbracket 0,\mathcal{K} \rrbracket$. Then let us set its \textbf{market impact} by
$$ I_{N,\mathcal{K}}(n_1,\dots,n_{\mathcal{K}}) := S_{N,\mathcal{K}}(n_1,\dots,n_{\mathcal{K}}) - S $$
and its \textbf{average execution price} by
$$ \overline{S}_{N,\mathcal{K}}(n_1,\dots,n_{\mathcal{K}}) := \frac{1}{N}\sum_{k=1}^{\mathcal{K}}{n_k S_{N,k-1}(n_1,\dots,n_{\mathcal{K}})}. $$
\end{defi}

\begin{prop}
\label{invariance market impact under permutation}
Let $N \in \mathbb{R}_+^*$ and $\mathcal{K} \in \mathbb{N}^*$. The market impact of an $(N,\mathcal{K})-$execution strategy depends only on the sizes of the $\mathcal{K}$ child orders. The order in which the child orders are executed does not affect the final value of
$I_{N,\mathcal{K}}(n_1,\dots,n_{\mathcal{K}})$.
\end{prop}

\begin{proof}
Let $(n_1,\dots,n_\mathcal{K})$ an $(N,\mathcal{K})-$exection strategy and $\sigma \in \mathcal{S}_{\mathcal{K}}$ the set of the permutations of $\llbracket 1,\mathcal{K}\rrbracket$. The result is a straightforward consequence of the fact that for any permutation $\sigma \in \mathcal{S}_{\mathcal{K}}$, $I_{N,\mathcal{K}}\left(n_{\sigma(1)},\dots,n_{\sigma(\mathcal{K})}\right) = I_{N,\mathcal{K}}(n_1,\dots,n_{\mathcal{K}}).$
\end{proof}

Proposition \ref{invariance market impact under permutation} has shown that for a given $(N,\mathcal{K})-$execution strategy any permutation of $(n_1,\dots,n_{\mathcal{K}})$ is similar in terms of market impact. The point underlined by Proposition \ref{invariance market impact under permutation} is in fact a well known property of market impact which appears in the litterature through the \textit{square root formula} or other related formulas. The square root formula expresses the final market impact of a metaorder as a function of its size but it doest not tell much about the dynamics itself of the metaorder.


The following result (Theorem \ref{NK execution strategies thm - inequalities}) shows that among all the $(N,\mathcal{K})-$execution strategies the one which impacts the most the market is the most predictable one. Indeed given a metaorder size $N$ and $\mathcal{K}$ child orders the first execution strategy that comes in mind is to split the metaorder in lots of equal size. But doing that could be considerably increase the probability to be detected by the market makers especially when $\mathcal{K}$ becomes large enough. Hence it seems normal that the most expensive strategy in terms of market impact and average execution price corresponds to the most obvious and the less clever.

\begin{thm}
\label{NK execution strategies thm - inequalities}
Let $N \in \mathbb{R}_+^*$ and $\mathcal{K} \in \mathbb{N}^*$. The market impact and the average execution price of an $(N,\mathcal{K})-$execution strategy are bounded and reach their upper bound if, and only if the strategy is equally-sized i.e. $n_1 = \dots = n_{\mathcal{K}} = \displaystyle\frac{N}{\mathcal{K}}$. Besides the following inequalities hold
\begin{equation}
\label{NK execution strategies thm inequality price}
S + \lambda S N \leq S_{N,\mathcal{K}}(n_1,\dots,n_{\mathcal{K}}) \leq S e^{\lambda N}
\end{equation}
and
\begin{equation}
\label{NK execution strategies thm inequality average execution price}
S + \frac{1}{2}\lambda S \left(N - \sup\limits_{1 \leq k \leq \mathcal{K}}{n_k}\right)  \leq \overline{S}_{N,\mathcal{K}}(n_1,\dots,n_{\mathcal{K}}) \leq S \frac{e^{\lambda N} - 1}{\lambda N}.
\end{equation}
\end{thm}

\begin{proof}
see Appendix \ref{proof NK execution strategies thm - inequalities}.
\end{proof}

The lower bound of (\ref{NK execution strategies thm inequality price}) corresponds to the linear market impact model presented in \cite{loeper2018option} and the upper bound to the one presented in \cite{abergel2017option}. One can notice that these two bounds do not depend of the execution strategy chosen -- as they are independant of the sizes and the number $\mathcal{K}$ of the child orders -- but only of the final size $N$ of the order. This has the advantage of avoiding arbitrage opportunities, as well as not being sensitive to the hedging frequency. One can notice that $ S_{N, \mathcal{K}}(n_1,\dots,n_{\mathcal{K}}) \approx S + \gamma\lambda N S$ when $N \rightarrow 0$ in agreement with (\ref{main permanent market impact rule zeta = 0}). Hence the lower bound of (\ref{NK execution strategies thm inequality price}) can be reached when $N$ is small enough. More interesting, regarding metaorders with a size $N$ much larger, we will show that the upper bound of (\ref{NK execution strategies thm inequality price}) corresponds to the case when the number of child orders $\mathcal{K}$ goes to infinity. This result seems pretty obvious from a market making perspective. Indeed when $\mathcal{K} \rightarrow +\infty$ the probability that the market makers detect the metaorder is going to converge towards 1 and the market impact is going to attain its peak value.

\begin{thm}
\label{NK execution strategies thm - limits}
Let $(n_1,\dots,n_{\mathcal{K}})$ an admissible $(N,\mathcal{K})-$execution strategy i.e. such that
$$ \lim\limits_{\mathcal{K} \rightarrow +\infty}\sup\limits_{1 \leq k \leq \mathcal{K}}{|n_k|} = 0. $$
Then we have that
\begin{equation}
\label{price impact max}
\lim\limits_{\mathcal{K} \rightarrow +\infty}{S_{N,\mathcal{K}}(n_1,\dots,n_{\mathcal{K}})} = S e^{\lambda N}
\end{equation}
and
\begin{equation}
\label{average execution price impact max}
\lim\limits_{\mathcal{K} \rightarrow +\infty}{\overline{S}_{N,\mathcal{K}}(n_1,\dots,n_{\mathcal{K}})} = S \frac{e^{\lambda N} - 1}{\lambda N}.
\end{equation}

Hence when $\mathcal{K}$ goes to infinity the market impact attains its peak value. Furthermore the market impact and the average execution price do not depend longer of the execution strategy chosen. This implies that arbitrage opportunities vanish when the trading frequency is going to infinity.
\end{thm}

\begin{proof}
see Appendix \ref{proof NK execution strategies thm - limits}.
\end{proof}

\begin{rem}
\label{remark N infty market impact thm}
Theorem \ref{NK execution strategies thm - limits} is actually true for $n_1,\dots,n_{\mathcal{K}} \in \mathbb{R}$ not necessarily positive by weakening (\ref{condition on n_k}) and assuming that for some $\varepsilon > 0$, $\sup\limits_{1 \leq k \leq \mathcal{K}}{|n_k|} = O\left(\left(\displaystyle\frac{1}{{\mathcal{K}}}\right)^{1/2 + \varepsilon}\right).$
\end{rem}

Remark \ref{remark N infty market impact thm} states in fact that when $\mathcal{K} \rightarrow +\infty$ there does not exist any trading strategy that could affect the average execution price or the final market impact. In other words it is not possible to manipulate the price formation process by executing \textit{round trip trades} as defined and studied in \cite{gatheral2010no} and \cite{alfonsi2010optimal} for instance. Theorem \ref{NK execution strategies thm - limits} motivates to extent market impact and average execution price definitions for continuous execution strategies by: 

\begin{defi}
Let $N \in \mathbb{R}$. A function $F : [0,1] \rightarrow \mathbb{R}$ of bounded variation and continuous such that $F(0) = 0$ and $F(1) = N$ is said to be an $(N,+\infty)-$\textbf{execution strategy}. Any $(0,+\infty)-$execution strategy is said to be a \textbf{round trip trade}.
\end{defi}

\begin{defi}
Let $N \in \mathbb{R}$. The \textbf{market impact} of an $(N,+\infty)-$execution strategy is given by
\begin{equation}
\label{definition market impact N infty}
I_{N,+\infty} := S e^{\lambda \int_{0}^{1}{\mathrm{d}F(s)}} - S = S e^{\lambda N} - S
\end{equation}
and the \textbf{average execution price} by
\begin{equation}
\label{definition average execution price N infty}
\overline{S}_{N,+\infty} := \displaystyle\frac{\int_{0}^{1}{S e^{\lambda F(s)}\,\mathrm{d}F(s)}}{\int_{0}^{1}{\lambda\mathrm{d}F(s)}} = S \frac{e^{\lambda N} - 1}{\lambda N}.
\end{equation}
\end{defi}

One can see in (\ref{definition market impact N infty})-(\ref{definition average execution price N infty}) how the choice of the continuous execution strategy $F$ does not affect the market impact and the average execution price in agreement with Theorem \ref{NK execution strategies thm - limits}.

\subsection{\texorpdfstring{$\lambda(t,S) \equiv \lambda S$}{zeta = 1}}
\label{zeta = 1}

Let us consider now the market impact rule (\ref{main permanent market impact rule}) when $\zeta = 1$
\begin{equation}
\label{main permanent market impact rule zeta = 1}
S_* \rightarrow S_* + \lambda S_*^2 N_*.
\end{equation}
Let us show that under the condition (\ref{condition on n_k}) all the results of this case can be derived from the previous one by considering the following approximation 
$$ (1+x)^2 \approx 1 + 2x, $$
when $x$ is small enough. When $\mathcal{K} = 2$ we have
$$  S \xrightarrow[]{n_1} S + \lambda S^2 n_1 \xrightarrow[]{n_2} S + \lambda S^2 n_1 + \underbrace{\lambda S^2 (1+ \lambda S n_1)^2 n_2}_{\approx \lambda S^2 (1+ 2\lambda S n_1) n_2} $$
leading approximately to
$$  S \xrightarrow[]{n_1} S + \lambda S^2 n_1 \xrightarrow[]{n_2} S \left(1 + \frac{1}{2} (2\lambda S) (n_1 + n_2) + \frac{1}{2}  (2\lambda S)^2 n_1 n_2 \right). $$
A straightforward induction gives for all $\mathcal{K} \in \mathbb{N}^*$,
$$ S \xrightarrow[]{n_1, \dots, n_{\mathcal{K}}} S\left(1 + \frac{1}{2} \sum_{k=1}^{\mathcal{K}}{(2\lambda S)^k \sum_{1 \leq i_1 < i_2 < \dots < i_k \leq \mathcal{K}}{n_{i_1}n_{i_2} \dots n_{i_k}}}\right). $$
Hence all the results derived in the previous section still hold. Particularly Theorems \ref{NK execution strategies thm - inequalities} and \ref{NK execution strategies thm - limits} can now be written

\begin{thm}
\label{NK execution strategies thm - inequalities - zeta = 1}
Let $N \in \mathbb{R}_+^*$ and $\mathcal{K} \in \mathbb{N}^*$. The market impact and the average execution price of an $(N,\mathcal{K})-$execution strategy are bounded and reach their upper bound if, and only if the strategy is equally-sized i.e. $n_1 = \dots = n_{\mathcal{K}} = \displaystyle\frac{N}{\mathcal{K}}$. Besides the following inequalities hold
\begin{equation}
\label{NK execution strategies thm inequality price - zeta = 1}
S + \lambda S^2 N \leq S_{N,\mathcal{K}}(n_1,\dots,n_{\mathcal{K}}) \leq \frac{S}{2}(1 + e^{2\lambda S N})
\end{equation}
and
\begin{equation}
\label{NK execution strategies thm inequality average execution price - zeta = 1}
S + \frac{1}{2}\lambda S^2 \left(N - \sup\limits_{1 \leq k \leq \mathcal{K}}{n_k}\right)  \leq \overline{S}_{N,\mathcal{K}}(n_1,\dots,n_{\mathcal{K}}) \leq \frac{S}{2} \left(1 + \frac{e^{2\lambda S N} - 1}{2\lambda S N}\right).
\end{equation}
\end{thm}

\begin{thm}
\label{NK execution strategies thm - limits - zeta = 1}
Let $(n_1,\dots,n_{\mathcal{K}})$ an admissible $(N,\mathcal{K})-$execution strategy i.e. such that
$$ \lim\limits_{\mathcal{K} \rightarrow +\infty}\sup\limits_{1 \leq k \leq \mathcal{K}}{|n_k|} = 0. $$
Then we have that
\begin{equation}
\label{price impact max - zeta = 1}
\lim\limits_{\mathcal{K} \rightarrow +\infty}{S_{N,\mathcal{K}}(n_1,\dots,n_{\mathcal{K}})} = \frac{S}{2}(1 + e^{2\lambda S N})
\end{equation}
and
\begin{equation}
\label{average execution price impact max - zeta = 1}
\lim\limits_{\mathcal{K} \rightarrow +\infty}{\overline{S}_{N,\mathcal{K}}(n_1,\dots,n_{\mathcal{K}})} = \frac{S}{2} \left(1 + \frac{e^{2\lambda S N} - 1}{2\lambda S N}\right).
\end{equation}
\end{thm}

\section{Market Impact and Option Pricing}
\label{Market Impact and Option Pricing}

\subsection{Pricing equation}
\label{pricing equation}

We consider the framework of covered options, hence we start from a delta-hedged portfolio. Let us assume that the stock price $S$ moves initially by $\mathrm{d}S$ such that $(\mathrm{d}S_n)_{n \in \mathbb{N}}$ is a \textbf{regular} $\phi-$market impact scenario given by (\ref{equation dS}) with
\begin{equation}
\label{equation dS dynamics}
\mathrm{d}S = S(\nu \mathrm{d}t + \sigma \mathrm{d}W_t),
\end{equation}
where $(W_t)_{t \geq 0}$ is a Brownian motion. By following the hedging strategy presented above we move the spot from to $S$ to $S + \tilde{\mathrm{d}S}$ where $\tilde{\mathrm{d}S}$ represents the cumulative market impact at the end of the re-hedging procedure
\begin{equation}
\label{def dS tilde}
  \tilde{\mathrm{d}S} := \sum_{n=0}^{+\infty}{\mathrm{d}S_n},
\end{equation}
hence from (\ref{equation sum dS_n}) we have
\begin{equation}
\label{cumulative market impact}
\tilde{\mathrm{d}S} = \frac{\mathrm{d}S}{1 - \phi} + \frac{1}{S}\frac{(1 + \zeta) \phi}{1 - \phi}\sum_{n=0}^{+\infty}{\left(\sum_{k=0}^{n}{\mathrm{d}S_k}\right) \mathrm{d}S_n} + \frac{1}{S^2}\frac{\zeta \phi}{1 - \phi}\sum_{n=0}^{+\infty}{\left(\sum_{k=0}^{n}{\mathrm{d}S_k}\right)^2 \mathrm{d}S_n},
\end{equation}
which can be rewritten 
$$ \tilde{\mathrm{d}S} = \frac{\mathrm{d}S}{1 - \phi} + \frac{\mu}{S}(\mathrm{d}S)^2 + o((\mathrm{d}S)^2), $$
where $\mu$ is only a function of the parameter $\phi$. Plugging equation (\ref{equation dS dynamics}) in the previous equation leads to
$$ \tilde{\mathrm{d}S} = \frac{\mathrm{d}S}{1 - \phi} + \mu\sigma^2 S \mathrm{d}t + o(\mathrm{d}t), $$
as $(\mathrm{d}S)^2 = \sigma^2 S^2 \mathrm{d}t + o(\mathrm{d}t)$. This gives at the leading order
\begin{equation}
\label{equation dS tilde}
\tilde{\mathrm{d}S} = \frac{\mathrm{d}S}{1 - \phi} + \tilde{\nu} S \mathrm{d}t + o(\mathrm{d}t).
\end{equation}
The value $V$ of the hedging portfolio containing $\Delta(t,S)$ stocks at $t$ evolves as
$$ \mathrm{d}V = \Delta \tilde{\mathrm{d}S} + \mathcal{R}$$
as $S$ moves to $S + \tilde{\mathrm{d}S}$ with
\begin{align*}
\mathcal{R} &= N \times (\text{Final price of the stocks bought - Average execution price}) \\
            &= N \times \left(S + \displaystyle\sum_{n=0}^{+\infty}{\mathrm{d}S_n} -\displaystyle\frac{\displaystyle\sum_{n=0}^{+\infty}{\Gamma\mathrm{d}S_n \left(S + \displaystyle\sum_{k=0}^{n}{\mathrm{d}S_k}\right)}}{\displaystyle\sum_{n=0}^{+\infty}{\Gamma\mathrm{d}S_n}}\right) \\
            &= N \times \left(\displaystyle\sum_{n=0}^{+\infty}{\mathrm{d}S_n} -\displaystyle\frac{\displaystyle\sum_{n=0}^{+\infty}{\Gamma \mathrm{d}S_n \left(\displaystyle\sum_{k=0}^{n}{\mathrm{d}S_k}\right)}}{\displaystyle\sum_{n=0}^{+\infty}{\Gamma \mathrm{d}S_n}}\right)
\end{align*}
$N = \Gamma \tilde{\mathrm{d}S}$ being the number of stocks bought during the re-hedging procedure. As $\mathrm{d}S \rightarrow 0$, we have that $\sup\limits_{n \in \mathbb{N}}{|\mathrm{d}S_n|} \rightarrow 0$. Besides the proof given of Theorem \ref{thm phi} shows that for $n$ large enough there exists $r \in (0,1)$ such that $|\mathrm{d}S_{n+1}| < r |\mathrm{d}S_n|$, hence we can apply Remark \ref{remark N infty market impact thm} and use the results stated in Theorem \ref{NK execution strategies thm - limits}. When $\lambda(t,S) \equiv \lambda$ we have that
$$ \text{Final price of the stocks bought} = S + S (e^{\lambda N} - 1) $$
and
$$ \text{Average execution price} = S + S \left(\frac{e^{\lambda N} - 1}{\lambda N} - 1\right), $$
giving at the leading order
$$ \mathcal{R} =  \frac{1}{2} \lambda S N^2.$$
The same can be done when $\lambda(t,S) \equiv \lambda S$ and gives
$$ \text{Final price of the stocks bought} = S + \frac{S}{2}(e^{2\lambda S N} - 1) $$
and
$$ \text{Average execution price} = S + \frac{S}{2} \left(\frac{e^{2\lambda S N} - 1}{2\lambda S N} - 1\right), $$
impliying also at the leading order
$$ \mathcal{R} =  \frac{1}{2}\lambda S^2 N^2.$$
Finally one can write in any case
$$ \mathcal{R} = \frac{1}{2} \lambda N^2 S^{1 + \zeta} $$
with $N = \Gamma \tilde{\mathrm{d}S}$. This gives
\begin{equation}
\label{equation dV}
\mathrm{d}V = \Delta \tilde{\mathrm{d}S} + \frac{1}{2} \lambda S^{1 + \zeta} (\Gamma \tilde{\mathrm{d}S})^2.
\end{equation}
Besides assuming that the option is sold at its fair price, we have $\mathrm{d}V = \mathrm{d}u$ with at the leading order, for $S$ moving to $S + \tilde{\mathrm{d}S}$,
\begin{equation}
\label{equation du}
\mathrm{d}u = \partial_t{u} \,\mathrm{d}t + \partial_s{u} \,\tilde{\mathrm{d}S} + \frac{1}{2}\partial_{ss}{u} \,(\tilde{\mathrm{d}S})^2 + o(\mathrm{d}t).
\end{equation}
Therefore from equations (\ref{equation dV}) and (\ref{equation du}) we get
\begin{equation}
\label{equation dV = du}
\partial_t{u}\,\mathrm{d}t + \frac{1}{2}(\tilde{\mathrm{d}S})^2[\partial_{ss}u - \lambda S^{1 + \zeta}(\partial_{ss}u)^2] = o(\mathrm{d}t).
\end{equation}
If we plugg the spot dynamics (\ref{equation dS dynamics}) in the expression of $\tilde{\mathrm{d}S}$ (\ref{equation dS tilde}), we obtain at the leading order
\begin{equation}
\label{equation dS tilde square}
(\tilde{\mathrm{d}S})^2 = \frac{\sigma^2 S^2}{(1 - \phi)^2} \,\mathrm{d}t + o(\mathrm{d}t).
\end{equation}
By reinjecting (\ref{equation dS tilde square}) in (\ref{equation dV = du}), we obtain the pricing equation
\begin{align*}
\partial_t{u} + \frac{1}{2}\sigma^2 s^2 \partial_{ss}u \frac{1}{1 - \phi} &= 0, \numberthis \label{pricing equation} \\
\phi &= \lambda s^{1 + \zeta} \partial_{ss}u.
\end{align*}

\subsection{Perfect replication}

Any contingent claim can be perfectly replicated at zero cost, as long as one can exhibit a solution to the pricing equation (\ref{pricing equation}). This implies that we need to study the parabolicity of the given pricing equation. Furthermore we have established that only market impact scenarios admissible from a trading perspective have to be considered which leads according to Theorem \ref{thm phi} to the constraint
\begin{equation}
\label{constraint pricing equation}
\sup\limits_{(t, s) \in [0,T] \times \mathbb{R}_+}{\lambda s^{1 + \zeta} \partial_{ss}u} < 1,
\end{equation}
with $T$ the maturity of the option. In fact we have the following result

\begin{thm}
\label{parabolicity pricing equation}

There holds the two following propositions:
\begin{enumerate}[(i)]
\item The non-linear backward partial differential operator
$$ \partial_t {\cdot} + \frac{1}{2}\sigma^2 s^2  \frac{\partial_{ss}\cdot}{1 - \lambda s^{1 + \zeta} \partial_{ss} \cdot} $$
is parabolic.
\item Every european style contingent claim with payoff $\Phi$ satisfying the terminal constraint
$$  \sup\limits_{s \in \mathbb{R}_+}{\lambda s^{1 + \zeta} \partial_{ss}\Phi} < 1 $$
can be perfectly replicated via a $\delta-$hedging strategy given by the unique, smooth away from the maturity $T$, solution to equation (\ref{pricing equation}).
\end{enumerate}
\end{thm}

\begin{proof}
The parabolic nature of the operator is determined by the monotonicity of the function
$$ F : p \in (-\infty, 1) \longmapsto \frac{p}{1 - p}. $$
$F$ is globally increasing, therefore the pricing equation is globally well-posed. Besides, given that the payoff satisfies the terminal constraint, classical results on the maximum principle for the second derivative of the solution ensure that the same constraint is satisfied globally for $t \leq T$ which consitutes in fact the constraint (\ref{constraint pricing equation}). Hence, the perfect replication is possible. Results on maximum principle for the second derivative can be found in \cite{wang1992regularityI} and \cite{wang1992regularityII} for instance. A proof of this result in a more general setting is also given in \cite{abergel2017option}.
\end{proof}

\subsection{SDE formulation}
\label{SDE formulation}

\subsubsection{The system of SDE}
\label{The system of SDE}

Let $(\Omega, \mathcal{F}, (\mathcal{F}_t)_{t \geq 0}, \mathbb{P})$ a filtered probability space supporting an $(\mathcal{F}_t)_{t \geq 0}$ standard Brownian motion $(W_t)_{t \geq 0}$. Considering the result established in Theorem \ref{thm phi}, we will assume the following uniform condition:
\begin{ass}
\label{ass lambda gamma uniformly bounded}
There exists a constant $\varepsilon > 0$, such that for all $t \geq 0$, $\mathbb{P}-$a.s.,
\begin{equation}
\label{equation lambda gamma uniformly bounded}
1 - \lambda(t,S_t) S_t \Gamma_t \geq \varepsilon.
\end{equation}
\end{ass}

Let us consider then the following system of stochastic differential equations:
\begin{align*}
    \mathrm{d}\delta_t &= a_t \mathrm{d}t + \Gamma_t \mathrm{d}S_t, \numberthis \label{delta dynamics} \\
    \frac{\mathrm{d}S_t}{S_t} &= \sigma(t,S_t) \mathrm{d}W_t + \left(\nu(t,S_t) +\frac{\sigma^2(t,S_t) S_t \Gamma_t \big(\lambda(t,S_t) + \partial_x\lambda(t,S_t)\big)}{1 - \lambda(t,S_t) S_t \Gamma_t}\right)\mathrm{d}t  + \lambda(t,S_t) \mathrm{d}\delta_t, \numberthis \label{spot dynamics} \\
    \mathrm{d}V_t &= \delta_t \mathrm{d}S_t + \frac{1}{2} \lambda(t,S_t) S_t \mathrm{d}\langle \delta,\delta \rangle_t. \numberthis\label{hedging portfolio dynamics}
\end{align*}
where $\delta$, $S$ and $V$ are three processes starting respectively from $\delta_0$, $S_0$ and $V_0$ at $t = 0$.
We set
$$ \alpha_t \equiv \alpha(t,S_t,\Gamma_t) := \frac{\tilde{\sigma}(t,S_t) \Gamma_t}{1 - \tilde{\lambda}(t,S_t) \Gamma_t}, $$
and
$$ \beta_t \equiv \beta(t,S_t,\Gamma_t) := \frac{a_t}{1 - \tilde{\lambda}(t,S_t) \Gamma_t} + \frac{\Gamma_t}{1 - \tilde{\lambda}(t,S_t) \Gamma_t}\left(\tilde{\nu}(t,S_t) + \frac{\tilde{\sigma}^2(t,S_t) \Gamma_t \partial_x \tilde{\lambda}(t,S_t)}{1 - \tilde{\lambda}(t,S_t) \Gamma_t}\right), $$
by introducing also for all $t \in [0,T], x \in \mathbb{R}_+$,
\begin{align*}
\tilde{\sigma}(t,x) &:= x \sigma(t,x), \\
\tilde{\nu}(t,x) &:= x \nu(t,x),
\end{align*}
and
$$ \tilde{\lambda}(t,x) := x \lambda(t,x). $$
Hence the system of stochastic differential equations (\ref{delta dynamics})-(\ref{spot dynamics})-(\ref{hedging portfolio dynamics}) can be rewritten
\begin{align*}
    \mathrm{d}\delta_t &= \alpha_t \mathrm{d}W_t + \beta_t \mathrm{d}t , \numberthis \label{delta dynamics 2}\\
    \mathrm{d}S_t &= \tilde\sigma(t,S_t) \mathrm{d}W_t + \big(\tilde{\nu}(t,S_t) + \alpha_t (\tilde{\sigma}\partial_x \tilde{\lambda})(t,S_t)\big) \mathrm{d}t + \tilde{\lambda}(t,S_t) \mathrm{d}\delta_t, \numberthis \label{spot dynamics 2}\\
    \mathrm{d}V_t &= \delta_t \mathrm{d}S_t + \frac{1}{2} \tilde{\lambda}(t,S_t) \alpha_t^2 \mathrm{d}t. \numberthis \label{hedging portfolio dynamics 2}
\end{align*}

\begin{thm}
\label{existence and uniqueness of a strong solution}
Suppose the following assumptions hold:
\begin{itemize}
  \item $(a_t)_{t \geq 0}$ and $(\Gamma_t)_{t \geq 0}$ are uniformly bounded adapted processes
  \item $(\Gamma_t)_{t \geq 0}$ satisfies uniformly the condition (\ref{equation lambda gamma uniformly bounded})
  \item there exists a constant $K > 0$, such that for all $(t,x) \in [0,T] \times \mathbb{R}$,
  $$ |x \sigma(t,x)| + |x \nu(t,x)| \leq K $$
  \item there exists a constant $K > 0$, such that for all $(t,x,y) \in [0,T] \times \mathbb{R} \times \mathbb{R}$,
  $$ |x \sigma(t,x) - y \sigma(t,y)| + |x \nu(t,x) - y \nu(t,y)| \leq K|x-y| $$
  \item $(t,x) \in [0,T] \times \mathbb{R} \longmapsto \lambda(t,x)$ is of class $\mathcal{C}^{1,2}$ and all its partial derivatives are bounded.
  \item $(t,x) \in [0,T] \times \mathbb{R} \longmapsto x \lambda(t,x)$ has all its partial derivatives bounded.
\end{itemize}
Then there exists a unique strong solution to the system of stochastic differential equations (\ref{delta dynamics})-(\ref{spot dynamics})-(\ref{hedging portfolio dynamics}) starting from $(\delta_0, S_0, V_0) \in \mathbb{R}^3$ at $t=0$.
\end{thm}

\begin{proof}
By plugging (\ref{delta dynamics 2}) in (\ref{spot dynamics 2}), we have
\begin{equation}
\label{SDE}
\mathrm{d}S_t = \underbrace{\big(\tilde\sigma(t,S_t) + \alpha_t \tilde{\lambda}(t,S_t)\big)}_{:= a(t,S_t,\Gamma_t)} \mathrm{d}W_t + \underbrace{\big(\tilde{\nu}(t,S_t) + \alpha_t (\tilde{\sigma}\partial_x \tilde{\lambda})(t,S_t) + \beta_t \tilde{\lambda}(t,S_t)\big)}_{:= b(t,S_t,\Gamma_t)} \mathrm{d}t
\end{equation}
with $S$ starting from $S_0 \in \mathbb{R}$. $a(t,S_t,\Gamma_t)$ and $b(t,S_t, \Gamma_t)$ can be expressed as a sum of terms of the form
$$ h(t,S_t,\Gamma_t) := b_t f(\tilde{\lambda}(t,S_t)\Gamma_t)g(t,S_t) $$ 
where $(b_t)_{t \geq 0}$ is a uniformly bounded adapted process, $f$ bounded Lipschitz-continous on $(-\infty, 1 - \varepsilon)$ and $g$ bounded Lipschitz-continous in the space variable on $[0,T] \times \mathbb{R}$. For all $(t,x,y) \in [0,T] \times \mathbb{R} \times \mathbb{R}$,
\begin{align*}
  h(t,x,\Gamma_t) - h(t,y,\Gamma_t) &= b_t \big(f(\tilde{\lambda}(t,x)\Gamma_t)g(t,x) - f(\tilde{\lambda}(t,y)\Gamma_t)g(t,y)\big) \\
  &= b_t \big(f(\tilde{\lambda}(t,x)\Gamma_t) - f(\tilde{\lambda}(t,y)\Gamma_t) \big)g(t,x) \\
  &\quad + b_t f(\tilde{\lambda}(t,y)\Gamma_t) \big(g(t,x) - g(t,y)\big).
\end{align*}
Therefore there exists a constant $C_1 > 0$ such that for all $(t,x,y) \in [0,T] \times \mathbb{R} \times \mathbb{R}$,
$$ \big|h(t,x,\Gamma_t) - h(t,y,\Gamma_t)\big| \leq C_1 |x-y| \quad \mathbb{P}-a.s. $$
giving that there exists a constant $C_2 > 0$ such that for all $(t,x,y) \in [0,T] \times \mathbb{R} \times \mathbb{R}$,
$$ \big|a(t,x,\Gamma_t) - a(t,y,\Gamma_t)\big| + \big|b(t,x,\Gamma_t) - b(t,y,\Gamma_t)\big| \leq C_2 |x-y| \quad \mathbb{P}-a.s.. $$
This leads that the stochastic differential equation (\ref{SDE}) with the initial condition $S_0 \in \mathbb{R}$ admits a unique strong solution. Therefore we deduce that the system of stochastic differential equations (\ref{delta dynamics 2})-(\ref{spot dynamics 2})-(\ref{hedging portfolio dynamics 2}) admits a unique strong solution so does the system (\ref{delta dynamics})-(\ref{spot dynamics})-(\ref{hedging portfolio dynamics}).
\end{proof}

In what follows we will show that the system of stochastic differential equations (\ref{delta dynamics 2})-(\ref{spot dynamics 2})-(\ref{hedging portfolio dynamics 2}) can be derived from a discrete rebalancing trading strategy taking to the limit $n \rightarrow +\infty$, where $n$ represent the number of auctions of the trading strategy. To this end, let us consider for all $n \in \mathbb{N}^*$, the following discrete time rebalancing time grid $(t_i^n)_{i \in \llbracket 0,n \rrbracket}$ such that
$$ 0 = t_0^n < t_1^n < \ldots < t_{n-1}^n < t_n^n = T. $$
with $T > 0$ a fixed maturity. We set for all $n \in \mathbb{N}^*$,
$$ \Delta_n := \sup\limits_{1 \leq i \leq n}{|t_i^n - t_{i-1}^n|}. $$

Let us set $(\delta^n,S^n,V^n)_{n \in \mathbb{N}^*}$ the sequence of processes defined on $[0,T]$ such that for all $t \in [0,T]$,
\begin{align*}
  \delta_t^n &= \sum_{i=0}^{n-1}{\delta_{t_i^n} \mathbbm{1}_{\{t_i^n \leq t < t_{i+1}^n\}}} + \delta_T \mathbbm{1}_{\{t = T\}}, \numberthis \label{def discrete trading strategy delta} \\
  S_t^n &:= S_0 + \sum_{k=1}^{n}{\mathbbm{1}_{\{t_{k-1}^n \leq t \leq T\}}\left(\int_{t_{k-1}^n}^{t \wedge t_k^n}{\tilde{\sigma}(s,S_s^n) \,\mathrm{d}W_s}
+ \int_{t_{k-1}^n}^{t \wedge t_k^n}{\tilde{\nu}(s,S_s^n) \,\mathrm{d}s}\right)} \\
&\quad + \sum_{k=1}^{n}{\mathbbm{1}_{\{t_k^n \leq t \leq T\}} (\delta_{t_k^n} - \delta_{t_{k-1}^n}) \tilde{\lambda}(t_k^n,S_{t_k^n -}^n)}. \numberthis \label{def discrete trading strategy spot} \\
  V_t^n &:= V_0 + \sum_{k=1}^{n}{\mathbbm{1}_{\{t_{k-1}^n \leq t \leq T\}}\delta_{t_{k-1}^n}(S_{t \wedge t_k^n -}^n - S_{t_{k-1}^n}^n)} \\
&\quad + \sum_{k=1}^{n}{\mathbbm{1}_{\{t_k^n \leq t \leq T\}} \left[\frac{1}{2} (\delta_{t_k^n} - \delta_{t_{k-1}^n})^2 \tilde{\lambda}(t_k^n,S_{t_k^n -}^n) + \delta_{t_{k-1}^n}(\delta_{t_k^n} - \delta_{t_{k-1}^n})\tilde{\lambda}(t_k^n,S_{t_k^n -}^n) \right]}. \numberthis \label{def discrete trading strategy V}
\end{align*}

\begin{thm}
\label{thm time continuous formulation}
Suppose the assumptions of Theorem \ref{existence and uniqueness of a strong solution} hold. Let $X := (\delta,S,V)$ the unique strong solution of the system of stochastic differential equations (\ref{delta dynamics 2})-(\ref{spot dynamics 2})-(\ref{hedging portfolio dynamics 2}) on $[0,T]$ starting from $(\delta_0, S_0, V_0) \in \mathbb{R}^3$ at $t=0$ and $(X^n)_{n \in \mathbb{N}^*} := (\delta^n,S^n,V^n)_{n \in \mathbb{N}^*}$ defined as in (\ref{def discrete trading strategy delta})-(\ref{def discrete trading strategy spot})-(\ref{def discrete trading strategy V}). Assume that $\Delta_n = O\left(\displaystyle\frac{1}{n}\right)$. Then there exists a constant $C > 0$ such that for all $n \in \mathbb{N}^*$,
$$ \sup\limits_{t \in [0,T]}{\mathbb{E}\left[\big\Vert X_t - X_t^n \big\Vert_2^2\right]} \leq C \Delta_n. $$
In particular we have that
$$ X^n \xrightarrow[n \rightarrow +\infty]{\mathbb{L}^2, \mathbb{P}, \mathcal{L}} X. $$
Furthermore if $\displaystyle\sum_{n=1}^{+\infty}{\Delta_n} < + \infty$, we have also that
$$ X^n \xrightarrow[n \rightarrow +\infty]{\mathbb{P-}a.s., \mathbb{L}^2, \mathbb{P}, \mathcal{L}} X. $$
\end{thm}

\begin{proof}
The first convergence result is a straightforward application of Propositions \ref{prop discrete trading strategy delta}, \ref{prop discrete trading strategy spot} and \ref{prop discrete trading strategy V}. Assume now that $\displaystyle\sum_{n=1}^{+\infty}{\Delta_n} < +\infty$. Let $t \in [0,T]$, we have for all $n \in \mathbb{N}^*$,
$$ \mathbb{E}\left[\big\Vert X_t - X_t^n \big\Vert_2^2\right] \leq C \Delta_n. $$
Hence
$$ \sum_{n=1}^{+\infty}{\mathbb{E}\left[\big\Vert X_t - X_t^n \big\Vert_2^2\right]} = \mathbb{E}\left[\sum_{n=1}^{+\infty}{\big\Vert X_t - X_t^n \big\Vert_2^2}\right] < +\infty, $$
which implies that
$$ \sum_{n=1}^{+\infty}{\big\Vert X_t - X_t^n \big\Vert_2^2} < +\infty \quad\mathbb{P}-a.s. $$
and
$$ \lim\limits_{n \rightarrow +\infty}{\big\Vert X_t - X_t^n \big\Vert_2} = 0 \quad\mathbb{P-}a.s..$$
\end{proof}

\begin{coro}
\label{coro time continuous formulation}
Suppose the assumptions of Proposition \ref{existence and uniqueness of a strong solution} hold. Let $X := (\delta,S,V)$ the unique strong solution of the system of stochastic differential equations (\ref{delta dynamics 2})-(\ref{spot dynamics 2})-(\ref{hedging portfolio dynamics 2}) on $[0,T]$ starting from $(\delta_0, S_0, V_0) \in \mathbb{R}^3$ at $t=0$ and $(X^n)_{n \in \mathbb{N}^*} := (\delta^n,S^n,V^n)_{n \in \mathbb{N}^*}$ defined as in (\ref{def discrete trading strategy delta})-(\ref{def discrete trading strategy spot})-(\ref{def discrete trading strategy V}).  Assume that one of the two following conditions hold:
\begin{enumerate}[(i)]
    \item There exists $\varepsilon > 0$ such that $\Delta_n = O\left(\displaystyle\frac{1}{n^{1+\varepsilon}}\right)$. \label{condition i}
    \item $(\Delta_n)_{n \in \mathbb{N}^*}$ is a non-increasing sequence such that $\displaystyle\sum_{n=1}^{+\infty}{\Delta_n} < +\infty$. \label{condition ii}
\end{enumerate}
Then we have
$$ X^n \xrightarrow[n \rightarrow +\infty]{\mathbb{P-}a.s., \mathbb{L}^2, \mathbb{P}, \mathcal{L}} X. $$
\end{coro}

\begin{proof}
If we suppose (\ref{condition i}), using the results on Riemann series and Theorem \ref{thm time continuous formulation} leads to the conclusion. Let us suppose now (\ref{condition ii}) and let us show that this implies that $\Delta_n = o\left(\displaystyle\frac{1}{n}\right)$ when $n \rightarrow +\infty$. Let $n \in \mathbb{N}^*$, we have
$$ \sum_{k=1}^{2n}{\Delta_k} - \sum_{k=1}^{n}{\Delta_k} = \sum_{k=n+1}^{2n}{\Delta_k} \geq n \Delta_{2n},$$
hence $\lim\limits_{n \rightarrow +\infty}{n \Delta_{2n}} = 0$. Besides we have also
\begin{align*}
    (2n+1)\Delta_{2n+1} &\leq (2n+1)\Delta_{2n} \\
    &\leq 2n\Delta_{2n} + \Delta_{2n}.
\end{align*}
Therefore we have 
$$ \lim\limits_{n \rightarrow +\infty}{2n\Delta_{2n}} = \lim\limits_{n \rightarrow +\infty}{(2n+1)\Delta_{2n+1}} = 0, $$ 
wich implies that 
$$ \lim\limits_{n \rightarrow +\infty}{n\Delta_n} = 0. $$ 
The conclusion is now a straightforward application of Theorem \ref{thm time continuous formulation}.
\end{proof}

\subsubsection{The pricing equation}
\label{The pricing equation - sde}

In option theory pricing the establishment of the valuation partial differential equation is often connected to the \textit{replication problem} which consists in finding a self-financed strategy and an initial wealth $V_0$ such that $\Phi(S_T) = V_T$, $\mathbb{P}-$a.s. where $\Phi$ is a terminal \textit{payoff}. Let $(\delta, S,V)$ a strong solution to the system of stochastic differential equations (\ref{delta dynamics 2})-(\ref{spot dynamics 2})-(\ref{hedging portfolio dynamics 2}), which can take the form given in (\ref{delta dynamics})-(\ref{spot dynamics})-(\ref{hedging portfolio dynamics}), and $u \in \mathcal{C}^{1,2}([0,T] \times \mathbb{R})$ a smooth function such that $u(T,.) = \Phi$. The absence of arbitrage opportunities implies that for all $t \in [0,T]$,  $\mathbb{P}-$a.s., $u(t,S_t) = V_t$. We consider the strategy given by $\delta(t,S_t) = \partial_s u(t,S_t)$, hence we have $\Gamma_t =  \partial_{ss} u(t,S_t)$, $\mathbb{P}-$a.s. by the unicity of the Ito decomposition. Applying Ito's lemma to $u$ gives
\begin{equation}
\label{equation du - sde}
\mathrm{d}u(t,S_t) = \partial_t u(t,S_t)\mathrm{d}t + \partial_s u(t,S_t)\mathrm{d}S_t + \frac{1}{2}\partial_{ss}u(t,S_t)\mathrm{d}\langle S,S \rangle_t,
\end{equation}
besides we have also
\begin{equation}
\label{equation dV - sde}
\mathrm{d}V_t = \partial_s u(t,S_t) \mathrm{d}S_t + \frac{1}{2} \lambda(t,S_t) \Gamma_t^2 \mathrm{d}\langle S,S \rangle_t.
\end{equation}
By combining (\ref{delta dynamics}) and (\ref{spot dynamics}), we obtain
$$ \mathrm{d}S_t = \frac{\sigma(t,S_t)S_t}{1 - \lambda(t,S_t)S_t \Gamma_t}\mathrm{d}W_t + \frac{\mu_t S_t + a_t \lambda(t,S_t)}{1 - \lambda(t,S_t)S_t \Gamma_t}\mathrm{d}t $$
with
\begin{equation}
\label{definition mu}
\mu_t := \nu(t,S_t) + \frac{\sigma^2(t,S_t) \Gamma_t \big(\lambda(t,S_t) + S_t \partial_x\lambda(t,S_t)\big)}{1 - \lambda(t,S_t)S_t\Gamma_t}.
\end{equation}
hence 
$$ \mathrm{d}\langle S,S \rangle_t = \frac{\sigma^2(t,S_t)S_t^2}{(1 - \lambda(t,S_t)S_t\Gamma_t)^2}, $$
which reads in the equality between (\ref{equation du - sde}) and (\ref{equation dV - sde})
$$ \partial_t{u}(t,S_t) + \frac{1}{2}\sigma^2(t,S_t)S_t^2 \Gamma_t \frac{1}{1 - \lambda(t,S_t)S_t \Gamma_t} = 0, \mathbb{P}-a.s.. $$
Therefore $u \in \mathcal{C}^{1,2}([0,T] \times \mathbb{R})$ is a solution of
\begin{align*}
\partial_t{u} + \frac{1}{2}\sigma^2 s^2 \partial_{ss}u \frac{1}{1 - \phi} &= 0, \numberthis \label{pricing equation SDE - 1} \\
\phi(t,s) &= \lambda(t,s) s\partial_{ss}u(t,s), \numberthis \label{pricing equation SDE - 2} \\
1 - \phi &\geq \varepsilon, \numberthis \label{pricing equation SDE - 3} \\
u(T,.) &= \Phi. \numberthis \label{pricing equation SDE - 4}
\end{align*}

\section{The Market Impact of Hedging Metaorders}
\label{The Market Impact of Hedging Metaorders}

In \cite{said2018} the authors define a metaorder as a large trading order split into small pieces and executed incrementally the same day by the same \textit{agent} on the same \textit{stock} and all having the same \textit{direction} (buy or sell). This and Proposition \ref{prop regular} motivates the following definition:
\begin{defi}
\label{defi metaorder}
Let $(u_n)_{n \in \mathbb{N}}$ a regular $\phi-$market impact scenario. The market impact scenario $(u_n)_{n \in \mathbb{N}}$ is said to be a \textbf{metaorder} if, and only if, $\phi \in (0,1)$. We will also refer to these metaorders as \textbf{hedging metaorders}.
\end{defi}

\subsection{Dynamics of Hedging Metaorders}
\label{Dynamics of Hedging Metaorders}

Let us consider an \textit{hedging metaorder} $(\mathrm{d}S_n)_{n \in \mathbb{N}}$. The proofs given to establish Theorem \ref{thm phi} have shown that there exists $r \in (0,1)$ such that for all $n \in \mathbb{N}$, $|\mathrm{d}S_{n+1}| \leq r \times |\mathrm{d}S_n|$. Without loss of generality we will consider for now in the rest of the section that $(\mathrm{d}S_n)_{n \in \mathbb{N}}$ is a buy metaorder i.e. $\mathrm{d}S_n > 0$ for any $n \in \mathbb{N}$. Let us set for all $n \in \mathbb{N}$,
\begin{equation}
\label{temporary impact after n trades}
  \mathcal{I}_n := \sum_{k=0}^{n}{\mathrm{d}S_k}
\end{equation}
the temporary market impact after  $n$ trades. As the sequence $(\mathrm{d}S_n)_{n \in \mathbb{N}}$ is strictly non-increasing, the plot $(n,\mathcal{I}_n)_{n \in \mathbb{N}}$ has a concave non-decreasing shape reaching a plateau as $n \rightarrow +\infty$ since we have that $\lim\limits_{n \rightarrow +\infty}{\mathrm{d}S_n} = 0$. In order to illustrate those ideas let us plot the dynamics of an hedging metaorder in the case that the $\mathrm{d}S_{n+1} = r \times \mathrm{d}S_n$ for any $n \in \mathbb{N}$ with $r = 0.8$ and $\mathrm{d}S_1 = 5$ bps.

\begin{figure}[H]
\centering
\includegraphics[scale=0.75]{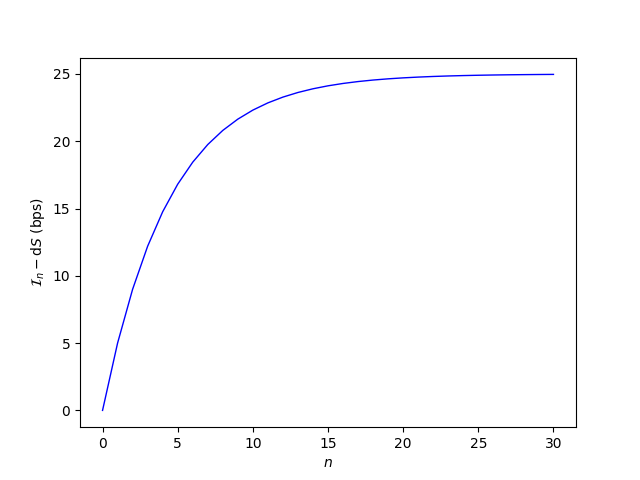}
\end{figure}

\subsection{Immediate Impact and Size of the Hedging Metaorder}
\label{Immediate Impact VS Size of the Hedging Metaorder}

Let us consider an \textit{hedging metaorder} $(\mathrm{d}S_n)_{n \in \mathbb{N}}$ and recall that the number of shares executed during the \textit{hedging metaorder} is given by
$$ N = \Gamma \sum_{n=0}^{+\infty}{\mathrm{d}S_n}. $$
This leads to the following linear relation
$$ \mathcal{I} =  \Gamma^{-1} N, $$

\subsection{Relaxation of Hedging Metaorders}
\label{Relaxation of Hedging Metaorders}

Let us consider an \textit{hedging metaorder} $(\mathrm{d}S_n)_{n \in \mathbb{N}}$. By Proposition \ref{prop regular} we can consider that $N(\mathrm{d}S) = +\infty$. We have already established that at the leading order
$$ \mathcal{R} =  \sum_{n=0}^{+\infty}{\Gamma \mathrm{d}S_n} \times \left(\displaystyle\sum_{n=0}^{+\infty}{\mathrm{d}S_n} -\displaystyle\frac{\displaystyle\sum_{n=0}^{+\infty}{\Gamma \mathrm{d}S_n \left(\displaystyle\sum_{k=0}^{n}{\mathrm{d}S_k}\right)}}{\displaystyle\sum_{n=0}^{+\infty}{\Gamma \mathrm{d}S_n}}\right) = \frac{1}{2} \lambda \left(\displaystyle\sum_{n=0}^{+\infty}{\Gamma \mathrm{d}S_n}\right)^2 S^{1 + \zeta},$$
leading to
$$ \displaystyle\frac{\displaystyle\sum_{n=0}^{+\infty}{\Gamma \mathrm{d}S_n \left(\displaystyle\sum_{k=0}^{n}{\mathrm{d}S_k}\right)}}{\displaystyle\sum_{n=0}^{+\infty}{\Gamma \mathrm{d}S_n}} = \left(1 - \frac{\phi}{2}\right) \sum_{n=0}^{+\infty}{\mathrm{d}S_n}. $$
It has been shown in several empirical works that the market impact of metaorders possesses two distinct phases: A \textit{temporary impact} as a consequence of the execution of the order followed by a relaxation phenomenon leading to a \textit{permanent impact}. The relaxation phenomenon has been studied from an empirical point of view in stocks market \cite{bershova2013non} \cite{bacry2015market} \cite{said2018} and with a theoretical insight in \cite{bouchaud2004fluctuations} \cite{gatheral2011exponential} and \cite{farmer2013efficiency} for instance. Recently the same has been done empirically concerning options market \cite{said2019market}. Let us consider $(\mathrm{d}S_n)_{n \in \mathbb{N}}$ a metaorder starting from $\mathrm{d}S \neq 0$, hence $N(\mathrm{d}S) = +\infty$. In the terminology of metaorders let us denote by
\begin{equation}
\label{temporary impact metaorder}
\mathcal{I} := \sum_{n=0}^{+\infty}{\mathrm{d}S_n}
\end{equation}
as the \textit{temporary impact} of the metaorder $(\mathrm{d}S_n)_{n \in \mathbb{N}}$. We will denote by $I$ the \textit{permanent impact} of the metaorder $(\mathrm{d}S_n)_{n \in \mathbb{N}}$. According to the \textit{Fair Pricing} condition -- predicted theoretically in \cite{farmer2013efficiency} and validated empirically in equity \cite{said2018} and options \cite{said2019market} market -- the \textit{permanent impact} $I$ satisfies the following identity:
\begin{equation}
\label{fair pricing identity}
  S + I = \displaystyle\frac{\displaystyle\sum_{n=0}^{+\infty}{\Gamma \mathrm{d}S_n \times (S + \mathcal{I}_n)}}{\displaystyle\sum_{n=0}^{+\infty}{\Gamma \mathrm{d}S_n}}
\end{equation}
where for all $n \in \mathbb{N}$, $\mathcal{I}_n := \displaystyle\sum_{k=0}^{n}{\mathrm{d}S_k}$ as defined in (\ref{temporary impact after n trades}), giving that
\begin{equation}
\label{fair pricing identity 2}
  I = \displaystyle\frac{\displaystyle\sum_{n=0}^{+\infty}{\Gamma \mathrm{d}S_n \times \mathcal{I}_n}}{\displaystyle\sum_{n=0}^{+\infty}{\Gamma \mathrm{d}S_n}}.
\end{equation}
The fair pricing identity (\ref{fair pricing identity}), introduced in \cite{farmer2013efficiency}, is in fact a \textit{non-arbitrage} condition as it states that the final price -- after the price reversion of the metaorder -- reaches a level such that it equals to the volume-weighted average price of the metaorder. The fair pricing identity (\ref{fair pricing identity 2}) reads that the \textit{permanent impact} is equal to the volume-weighted average \textit{temporary impacts} of the metaorder. Hence we have shown that:

\begin{thm}
\label{thm relaxation}
\begin{equation}
\label{metaorder identity}
\frac{I}{\mathcal{I}} = 1 - \frac{\phi}{2}
\end{equation}
giving the level of the relaxation phenomenon of the metaorder. Hence
\begin{equation}
\label{metorder inequality}
\frac{1}{2} \leq \frac{I}{\mathcal{I}} \leq 1.
\end{equation}
\end{thm}

Theorem \ref{thm relaxation} connects directly the intensity $\phi$ of the \textit{hedging metaorder} with the intensity of the price reversion after the end of the metaorder. Particularly it states that the more the intensity is the more the relaxation will be. According to recent empirical and theoretical works \cite{bershova2013non} \cite{farmer2013efficiency} \cite{said2018} \cite{bucci2019slow}, we have that $\left\langle 1 - \frac{1}{2} \phi \right\rangle_{metaorders} = 2/3$ for intraday metaorders where $\langle ... \rangle_{metaorders}$ stand for a mean value over all the metaorders. In \cite{bucci2019slow} the authors precise that if we study the price reversion $\sim 50$ days after the end of the metaorder we have approximately $\displaystyle\frac{I}{\mathcal{I}} \approx 0.5 $ also in agreement with (\ref{metaorder identity}). Nevertheless in our framework the limit case $\displaystyle\frac{I}{\mathcal{I}} \rightarrow \displaystyle\frac{1}{2}$ corresponds to the limit case when $\phi \rightarrow 1$ .i.e. when the market impact of the \textit{hedging metaorder} become bigger and bigger. Hence the lower bound of (\ref{metorder inequality}) can be attained in the case of very large metaorders. This is in favor that by taking in consideration very large metaorders whom the execution can last several days as in \cite{bucci2019slow}, we can find a mean-value that can be closed to $1/2$.





\section{Conclusion}
\label{Conclusion}

In this paper we have presented a perturbation theory of market impact that connects option pricing theory with market microstructure empirical findings. From the option pricing point of view our model appears to be an extension of the model presented in \cite{loeper2018option} in which we study the hedging process. Furthermore we have introduced what we have called the \textit{hedging metaorders} to establish a connection between option hedging and market microstructure metaorders. Particularly we have shown in our framework that our \textit{hedging metaorders} obey to linear market impact in the size of the metaorder and possess a relaxation factor in $[1/2,1]$ directly connected to their intensity characterized by the parameter $\phi$.

\newpage
\bibliographystyle{apalike}
\bibliography{bibliography}

\newpage
\appendix

\section{Proofs}

\subsection{Proof of Theorem \ref{thm phi}}
\label{proof of thm phi}

The proof of Theorem \ref{thm phi} is a consequence of the following propositions (see Propositions \ref{prop admissible 1}, \ref{prop admissible 2}, \ref{prop admissible 3} and \ref{prop admissible 4}). It gives an equivalent criterion on the parameter $\phi$ for a market impact scenario to be admissible. We will establish that a market impact scenario is admissible from a trading perspective if, and only if, $\phi \in (-\infty,1)$.

\begin{prop_appendix}
\label{prop admissible 1}
Let $(u_n)_{n \in \mathbb{N}}$ an $\phi-$market impact scenario. If $\phi \geq 1$, then $(u_n)_{n \in \mathbb{N}}$ is not admissible from a trading perspective.
\end{prop_appendix}

\begin{proof}
Let $R > 0$ and $x \in (0,R)$. Hence for all $n \in \mathbb{N}$, $u_n(x) > 0$ which implies that for all $n \in \mathbb{N}$, $s_n(x) > 0$. Thus for all $n \in \mathbb{N}$, 
$$ -s_n(x) -S \leq 0 $$ 
and 
$$ \phi \left(1 + \displaystyle\frac{s_n(x)}{S}\right)^{1 + \zeta} u_n(x) \geq 0, $$
which gives for all $n \in \mathbb{N}$,
$$ u_{n+1}(x) = \phi \left(1 + \displaystyle\frac{s_n(x)}{S}\right)^{1 + \zeta} u_n(x) \geq \phi u_n(x).$$ 
In that case the series $\sum_{n \geq 0}{u_n(x)}$ diverges and $\displaystyle\sum_{n=0}^{+\infty}{u_n(x)} = +\infty$.
\end{proof}

\begin{prop_appendix}
\label{prop admissible 2}
Let $(u_n)_{n \in \mathbb{N}}$ an $\phi-$market impact scenario. If $\phi \in [0,1)$, then $(u_n)_{n \in \mathbb{N}}$ is admissible from a trading perspective.
\end{prop_appendix}

\begin{proof}
When $\phi = 0$, for all $n \in \mathbb{N}$, $u_n = 0$ and $(u_n)_{n \in \mathbb{N}}$ is obviously admissible from a trading perspective. We consider now the case $\phi \in (0,1)$. Let $r \in (\phi, 1 \wedge 4\phi)$, $R = (1-r)S \left(\displaystyle\sqrt{\frac{r}{\phi}} - 1 \right)$ and $x \in [-R,R]$. Set $\mathcal{A}(x) = \left\{n \in \mathbb{N} \,\bigg|\, s_n(x) = -S \right\}$.
\begin{itemize}
    \item If $\mathcal{A}(x) \neq \emptyset$, $\mathcal{A}(x)$ has a least element $n_0(x)$ and for all $n > n_0(x)$, $u_n(x) = 0$. This gives the absolute convergence of the series $\sum_{n \geq 0}{u_n(x)}$.
    \item Let us suppose $\mathcal{A}(x) = \emptyset$, hence for all $n \in \mathbb{N}$, $\phi \left(1 + \displaystyle\frac{s_n(x)}{S}\right)^{1 + \zeta} u_n(x) \geq -s_n(x) -S$ and for all $n \in \mathbb{N}$, $u_{n+1}(x) = \phi \left(1 + \displaystyle\frac{s_n(x)}{S}\right)^{1 + \zeta} u_n(x)$. Let us show by induction that for all $n \in \mathbb{N}$, $|u_{n+1}(x)| \leq r |u_n(x)|$.
    \begin{itemize}
    \item If $n = 0$, $|u_1(x)| = \phi \left|1 + \displaystyle\frac{x}{S}\right|^{1 + \zeta} |u_0(x)|$ and $|x| \leq R < S$. Thus 
    $$ |u_1(x)| \leq \phi \left(1 + \displaystyle\frac{|x|}{S}\right)^2 |u_0(x)| $$
    with $\displaystyle\frac{|x|}{S} \leq \displaystyle\frac{|x|}{(1-r)S} \leq \sqrt{\displaystyle\frac{r}{\phi}} - 1$, which gives $|u_1(x)| \leq r |u_0(x)|$.
    \item Let $n \geq 1$ and suppose that for all $k \in \llbracket0,n-1\rrbracket$, $|u_{k+1}(x)| \leq r |u_k(x)|$. Hence for all $k \in \llbracket0,n\rrbracket$, $|u_k(x)| \leq r^k |x|$ and $ \displaystyle\frac{|s_n(x)|}{S} \leq \displaystyle\frac{1}{S}\sum_{k=0}^{n}{r^k |x|} \leq \displaystyle\frac{|x|}{(1-r)S} \leq \sqrt{\displaystyle\frac{r}{\phi}} - 1 < 1$ which implies $$ |u_{n+1}(x)| = \phi \left(1 + \displaystyle\frac{s_n(x)}{S}\right)^{1 + \zeta} |u_n(x)| \leq \phi \left(1 + \displaystyle\frac{|s_n(x)|}{S}\right)^2 |u_n(x)| \leq r|u_n(x)|. $$
    \end{itemize}
    Thus for all $n \in \mathbb{N}$, $|u_{n+1}(x)| \leq r |u_n(x)|$ leading to the absolute convergence of the series $\sum_{n \geq 0}{u_n(x)}$. 
\end{itemize}
Finally for all $x \in [-R,R]$, $\sum_{n \geq 0}{u_n(x)}$ converges absolutely.
\end{proof}

\begin{prop_appendix}
\label{prop admissible 3}
Let $(u_n)_{n \in \mathbb{N}}$ an $\phi-$market impact scenario. If $\phi \in (-1,0)$, then $(u_n)_{n \in \mathbb{N}}$ is admissible from a trading perspective.
\end{prop_appendix}

\begin{proof}
Let $r \in (|\phi|, 1 \wedge 4|\phi|)$, $R = S \left(\displaystyle\sqrt{\frac{r}{|\phi|}} - 1 \right)$ and $x \in [-R,R]$. $\sum_{n \geq 0}{u_n(x)}$ is an alternating series and for all $n \in \mathbb{N}$, $|s_n(x)| \leq |x| \leq R < S$. Set $\mathcal{A}(x) = \left\{n \in \mathbb{N} \,\bigg|\, s_n(x) = -S \right\}$.
\begin{itemize}
    \item If $\mathcal{A}(x) \neq \emptyset$, $\mathcal{A}(x)$ has a least element $n_0(x)$ and for all $n > n_0(x)$, $u_n(x) = 0$. This gives the absolute convergence of the series $\sum_{n \geq 0}{u_n(x)}$.
    \item Let us suppose $\mathcal{A}(x) = \emptyset$, hence for all $n \in \mathbb{N}$, $\phi \left(1 + \displaystyle\frac{s_n(x)}{S}\right)^{1 + \zeta} u_n(x) \geq -s_n(x) -S$ and for all $n \in \mathbb{N}$, $u_{n+1}(x) = \phi \left(1 + \displaystyle\frac{s_n(x)}{S}\right)^{1 + \zeta} u_n(x)$. We have for all $n \in \mathbb{N}$, 
    $$ |u_{n+1}(x)| \leq |\phi| \left(1 + \displaystyle\frac{s_n(x)}{S}\right)^2 |u_n(x)| $$ 
    and 
    $$ |\phi| \left(1 + \displaystyle\frac{s_n(x)}{S}\right)^2 \leq |\phi| \left(1 + \displaystyle\frac{|x|}{S}\right)^2 \leq r. $$ 
    Thus for all $n \in \mathbb{N}$, $|u_{n+1}(x)| \leq r |u_n(x)|$ which gives the absolute convergence of the series $\sum_{n \geq 0}{u_n(x)}$.
\end{itemize}
Finally for all $x \in [-R,R]$, $\sum_{n \geq 0}{u_n(x)}$ converges absolutely.
\end{proof}

\begin{prop_appendix}
\label{prop admissible 4}
Let $(u_n)_{n \in \mathbb{N}}$ an $\phi-$market impact scenario. If $\phi \in (-\infty,-1]$, then $(u_n)_{n \in \mathbb{N}}$ is admissible from a trading perspective.
\end{prop_appendix}

\begin{proof}
Let $x \in \mathbb{R}$, the case $x = 0$ being trivial let us consider now $x \neq 0$. Without loss of generality let us assume $x > 0$, the proof would be similar for $x < 0$. As $\phi < 0$, $\sum_{n \geq 0}{u_n(x)}$ is an alternating series. Set
$$ \mathcal{A}(x) = \left\{n \in \mathbb{N} \,\bigg|\, s_n(x) = -S \right\}. $$
\begin{itemize}
    \item If $\mathcal{A}(x) \neq \emptyset$ it has a least element denoted by $n_0(x)$. Hence for all $n \geq n_0(x)$, $u_n(x) = 0$ and the series $\sum_{n \geq 0}{u_n(x)}$ converges absolutely.
    \item Let us suppose $\mathcal{A}(x) \neq \emptyset$. Hence for all $n \in \mathbb{N}$, $\phi \left(1 + \displaystyle\frac{s_n(x)}{S}\right)^{1 + \zeta} u_n(x) \geq -s_n(x) -S$ and for all $n \in \mathbb{N}$, $u_{n+1}(x) = \phi \left(1 + \displaystyle\frac{s_n(x)}{S}\right)^{1 + \zeta} u_n(x)$. Let 
    $$ \mathcal{B}(x) = \left\{n \in \mathbb{N} \,\bigg|\, s_n(x) < -S + \displaystyle\frac{S}{|\phi|^{1 / (1 + \zeta)}}\right\}. $$
    Let us assume $\mathcal{B}(x) = \emptyset$, hence for all $n \in \mathbb{N}$, $s_n(x) \geq -S + \displaystyle\frac{S}{|\phi|^{1 / (1 + \zeta)}}$. As $x \neq 0$, for all $n \in \mathbb{N}$, $u_n(x) \neq 0$ and $\displaystyle\frac{|u_{n+1}(x)|}{|u_n(x)|} = |\phi| \left(1 + \displaystyle\frac{s_n(x)}{S}\right)^{1 + \zeta} $ implying that
    $|u_{n+1}(x)| \geq |u_n(x)|$. Hence $(s_{2n}(x))_{n \in \mathbb{N}}$ is a non-decreasing real valued sequence and $(s_{2n+1}(x))_{n \in \mathbb{N}}$ is non-increasing. The sequence $(s_n(x))_{n \in \mathbb{N}}$ being bounded by Assumption \ref{ass s_n(x) bounded}, $(s_{2n}(x))_{n \in \mathbb{N}}$ and $(s_{2n+1}(x))_{n \in \mathbb{N}}$ are convergent. Let $l(x) = \lim\limits_{n \rightarrow +\infty}{s_{2n+1}(x)}$ and $l'(x) = \lim\limits_{n \rightarrow +\infty}{s_{2n}(x)}$. As $l(x) < 0$ and $l'(x) > 0$, we have $l(x) < l'(x)$. Besides the sequence $(|u_n(x)|)_{n \in \mathbb{N}}$ is convergent and $\lim\limits_{n \rightarrow +\infty}{|u_n(x)|} = l'(x) - l(x) > 0$, hence $\lim\limits_{n \rightarrow +\infty}{\displaystyle\frac{|u_{2n+1}(x)|}{|u_{2n}(x)|}} = 1$. We have also for all $n \in \mathbb{N}$,
    $$ \displaystyle\frac{|u_{2n+1}(x)|}{|u_{2n}(x)|} = |\phi| \left(1 + \displaystyle\frac{s_{2n}(x)}{S}\right)^{1 + \zeta} $$ 
which gives 
$$ \lim\limits_{n \rightarrow +\infty}{\displaystyle\frac{|u_{2n+1}(x)|}{|u_{2n}(x)|}} = |\phi| \left(1 + \displaystyle\frac{l'(x)}{S}\right)^{1 + \zeta}. $$
This leads to $l'(x) = -S + \displaystyle\frac{S}{|\phi|^{1 / (1 + \zeta)}} \leq l(x)$ which implies a contradiction. Thus $\mathcal{B}(x)$ is a non empty subset of $\mathbb{N}$ and has a least element $n_1(x)$. Without loss of generality let us suppose that $n_1(x)$ is odd hence we can write $n_1(x) = 2k(x) + 1$ with $k(x) \in \mathbb{N}$. Set
    $$ \mathcal{C}(x) = \left\{n \in \mathbb{N} \,\bigg|\, s_n(x) \vee s_{n+1}(x) < -S + \displaystyle\frac{S}{|\phi|^{1 / (1 + \zeta)}}\right\}. $$
    If $\mathcal{C}(x) = \emptyset$ then $(s_{2n}(x))_{n \geq k(x)}$ and $(s_{2n+1}(x))_{n \geq k(x)}$ are two sequences non-increasing bounded from below. Thus $(s_{2n}(x))_{n \geq k(x)}$ and $(s_{2n+1}(x))_{n \geq k(x)}$ are convergent. Let denote by $l(x)$ and $l'(x)$ their respective limit such that $l(x) = \lim\limits_{n \rightarrow +\infty}{s_{2n+1}(x)}$ and $l'(x) = \lim\limits_{n \rightarrow +\infty}{s_{2n}(x)}$. As $l(x) < -S + \displaystyle\frac{S}{|\phi|^{1 / (1 + \zeta)}} $ and $-S + \displaystyle\frac{S}{|\phi|^{1 / (1 + \zeta)}} \leq l'(x)$ we have $l(x) < l'(x)$. As previously this leads to $l(x) = l'(x) = -S + \displaystyle\frac{S}{|\phi|^{1 / (1 + \zeta)}} $ which gives a contradiction. Thus $\mathcal{C}(x)$ is a non empty subset of $\mathbb{N}$ and has a least element $n_2(x) \geq n_1(x)$ since $\mathcal{C}(x) \subset \mathcal{B}(x)$. For all $n > n_2(x)$, $\displaystyle\frac{|u_{n+1}(x)|}{|u_n(x)|} = |\phi| \left(1 + \displaystyle\frac{s_{n_2(x)}(x)}{S}\right)^{1 + \zeta} < 1$.
\end{itemize}
Thus the series $\sum_{n \geq 0}{u_n(x)}$ is absolutely convergent for any $x \in \mathbb{R}$.
\end{proof}

\subsection{Proof of Proposition \ref{prop sum series}}
\label{proof prop sum series}

\begin{proof}
Let $R' > 0$ such that $(-R',R') \subset (-S, +\infty)$. Besides by Theorem \ref{thm lambda gamma 2} there exists also $R'' > 0$ such that for all $x \in (-R'',R'')$, $\sum_{n \geq 0}{u_n(x)}$ converges absolutely. Let $R = R' \wedge R''$, for all $x \in (-R,R)$, $N(x) \geq 1$.
\begin{itemize}
    \item When $N(x) < +\infty$, for all $n \geq N(x)$, $s_n(x) = -S$ and for all $n \geq N(x)+1$, $u_n(x) = 0$. Hence for all $n \in \mathbb{N}$, $u_{n+1}(x) =  \phi \left(1 + \displaystyle\frac{s_n(x)}{S}\right)^{1 + \zeta} u_n(x) $.
    \item When $N(x) = +\infty$, for all $n \in \mathbb{N}$, $\phi \left(1 + \displaystyle\frac{s_n(x)}{S}\right)^{1 + \zeta} u_n(x) \geq -s_n(x) -S$ and for all $n \in \mathbb{N}$, $u_{n+1}(x) = \phi \left(1 + \displaystyle\frac{s_n(x)}{S}\right)^{1 + \zeta} u_n(x)$.
\end{itemize}
In any case for all $n \in \mathbb{N}$, $ u_{n+1}(x) = \phi \left(1 + \displaystyle\frac{(1 + \zeta) s_n(x)}{S} + \displaystyle\frac{ \zeta s_n^2(x)}{S^2}\right) u_n(x) $. This implies
\begin{align*}
    \sum_{n=1}^{N(x)}{u_n(x)} &= \phi \sum_{n=0}^{N(x)}{u_n(x)} + \frac{(1 + \zeta) \phi}{S} \sum_{n=0}^{N(x)}{s_n(x) u_n(x)} + \frac{\zeta \phi}{S^2} \sum_{n=0}^{N(x)}{s_n^2(x) u_n(x)} \\
    \sum_{n=0}^{N(x)}{u_n(x)} &= x + \phi \sum_{n=0}^{N(x)}{u_n(x)} + \frac{(1 + \zeta) \phi}{S} \sum_{n=0}^{N(x)}{s_n(x) u_n(x)} + \frac{\zeta \phi}{S^2} \sum_{n=0}^{N(x)}{s_n^2(x) u_n(x)} \\
    (1 - \phi)\sum_{n=0}^{N(x)}{u_n(x)} &= x + \frac{(1 + \zeta) \phi}{S} \sum_{n=0}^{N(x)}{s_n(x) u_n(x)} + \frac{\zeta \phi}{S^2} \sum_{n=0}^{N(x)}{s_n^2(x) u_n(x)}
\end{align*}
which finally gives
$$ \sum_{n=0}^{N(x)}{u_n(x)} = \frac{x}{1 - \phi} + \frac{1}{S}\frac{(1 + \zeta) \phi}{1 - \phi}\sum_{n=0}^{N(x)}{s_n(x)u_n(x)} + \frac{1}{S^2}\frac{\zeta \phi}{1 - \phi}\sum_{n=0}^{N(x)}{s_n^2(x)u_n(x)}. $$
\end{proof}

\subsection{Proof of Proposition \ref{prop regular}}
\label{proof prop regular}

\begin{proof}
\begin{itemize}
    \item $\phi = 0$, in that case for all $x \in \mathbb{R}$, $N(x) = +\infty$.
    \item $|\phi| \in (0,1)$, let $r \in (|\phi|, 1 \wedge \frac{9}{4}|\phi|)$, $R = (1-r)S \left(\displaystyle\sqrt{\frac{r}{|\phi|}} - 1 \right)$ and $x \in [-R,R]$. Let us show by induction that for all $n \in \mathbb{N}$, $|u_{n+1}(x)| \leq r |u_n(x)|$.
    \begin{itemize}
        \item If $n = 0$, $|u_1(x)| = |\phi| \left|1 + \displaystyle\frac{x}{S}\right|^{1 + \zeta} |u_0(x)|$ and $|x| \leq R < S$. Thus 
        $$ |u_1(x)| \leq |\phi| \left(1 + \displaystyle\frac{|x|}{S}\right)^2 |u_0(x)| $$
        with $\displaystyle\frac{|x|}{S} \leq \displaystyle\frac{|x|}{(1-r)S} \leq \sqrt{\displaystyle\frac{r}{|\phi|}} - 1$, which gives $|u_1(x)| \leq r |u_0(x)|$.
        \item Let $n \geq 1$ and suppose that for all $k \in \llbracket0,n-1\rrbracket$, $|u_{k+1}(x)| \leq r |u_k(x)|$. Hence for all $k \in \llbracket0,n\rrbracket$, $|u_k(x)| \leq r^k |x|$ and 
        $$ |s_n(x)| \leq \displaystyle\sum_{k=0}^{n}{r^k |x|} \leq \displaystyle\frac{|x|}{(1-r)} \leq S\left(\sqrt{\displaystyle\frac{r}{|\phi|}} - 1\right) < S. $$ 
        Besides 
        $$ -s_n(x) -S \leq -S + \displaystyle\frac{|x|}{1-r} $$ 
        and 
        \begin{align*}
            \left|\phi \left(1 + \displaystyle\frac{s_n(x)}{S}\right)^{1 + \zeta} u_n(x)\right| &\leq |\phi| \left(1 + \displaystyle\frac{s_n(x)}{S}\right)^2 |u_n(x)| \\ 
            &\leq |\phi| \left(1 + \displaystyle\frac{|s_n(x)|}{S}\right)^2 |u_n(x)| \\
            &\leq r|u_n(x)| \\
            &\leq r^{n+1}|x| \\
            &\leq \frac{|x|}{1-r}.
        \end{align*}
        Hence 
        $$ -s_n(x) -S \leq -S + \displaystyle\frac{|x|}{1-r} \leq -\displaystyle\frac{|x|}{1-r} \leq \phi \left(1 + \displaystyle\frac{s_n(x)}{S}\right)^{1 + \zeta} u_n(x) $$
        and 
        $$ u_{n+1}(x) = \phi \left(1 + \displaystyle\frac{s_n(x)}{S}\right)^{1 + \zeta} u_n(x) $$ 
        which implies 
        $$ |u_{n+1}(x)| = |\phi| \left(1 + \displaystyle\frac{s_n(x)}{S}\right)^2 |u_n(x)| \leq |\phi| \left(1 + \displaystyle\frac{|s_n(x)|}{S}\right)^2 |u_n(x)| \leq r|u_n(x)|. $$
    \end{itemize}
    Thus for all $x \in [-R,R]$, $n \in \mathbb{N}$, $|u_{n+1}(x)| \leq r |u_n(x)|$ and $|s_n(x)| \leq S \left(\sqrt{\displaystyle\frac{r}{|\phi|}} - 1\right)$ which gives the absolute convergence of the series $\sum_{n \geq 0}{u_n(x)}$ with 
    $$\left|\displaystyle\sum_{n=0}^{+\infty}{u_n(x)}\right| \leq S \left(\sqrt{\displaystyle\frac{r}{|\phi|}} - 1\right) < S.$$
    Therefore for all $x \in [-R,R]$, $N(x) = +\infty$.
\end{itemize}
\end{proof}

\subsection{Proof of Proposition \ref{prop regular power series}}
\label{proof prop regular power series}

\begin{proof}
\begin{itemize}
    \item By Theorem \ref{thm lambda gamma 2} and Proposition \ref{prop regular} there exists $R > 0$ such that for all $x \in (-R,R)$, $N(x) = +\infty$ and $\sum_{n \geq 0}{u_n(x)}$ converges absolutely. Let $x \in (-R,R)$, by Proposition \ref{prop sum series} we have
    $$ \underbrace{\sum_{n=0}^{+\infty}{u_n(x)}}_{:= u(x)} = \frac{x}{1 - \phi} + \frac{1}{S}\frac{(1 + \zeta) \phi}{1 - \phi}\sum_{n=0}^{+\infty}{s_n(x)u_n(x)} + \frac{1}{S^2}\frac{\zeta \phi}{1 - \phi}\sum_{n=0}^{+\infty}{s_n^2(x)u_n(x)} $$
    and for all $n \in \mathbb{N}$,
    $$\left\{
        \begin{array}{l}
        u_0(x) = x \\
        \forall n \in \mathbb{N}, u_{n+1}(x) = \phi \left(1 + \displaystyle\frac{s_n(x)}{S}\right)^{1 + \zeta} u_n(x).
        \end{array}
    \right.$$
    Hence by induction for all $n \in \mathbb{N}$, $x \in (-R,R)$, $u_n(x) = P_n(x)$ with $P_n \in \mathbb{R}[X]$ and the sequence of polynomials $(P_n)_{n \in \mathbb{N}}$ satisfying the following properties:
    \begin{itemize}
        \item $\left\{
        \begin{array}{l}
        P_0 = X \\
        \forall n \in \mathbb{N}, P_{n+1} = \phi \left(1 + \displaystyle\frac{1}{S}\sum_{k=0}^{n}{P_k}\right)^{1 + \zeta} P_n
        \end{array}
        \right.$
        \item $\forall n \in \mathbb{N}$, $P_n(0) = 0$.
    \end{itemize}
    The series $\sum_{n \geq 0}{u_n(x)}$ being absolutely convergent is also unconditionally convergent which implies that $\sum_{n \geq 0}{u_n}$ has a power series expansion on $(-R,R)$.
    \item For all $n \in \mathbb{N}$, $P_n(0) = 0$ hence there exists $Q_n \in \mathbb{R}[X]$ such that $P_n = X Q_n$. The sequence of polynomials $(Q_n)_{n \in \mathbb{N}}$ satisfy the following recurrence relation
    $$\left\{
        \begin{array}{l}
        Q_0 = 1 \\
        \forall n \in \mathbb{N}, Q_{n+1} = \phi \left(1 + \displaystyle\frac{1}{S}\sum_{k=0}^{n}{P_k}\right)^{1 + \zeta} Q_n
        \end{array}
    \right.$$
which implies for all $n \in \mathbb{N}$, $Q_n(0) = \phi^n$. Besides
$$ u(x) = \frac{x}{1 - \phi} +  \frac{1}{S}\frac{(1 + \zeta) \phi}{1 - \phi} \underbrace{\sum_{n=0}^{+\infty}{\left(\sum_{k=0}^{n}{P_k(x)}\right) P_n(x)}}_{:= v(x)} + \frac{1}{S^2}\frac{\zeta \phi}{1 - \phi} \underbrace{\sum_{n=0}^{+\infty}{\left(\sum_{k=0}^{n}{P_k(x)}\right)^2 P_n(x)}}_{:= w(x)}. $$
Let us show that $v(x) = \displaystyle\sum_{n=0}^{+\infty}{\left(\sum_{k=0}^{n}{Q_k(0)}\right) Q_n(0)} x^2 + o(x^2)$ and $w(x) = o(x^2)$ as $x \rightarrow 0$. Let $x \in \mathbb{R}$ such that $|x| < R \wedge (1-r)S \left(\displaystyle\sqrt{\frac{r}{|\phi|}} - 1 \right)$ with $r \in (|\phi|, 1 \wedge \frac{9}{4}|\phi|)$. As shown in Proposition \ref{prop regular} for all $n \in \mathbb{N}$, $|P_n(x)| \leq r^n |x|$ which gives for all $n \in \mathbb{N}$, $|Q_n(x)| \leq r^n$. Let $x \neq 0$, we have $\displaystyle\frac{v(x)}{x^2} = \sum_{n=0}^{+\infty}{\left(\sum_{k=0}^{n}{Q_k(x)}\right) Q_n(x)}$ and for all $n \in \mathbb{N}$, 
\begin{align*}
    \left|\left(\displaystyle\sum_{k=0}^{n}{Q_k(x)}\right) Q_n(x) \right| &\leq \left(\displaystyle\sum_{k=0}^{n}{|Q_k(x)|}\right) |Q_n(x)| \\
    &\leq \left(\displaystyle\sum_{k=0}^{n}{r^k}\right) r^n \\
    &\leq \frac{1 - r^{n+1}}{1 - r} r^n
\end{align*}
and $\displaystyle\sum_{n=0}^{+\infty}{(1 - r^{n+1})r^n} < +\infty$. This leads to $\lim\limits_{x \rightarrow 0}{\displaystyle\frac{v(x)}{x^2}} = \displaystyle\sum_{n=0}^{+\infty}{\left(\sum_{k=0}^{n}{Q_k(0)}\right) Q_n(0)}$. Similarly we have $\lim\limits_{x \rightarrow 0}{\displaystyle\frac{w(x)}{x^2}} = \displaystyle\sum_{n=0}^{+\infty}{\left(\sum_{k=0}^{n}{Q_k(0)}\right)^2 P_n(0)} = 0$.
We have also
\begin{align*}
    \sum_{n=0}^{+\infty}{\left(\sum_{k=0}^{n}{Q_k(0)}\right) Q_n(0)} &= \sum_{n=0}^{+\infty}{\left(\sum_{k=0}^{n}{\phi^k}\right) \phi^n} \\
    &= \sum_{n=0}^{+\infty}{\frac{1 -  \phi^{n+1}}{1 - \phi} \phi^n} \\
    &= \frac{1}{1 - \phi}\sum_{n=0}^{+\infty}{[\phi^n - \phi^{2n+1}]} \\
    &= \frac{1}{1 - \phi}\left(\frac{1}{1 - \phi} - \frac{\phi}{1 - \phi^2}\right) \\
    &= \frac{1}{1 - \phi}\left(\frac{1}{1 - \phi} - \frac{\phi}{(1 - \phi)(1 + \phi)}\right) \\
    &= \frac{1}{(1 + \phi)(1 - \phi)^2}
\end{align*}
leading to
$$ u(x) = \displaystyle\frac{1}{1 - \phi} x + \displaystyle\frac{1}{S}\displaystyle\frac{(1 + \zeta) \phi}{(1 - \phi)^3 (1 + \phi)} x^2 + o(x^2) \text{ as } x \rightarrow 0. $$
\end{itemize}
\end{proof}

\subsection{Proof of Theorem \ref{NK execution strategies thm - inequalities}}
\label{proof NK execution strategies thm - inequalities}

\begin{proof}
The final price of an $(N,\mathcal{K})-$execution strategy is given by
$$ S_{N,\mathcal{K}}(n_1,\dots,n_{\mathcal{K}}) = S\left(1 + \sum_{k=1}^{\mathcal{K}}{\lambda^k \sum_{1 \leq i_1 < i_2 < \dots < i_k \leq \mathcal{K}}{n_{i_1}n_{i_2} \dots n_{i_k}}}\right). $$
Using Maclaurin's inequalities (see e.g. \cite{hardy1952inequalities}) we have
$$ S_{N,\mathcal{K}}(n_1,\dots,n_{\mathcal{K}}) \leq S\left(1 + \sum_{k=1}^{\mathcal{K}}{\frac{(\lambda N)^k}{\mathcal{K}^k} \binom{\mathcal{K}}{k}}\right) $$
with equality exactly if and only if all the $n_k$ are equal. Furthermore we have
$$ S \left(1 + \lambda \sum_{k=1}^{\mathcal{K}}{n_k}\right) \leq S_{N, \mathcal{K}}(n_1,\dots,n_{\mathcal{K}}) \leq S \left(1 + \sum_{k=1}^{\mathcal{K}}{\frac{(\lambda N)^k}{k!}}\right),$$
leading straightforwardly to (\ref{NK execution strategies thm inequality price}).
To derive (\ref{NK execution strategies thm inequality average execution price}) let us study the maximum of the following function
$$ \begin{array}{ccccc}
f & : & \mathbb{R}^{\mathcal{K}} & \longrightarrow & \mathbb{R} \\
 & & (x_1,\dots,x_{\mathcal{K}}) & \longmapsto & \displaystyle\sum_{k=1}^{\mathcal{K}}{x_k \displaystyle\prod_{i=1}^{k-1}{(1 + x_i)}} \\
\end{array} $$
on 
$$ \Lambda^* := \left\{(x_1,\dots,x_{\mathcal{K}}) \in (\mathbb{R}_+^*)^{\mathcal{K}} \,\bigg|\, \sum_{k=1}^{\mathcal{K}}{x_k} = \lambda N \right\}. $$
by setting for all $k \in \llbracket 1, \mathcal{K} \rrbracket$, $x_k := \lambda n_k$. Let us consider the compact set
$$ \Lambda := \left\{(x_1,\dots,x_{\mathcal{K}}) \in (\mathbb{R}_+)^{\mathcal{K}} \,\bigg|\, \sum_{k=1}^{\mathcal{K}}{x_k} = \lambda N \right\}. $$
As the function $f$ is continuous $f_{|\Lambda}$ is bounded and reaches its upper bound, so there exists $a \in \Lambda$ such that $f(a) = \sup\limits_{x \in \Lambda}{f(x)}$. Let us suppose that there exists $i \in \llbracket 1,\mathcal{K} \rrbracket$ such that $a_i = 0$. Without loss of generality by rearranging the terms $a_k$ we can take $i = \mathcal{K}$. We set
$$ k^* := \max \big\{1 \leq i \leq \mathcal{K}\,\big|\, a_i > 0 \big\} < \mathcal{K}, $$
well-defined as $\displaystyle\sum_{k=1}^{\mathcal{K}}{a_k} = 1$.
Let $0 < \varepsilon < a_{k^*}$, we have that
$$ f(\underbrace{a_1,\dots,a_{k^*} - \varepsilon}_{k^* \text{ first terms}}, \varepsilon,0,\dots,0) - f(\underbrace{a_1,\dots,a_{k^*}}_{k^* \text{ first terms}},0,\dots,0) = \varepsilon(a_{k^*} - \varepsilon) \prod_{i=1}^{k^* - 1}{(1 + a_i)} > 0 $$
giving a contradiction. Hence for all $i \in \llbracket 1,\mathcal{K} \rrbracket$, $a_i > 0$ and $\sup\limits_{x \in \Lambda}{f(x)} = \sup\limits_{x \in \Lambda^*}{f(x)} = f(a)$. Let us set
$$ \begin{array}{ccccc}
g & : & \mathbb{R}^{\mathcal{K}} & \longrightarrow & \mathbb{R} \\
 & & (x_1,\dots,x_{\mathcal{K}}) & \longmapsto & \displaystyle\sum_{k=1}^{\mathcal{K}}{x_k} - \lambda N \\
\end{array} $$
hence $\big\{x \in (\mathbb{R_+^*})^{\mathcal{K}} \,\big|\, g(x) = 0 \big\} = \Lambda^*$. $f$ and $g$ are of class $\mathcal{C}^1$, so there exists a Lagrange multiplier $\beta \in \mathbb{R}$, such that for all $k \in \llbracket 1,\mathcal{K} \rrbracket$,
$$ \frac{\partial f}{\partial x_k}(a) = \beta \frac{\partial g}{\partial x_k}(a) $$
which reads
$$ \prod_{\substack{1 \leq i \leq \mathcal{K} \\ i \neq k}}{(1 + a_i)} = \frac{\displaystyle\prod_{1 \leq i \leq \mathcal{K}}{(1 + a_i)}}{1 + a_k} = \beta.$$
Thus we have that $a_1 = \dots = a_{\mathcal{K}} = \displaystyle\frac{\lambda N}{\mathcal{K}}$. Noticing that
$$ 1 + \sum_{k=1}^{\mathcal{K}}{\lambda^k \sum_{1 \leq i_1 < i_2 < \dots < i_k \leq \mathcal{K}}{n_{i_1} n_{i_2} \dots n_{i_k}}} = \prod_{k=1}^{\mathcal{K}}{(1 + \lambda n_k)}, $$
we get that
\begin{align*}
  \overline{S}_{N,\mathcal{K}}(n_1,\dots,n_{\mathcal{K}}) &\leq \overline{S}_{N,\mathcal{K}}\left(\frac{N}{\mathcal{K}},\dots,\frac{N}{\mathcal{K}}\right) \\
  &\leq \frac{1}{N} \sum_{k=1}^{\mathcal{K}}{\frac{N}{\mathcal{K}} S \prod_{i=1}^{k-1}{\left(1 + \frac{\lambda N}{\mathcal{K}}\right)}} \\
  &\leq \frac{S}{\mathcal{K}} \sum_{k=0}^{\mathcal{K} - 1}{\left(1 + \frac{\lambda N}{\mathcal{K}}\right)^k} \\
  &\leq \frac{S}{\mathcal{K}} \frac{\left(1 + \frac{\lambda N}{\mathcal{K}}\right)^{\mathcal{K}} - 1}{\frac{\lambda N}{\mathcal{K}}} \\
  &\leq S \frac{e^{\lambda N} - 1}{\lambda N}.
\end{align*}
Besides we have also
\begin{align*}
  \overline{S}_{N,\mathcal{K}}(n_1,\dots,n_{\mathcal{K}}) &\geq \frac{1}{N}\sum_{k=1}^{\mathcal{K}}{n_k S\left(1 + \lambda \sum_{l=1}^{k-1}{n_l}\right)} \\
  &\geq S + \frac{\lambda S}{N} \sum_{1 \leq l < k \leq \mathcal{K}}{n_k n_l} \\
  &\geq S + \frac{1}{2 N} \lambda S \sum_{\substack{1 \leq l,k \leq \mathcal{K} \\ l \neq k}}{n_k n_l} \\
  &\geq S + \frac{1}{2 N} \lambda S \left(\sum_{1 \leq k,l \leq \mathcal{K}}{n_k n_l} - \sum_{k=1}^{\mathcal{K}}{n_k^2}\right) \\
  &\geq S + \frac{1}{2} \lambda S \left(N - \frac{1}{N}\sum_{k=1}^{\mathcal{K}}{n_k^2}\right) \\
  &\geq S + \frac{1}{2} \lambda S \left(N - \sup\limits_{1 \leq k \leq \mathcal{K}}{n_k}\right)
\end{align*}
since the following inequality holds
$$ \frac{N}{\mathcal{K}} \leq \frac{1}{N}\sum_{k=1}^{\mathcal{K}}{n_k^2} \leq \sup\limits_{1 \leq k \leq \mathcal{K}}{n_k}. $$
\end{proof}

\subsection{Proof of Theorem \ref{NK execution strategies thm - limits}}
\label{proof NK execution strategies thm - limits}

\begin{proof}
Let us notice that 
$$ 1 + \sum_{k=1}^{\mathcal{K}}{\lambda^k \sum_{1 \leq i_1 < i_2 < \dots < i_k \leq \mathcal{K}}{n_{i_1} n_{i_2} \dots n_{i_k}}} = \prod_{k=1}^{\mathcal{K}}{(1 + \lambda n_k)}. $$
We have
\begin{align*}
  0 \leq \lambda N - \ln \left(\prod_{k=1}^{\mathcal{K}}{(1 + \lambda n_k)}\right) &= \sum_{k=1}^{\mathcal{K}}{\left(\lambda n_k - \ln(1 + \lambda n_k)\right)} \\
  &\leq \frac{\lambda^2}{2} \sum_{k=1}^{\mathcal{K}}{n_k^2} \\
  &\leq \frac{\lambda^2}{2} N \sup\limits_{1 \leq k \leq \mathcal{K}}{n_k} \xrightarrow[\mathcal{K} \rightarrow +\infty]{} 0.
\end{align*}
Hence
$$ \lim\limits_{\mathcal{K} \rightarrow +\infty}{\left(1 + \sum_{k=1}^{\mathcal{K}}{\lambda^k \sum_{1 \leq i_1 < i_2 < \dots < i_k \leq \mathcal{K}}{n_{i_1} n_{i_2} \dots n_{i_k}}}\right)} = e^{\gamma \lambda N}. $$
Besides we also have that
\begin{align*}
  \left| \frac{e^{\lambda N} - 1}{\lambda} - \sum_{k=1}^{\mathcal{K}}{n_k \prod_{i=1}^{k-1}{(1 + \lambda n_i)}} \right| &=
  \left| \int_{0}^{N}{e^{\lambda n}\,\mathrm{d}n} - \sum_{k=1}^{\mathcal{K}}{n_k \prod_{i=1}^{k-1}{(1 + \lambda n_i)}} \right| \\
  &\leq \left| \int_{0}^{N}{e^{\lambda n}\,\mathrm{d}n} - \sum_{k=1}^{\mathcal{K}}{n_k e^{\lambda \sum_{i=1}^{k-1}{n_i}}} \right| \\
  &\qquad + \left| \sum_{k=1}^{\mathcal{K}}{n_k e^{\lambda \sum_{i=1}^{k-1}{n_i}}} - \sum_{k=1}^{\mathcal{K}}{n_k e^{\sum_{i=1}^{k-1}{\ln(1 + \lambda n_i)}}} \right| \\
  &\leq \left| \int_{0}^{N}{e^{\lambda n}\,\mathrm{d}n} - \sum_{k=1}^{\mathcal{K}}{n_k e^{\lambda \sum_{i=1}^{k-1}{n_i}}} \right| \\
  &\qquad + \sum_{k=1}^{\mathcal{K}}{n_k \left(e^{\lambda \sum_{i=1}^{k-1}{n_i}} - e^{\sum_{i=1}^{k-1}{\ln(1 + \lambda n_i)}}\right)}.
\end{align*}
In addition we have that,
\begin{align*}
  0 &\leq \sum_{k=1}^{\mathcal{K}}{n_k \left(e^{\lambda \sum_{i=1}^{k-1}{n_i}} - e^{\sum_{i=1}^{k-1}{\ln(1 + \lambda n_i)}}\right)} \\
  &\leq e^{\lambda N} \sum_{k=1}^{\mathcal{K}}{n_k \left(\lambda \sum_{i=1}^{k-1}{n_i} - \sum_{i=1}^{k-1}{\ln(1 + \lambda n_i)}\right)} \\
  &\leq e^{\lambda N} \sum_{k=1}^{\mathcal{K}}{n_k \frac{\lambda^2}{2} N \sup\limits_{1 \leq k \leq \mathcal{K}}{n_k}} \\
  &\leq e^{\lambda N} \frac{(\lambda N)^2}{2} \sup\limits_{1 \leq k \leq \mathcal{K}}{n_k}
\end{align*}
where we have used that for all $(a,b) \in [0,\lambda N]^2$, $|e^a - e^b| \leq e^{\lambda N}|a -b|$. As
$$ \lim\limits_{\mathcal{K} \rightarrow +\infty}\sup\limits_{1 \leq k \leq \mathcal{K}}{|n_k|} = 0, $$
we get that
$$ \lim\limits_{\mathcal{K} \rightarrow +\infty}{\sum_{k=1}^{\mathcal{K}}{n_k \prod_{i=1}^{k-1}{(1 + \lambda n_i)}}} = \int_{0}^{N}{e^{\lambda n}\,\mathrm{d}n}. $$
Consequently,
$$ \lim\limits_{\mathcal{K} \rightarrow +\infty}{\sum_{k=1}^{\mathcal{K}}{n_k \prod_{i=1}^{k-1}{(1 + \lambda n_i)}}} = \frac{e^{\lambda N} - 1}{\lambda}. $$
\end{proof}

\subsection{Proof of Theorem \ref{thm time continuous formulation}}

\begin{prop_appendix}
\label{prop discrete trading strategy delta}
Suppose the assumptions of Theorem \ref{existence and uniqueness of a strong solution} hold. Let $(\delta^n)_{n \in \mathbb{N}}$ a sequence of processes defined on $[0,T]$ such that for all $t \in [0,T]$,
\begin{equation}
\label{def discrete trading strategy delta}
\delta_t^n := \sum_{i=0}^{n-1}{\delta_{t_i^n} \mathbbm{1}_{\{t_i^n \leq t < t_{i+1}^n\}}} + \delta_T \mathbbm{1}_{\{t = T\}}.
\end{equation}
Assume that $\lim\limits_{n \rightarrow +\infty}{\Delta_n} = 0$. Then there exists a constant $C > 0$ such that for all $n \in \mathbb{N}^*$,
$$ \sup_{t \in [0,T]}{\mathbb{E}[|\delta_t - \delta_t^n|^2]} \leq C \Delta_n. $$
\end{prop_appendix}

\begin{proof}
Let $t \in [0,T]$, if $t = T$, then we have $\delta_t^n = \delta_t$. Let us take now $t \in [t_{i-1}^n, t_i^n)$ with $i \in \llbracket 1,n \rrbracket$. For all $t \in [t_{i-1}^n, t_i^n)$,
\begin{align*}
\delta_t - \delta_t^n &= \delta_t - \delta_{t_{i-1}^n} \\
                      &= \int_{t_{i-1}^n}^{t}{\alpha_s\,\mathrm{d}W_s} + \int_{t_{i-1}^n}^{t}{\beta_s\,\mathrm{d}s}.
\end{align*}
This gives by Ito's lemma
\begin{align*}
\mathrm{d}|\delta_t - \delta_{t_{i-1}^n}|^2 &= 2(\delta_t - \delta_{t_{i-1}^n})(\alpha_t \mathrm{d}W_t + \beta_t \mathrm{d}t) + \alpha_t^2 \mathrm{d}t \\
&= 2(\delta_t - \delta_{t_{i-1}^n})\alpha_t \mathrm{d}W_t + [2(\delta_t - \delta_{t_{i-1}^n})\beta_t + \alpha_t^2] \mathrm{d}t
\end{align*}
and for all $t \in [t_{i-1}^n, t_i^n)$,
\begin{align*}
\mathbb{E}[|\delta_t - \delta_{t_{i-1}^n}|^2] &= 2 \mathbb{E}\left[\int_{t_{i-1}^n}^{t}{(\delta_s - \delta_{t_{i-1}^n})\beta_s\,\mathrm{d}s}\right] + \mathbb{E}\left[\int_{t_{i-1}^n}^{t}{\alpha_s^2 \,\mathrm{d}s}\right] \\
&\leq \mathbb{E}\left[\int_{t_{i-1}^n}^{t}{|\delta_s - \delta_{t_{i-1}^n}|^2\,\mathrm{d}s}\right] + \mathbb{E}\left[\int_{t_{i-1}^n}^{t}{\beta_s^2\,\mathrm{d}s}\right] + \mathbb{E}\left[\int_{t_{i-1}^n}^{t}{\alpha_s^2 \,\mathrm{d}s}\right] \\
&\leq \int_{t_{i-1}^n}^{t}{\mathbb{E}[|\delta_s - \delta_{t_{i-1}^n}|^2]\,\mathrm{d}s} + \int_{t_{i-1}^n}^{t}{(\mathbb{E}[\beta_s^2] + \mathbb{E}[\alpha_s^2]) \,\mathrm{d}s}.
\end{align*}
Therefore there exists a constant $C_{1} > 0$ such that for all $t \in  [t_{i-1}^n, t_i^n)$,
$$ \mathbb{E}[|\delta_t - \delta_{t_{i-1}^n}|^2] \leq  C_1 \int_{t_{i-1}^n}^{t}{\mathbb{E}[|\delta_s - \delta_{t_{i-1}^n}|^2]\,\mathrm{d}s} + C_1 \Delta_n $$
which implies by Gronwall's lemma, for all $t \in  [t_{i-1}^n, t_i^n)$,
$$ \mathbb{E}[|\delta_t - \delta_{t_{i-1}^n}|^2] \leq C_1 \Delta_n \exp{(C_1(t-t_{i-1}^n))} \leq \underbrace{C_1 \Delta_n \exp{(C_1 \Delta_n)}}_{= O(\Delta_n)}, $$
hence there exists a constant $C_2 > 0$ such that for all $t \in  [t_{i-1}^n, t_i^n)$,
$$ \mathbb{E}[|\delta_t - \delta_{t_{i-1}^n}|^2] \leq  C_2 \Delta_n. $$
Let $(\tau_k^{i,n})_{k \in \mathbb{N}}$ a non-decreasing sequence such that $\lim\limits_{k \rightarrow +\infty}{\tau_k^{i,n}} = t_i^n$, by Fatou's lemma we have
$$ \mathbb{E}\left[\liminf\limits_{k \rightarrow +\infty}{|\delta_{\tau_k^{i,n}} - \delta_{t_{i-1}^n}|^2}\right] \leq \liminf\limits_{k \rightarrow +\infty}{\mathbb{E}[|\delta_{\tau_k^{i,n}} - \delta_{t_{i-1}^n}|^2]} \leq C_2 \Delta_n, $$
and
$$ \mathbb{E}\left[\liminf\limits_{k \rightarrow +\infty}{|\delta_{\tau_k^{i,n}} - \delta_{t_{i-1}^n}|^2}\right] = \mathbb{E}\left[\lim\limits_{k \rightarrow +\infty}{|\delta_{\tau_k^{i,n}} - \delta_{t_{i-1}^n}|^2}\right] = \mathbb{E}[|\delta_{t_i^n} - \delta_{t_{i-1}^n}|^2] $$
since $\delta$ is a continuous process. Therefore we have for all $n \in \mathbb{N}^*, i \in \llbracket 1,n \rrbracket, t \in  [t_{i-1}^n, t_i^n]$,
$$ \mathbb{E}[|\delta_t - \delta_t^n|^2] \leq C_2 \Delta_n, $$
and
$$ \sup\limits_{t \in [0,T]}{\mathbb{E}[|\delta_t - \delta_t^n|^2]} = \max\limits_{1 \leq i \leq n}{\sup\limits_{t \in [t_{i-1}^n,t_i^n]}{\mathbb{E}[|\delta_t - \delta_t^n|^2]}}. $$
This gives that for all $n \in \mathbb{N}^*$,
$$ \sup\limits_{t \in [0,T]}{\mathbb{E}[|\delta_t - \delta_t^n|^2]} \leq C_2 \Delta_n. $$
\end{proof}

\begin{prop_appendix}
\label{prop discrete trading strategy spot}
Suppose the assumptions of Theorem \ref{existence and uniqueness of a strong solution} hold. Let $(S^n)_{n \in \mathbb{N}}$ a sequence of processes defined on $[0,T]$ such that for all $t \in [0,T]$,
\begin{align*}
S_t^n &:= S_0 + \sum_{k=1}^{n}{\mathbbm{1}_{\{t_{k-1}^n \leq t \leq T\}}\left(\int_{t_{k-1}^n}^{t \wedge t_k^n}{\tilde{\sigma}(s,S_s^n) \,\mathrm{d}W_s}
+ \int_{t_{k-1}^n}^{t \wedge t_k^n}{\tilde{\nu}(s,S_s^n) \,\mathrm{d}s}\right)} \\
&\quad + \sum_{k=1}^{n}{\mathbbm{1}_{\{t_k^n \leq t \leq T\}} (\delta_{t_k^n} - \delta_{t_{k-1}^n}) \tilde{\lambda}(t_k^n,S_{t_k^n -}^n)}. \numberthis \label{def discrete trading strategy spot}
\end{align*}
Assume that $\Delta_n = O\left(\displaystyle\frac{1}{n}\right)$. Then there exists a constant $C > 0$ such that for all $n \in \mathbb{N}^*$,
$$ \sup_{t \in [0,T]}{\mathbb{E}[|S_t - S_t^n|^2]} \leq C \Delta_n. $$
\end{prop_appendix}

\begin{proof}
Let $n \in \mathbb{N}^*$, $\Delta S^n := S - S^n$ and $t \in [t_{i-1}^n, t_i^n)$ for $i \in \llbracket 1,n\rrbracket$. We have for all $t \in [t_{i-1}^n, t_i^n)$,
\begin{align*}
\mathrm{d}S_t &= \tilde{\sigma}(t,S_t) \mathrm{d}W_t + \left(\tilde{\nu}(t,S_t) + \alpha_t (\tilde{\sigma} \partial_x \tilde{\lambda})(t,S_t)\right) \mathrm{d}t + \tilde{\lambda}(t,S_t) \mathrm{d}\delta_t, \\
\mathrm{d}S_t^n &= \tilde{\sigma}(t,S_t^n) \mathrm{d}W_t + \tilde{\nu}(t,S_t^n) \mathrm{d}t,
\end{align*}
which gives
\begin{align*}
\mathrm{d}\Delta S_t^n &= \mathrm{d}S_t - \mathrm{d}S_t^n \\
&= \big(\tilde{\sigma}(t,S_t) - \tilde{\sigma}(t,S_t^n) + \alpha_t \tilde{\lambda}(t,S_t)\big) \mathrm{d}W_t \\
&\quad + \big(\tilde{\nu}(t,S_t) - \tilde{\nu}(t,S_t^n) + \alpha_t (\tilde{\sigma} \partial_x \tilde{\lambda})(t,S_t) + \beta_t \tilde{\lambda}(t,S_t) \big) \mathrm{d}t.
\end{align*}
Hence by Ito's lemma,
$$ \mathrm{d}|\Delta S_t^n|^2 = 2 \Delta S_t^n \mathrm{d}\Delta S_t^n + \mathrm{d}\langle \Delta S^n, \Delta S^n \rangle_t. $$
Therefore for all $t \in [t_{i-1}^n, t_i^n)$,
\begin{align*}
\mathbb{E}[|\Delta S_t^n|^2] &= \mathbb{E}[|\Delta S_{t_{i-1}^n}^n|^2] \\
& \quad + 2\mathbb{E}\left[\int_{t_{i-1}^n}^{t}{\Delta S_s^n \big(\tilde{\nu}(s,S_s) - \tilde{\nu}(s,S_s^n) + \alpha_s (\tilde\sigma \partial_x \tilde{\lambda})(s,S_s) + \beta_s \tilde{\lambda}(s,S_s) \big)\,\mathrm{d}s}\right] \\
& \quad + \mathbb{E}\left[\int_{t_{i-1}^n}^{t}{\big(\tilde{\sigma}(s,S_s) - \tilde{\sigma}(s,S_s^n) + \alpha_s \tilde{\lambda}(s,S_s)\big)^2\,\mathrm{d}s}\right] \\
&= \mathbb{E}[|\Delta S_{t_{i-1}^n}^n|^2] \\
& \quad + 2\mathbb{E}\left[\int_{t_{i-1}^n}^{t}{\Delta S_s^n \big(\tilde{\nu}(s,S_s) - \tilde{\nu}(s,S_s^n)\big)\,\mathrm{d}s}\right] \\
& \quad + 2\mathbb{E}\left[\int_{t_{i-1}^n}^{t}{\Delta S_s^n \big(\alpha_s (\tilde{\sigma} \partial_x \tilde{\lambda})(s,S_s) + \beta_s \tilde{\lambda}(s,S_s) \big)\,\mathrm{d}s}\right] \\
& \quad + \mathbb{E}\left[\int_{t_{i-1}^n}^{t}{\big(\tilde{\sigma}(s,S_s) - \tilde{\sigma}(s,S_s^n) + \alpha_s \tilde{\lambda}(s,S_s)\big)^2\,\mathrm{d}s}\right] \\
&\leq \mathbb{E}[|\Delta S_{t_{i-1}^n}^n|^2] \\
& \quad + 2\mathbb{E}\left[\int_{t_{i-1}^n}^{t}{|\Delta S_s^n| |\tilde{\nu}(s,S_s) - \tilde{\nu}(s,S_s^n)|\,\mathrm{d}s}\right] \\
& \quad + \mathbb{E}\left[\int_{t_{i-1}^n}^{t}{(|\Delta S_s^n|^2 + |\alpha_s (\tilde{\sigma} \partial_x \tilde{\lambda})(s,S_s) + \beta_s \tilde{\lambda}(s,S_s)|^2) \,\mathrm{d}s}\right] \\
& \quad + 2 \mathbb{E}\left[\int_{t_{i-1}^n}^{t}{(|\tilde{\sigma}(s,S_s) - \tilde{\sigma}(s,S_s^n)|^2 + |\alpha_s \tilde{\lambda}(s,S_s)|^2) \,\mathrm{d}s}\right], \\
\end{align*}
hence there exists a constant $C_1 > 0$ such that for all $t \in [t_{i-1}^n,t_i^n)$,
\begin{equation}
\label{proof discrete trading strategy spot 1}
\mathbb{E}[|\Delta S_t^n|^2] \leq \mathbb{E}[|\Delta S_{t_{i-1}^n}^n|^2] + C_1 \Delta_n + C_1 \int_{t_{i-1}^n}^{t}{\mathbb{E}[|\Delta S_s^n|^2]\,\mathrm{d}s}.
\end{equation}
Applying Gronwall's lemma we obtain for all $t \in [t_{i-1}^n,t_i^n)$,
\begin{align*}
\mathbb{E}[|\Delta S_t^n|^2] &\leq \left(\mathbb{E}[|\Delta S_{t_{i-1}^n}^n|^2] + C_1 \Delta_n \right) \exp(C_1(t-t_{i-1}^n)) \\
&\leq \left(\mathbb{E}[|\Delta S_{t_{i-1}^n}^n|^2] + C_1 \Delta_n \right) \exp(C_1 \Delta_n) \numberthis \label{proof discrete trading strategy spot 2}.
\end{align*}
Plugging (\ref{proof discrete trading strategy spot 2}) in (\ref{proof discrete trading strategy spot 1}) leads to, for all $t \in [t_{i-1}^n,t_i^n)$,
\begin{align*}
\mathbb{E}[|\Delta S_t^n|^2] &\leq \mathbb{E}[|\Delta S_{t_{i-1}^n}^n|^2] + C_1 \Delta_n + C_1 \Delta_n \left(\mathbb{E}[|\Delta S_{t_{i-1}^n}^n|^2] + C_1 \Delta_n \right) \exp(C_1 \Delta_n) \\
&\leq \mathbb{E}[|\Delta S_{t_{i-1}^n}^n|^2] \left(1 + \underbrace{C_1 \Delta_n  \exp\left(C_1 \Delta_n \right)}_{= O(\Delta_n)}\right) + \underbrace{C_1 \Delta_n + C_1^2 \Delta_n^2 \exp(C_1 \Delta_n)}_{= O(\Delta_n)}.
\end{align*}
This implies that there exists a constant $C_2 > 0$ such that for all $t \in [t_{i-1}^n, t_i^n)$,
\begin{equation}
\label{proof discrete trading strategy spot 3}
\mathbb{E}[|\Delta S_t^n|^2] \leq \mathbb{E}[|\Delta S_{t_{i-1}^n}^n|^2](1 + C_2 \Delta_n) + C_2 \Delta_n.
\end{equation}
We have
$$ \mathrm{d}((\delta_t - \delta_{t_{i-1}^n})\tilde{\lambda}(t,S_t^n)) = (\delta_t - \delta_{t_{i-1}^n}) \mathrm{d}(\tilde{\lambda}(t,S_t^n)) + \tilde{\lambda}(t,S_t^n) \mathrm{d}(\delta_t - \delta_{t_{i-1}^n}) + \mathrm{d}\langle \delta - \delta_{t_{i-1}^n}, \tilde{\lambda}(.,S^n)\rangle_t$$
and
$$ \mathrm{d}\tilde{\lambda}(t,S_t^n) = (\tilde{\sigma} \partial_x \tilde{\lambda})(t,S_t^n) \mathrm{d}W_t + \left(\partial_t \tilde{\lambda} + \tilde{\nu} \partial_x \tilde{\lambda} + \frac{1}{2} \tilde{\sigma}^2 \partial_{xx}\tilde{\lambda} \right)(t,S_t^n) \mathrm{d}t. $$
This gives that
\begin{align*}
\mathrm{d}((\delta_t - \delta_{t_{i-1}^n}) \tilde{\lambda}(t,S_t^n)) &= \big((\delta_t - \delta_{t_{i-1}^n})(\tilde{\sigma} \partial_x \tilde{\lambda})(t,S_t^n) + \alpha_t \tilde{\lambda}(t,S_t^n)\big) \mathrm{d}W_t \\
&\quad + (\delta_t - \delta_{t_{i-1}^n})\left(\partial_t \tilde{\lambda} + \tilde{\nu} \partial_x \tilde{\lambda} + \frac{1}{2} \tilde{\sigma}^2 \partial_{xx}\tilde{\lambda} \right)(t,S_t^n) \mathrm{d}t \\
&\quad + \big(\beta_t \tilde{\lambda}(t,S_t^n) + \alpha_t (\tilde{\sigma}\partial_x \tilde{\lambda})(t,S_t^n)\big) \mathrm{d}t,
\end{align*}
which implies that for all $t \in [t_{i-1}^n, t_i^n)$,
$$ (\delta_t - \delta_{t_{i-1}^n}) \tilde{\lambda}(t,S_t^n) = \mathcal{R}_t^{i,n}, $$
by setting for all $t \in [t_{i-1}^n, t_i^n]$,
\begin{align*}
\mathcal{R}_t^{i,n} &:= \int_{t_{i-1}^n}^{t}{\big((\delta_s - \delta_{t_{i-1}^n})(\tilde{\sigma} \partial_x \tilde{\lambda})(s,S_s^n) + \alpha_s \tilde{\lambda}(s,S_s^n)\big) \,\mathrm{d}W_s} \\
&\quad + \int_{t_{i-1}^n}^{t}{(\delta_s - \delta_{t_{i-1}^n})\left(\partial_t \tilde{\lambda} + \tilde{\nu} \partial_x \tilde{\lambda} + \frac{1}{2} \tilde{\sigma}^2 \partial_{xx}\tilde{\lambda} \right)(s,S_s^n) \,\mathrm{d}s} \\
&\quad + \int_{t_{i-1}^n}^{t}{\big(\beta_s \tilde{\lambda}(s,S_s^n) + \alpha_s (\tilde{\sigma}\partial_x \tilde{\lambda})(s,S_s^n)\big) \,\mathrm{d}s}.
\end{align*}
We set for all $t \in [t_{i-1}^n, t_i^n)$, $\tilde{S}_t^n := S_t^n + (\delta_t - \delta_{t_{i-1}^n})\tilde{\lambda}(t,S_t^n)$. Hence 
$$ \lim\limits_{t \rightarrow (t_i^n)^-}{\tilde{S}_t^n} = S_{t_i^n -}^n + (\delta_{t_i^n} - \delta_{t_{i-1}^n})\tilde{\lambda}(t_i^n, S_{t_i^n -}^n) = S_{t_i^n}^n .$$
Let us consider now for all  $t \in [t_{i-1}^n, t_i^n)$, $\Delta \tilde{S}_t^n := S_t - \tilde{S}_t^n = \Delta S_t^n - \mathcal{R}_t^{i,n}$.
We have
\begin{align*}
\mathrm{d}\Delta \tilde{S}_t^n &= \mathrm{d}\Delta S_t^n - \mathrm{d} \mathcal{R}_t^{i,n} \\
&= \big(\tilde{\sigma}(t,S_t) - \tilde{\sigma}(t,S_t^n) + \alpha_t \tilde{\lambda}(t,S_t)\big) \mathrm{d}W_t \\
&\quad + \big(\tilde{\nu}(t,S_t) - \tilde{\nu}(t,S_t^n) + \alpha_t (\tilde{\sigma} \partial_x \tilde{\lambda})(t,S_t) + \beta_t \tilde{\lambda}(t,S_t) \big) \mathrm{d}t \\
&\quad - \big((\delta_t - \delta_{t_{i-1}^n})(\tilde{\sigma} \partial_x \tilde{\lambda})(t,S_t^n) + \alpha_t \tilde{\lambda}(t,S_t^n)\big) \mathrm{d}W_t \\
&\quad - (\delta_t - \delta_{t_{i-1}^n})\left(\partial_t \tilde{\lambda} + \tilde{\nu} \partial_x \tilde{\lambda} + \frac{1}{2} \tilde{\sigma}^2 \partial_{xx}\tilde{\lambda} \right)(t,S_t^n) \mathrm{d}t \\
&\quad - \big(\beta_t \tilde{\lambda}(t,S_t^n) + \alpha_t (\tilde{\sigma}\partial_x \tilde{\lambda})(t,S_t^n)\big) \mathrm{d}t \\
\end{align*}
This gives by Ito's lemma,
$$ \mathrm{d}|\Delta \tilde{S}_t^n|^2 = 2\Delta \tilde{S}_t^n \mathrm{d}\Delta \tilde{S}_t^n + \mathrm{d}\langle \Delta\tilde{S}^n,\Delta\tilde{S}^n \rangle_t $$
which implies that for all $t \in [t_{i-1}^n,t_i^n)$,
\begin{align*}
\mathbb{E}[|\Delta \tilde{S}_t^n|^2] &= \mathbb{E}[|\Delta S_{t_{i-1}^n}^n|^2] \\
&\quad + 2 \mathbb{E}\left[\int_{t_{i-1}^n}^{t}{\Delta\tilde{S}_s^n \big(\tilde{\nu}(s,S_s) - \tilde{\nu}(s,S_s^n)\big)\,\mathrm{d}s}\right] \\
&\quad + 2 \mathbb{E}\left[\int_{t_{i-1}^n}^{t}{\Delta\tilde{S}_s^n \alpha_s\big((\tilde{\sigma} \partial_x \tilde{\lambda})(s,S_s) - (\tilde{\sigma} \partial_x \tilde{\lambda})(s,S_s^n)\big)\,\mathrm{d}s}\right] \\
&\quad + 2 \mathbb{E}\left[\int_{t_{i-1}^n}^{t}{\Delta\tilde{S}_s^n \beta_s (\tilde{\lambda}(s,S_s) - \tilde{\lambda}(s,S_s^n))\,\mathrm{d}s}\right] \\
&\quad -2 \mathbb{E}\left[\int_{t_{i-1}^n}^{t}{\Delta\tilde{S}_s^n (\delta_s - \delta_{t_{i-1}^n})\left(\partial_t \tilde{\lambda} + \tilde{\nu} \partial_x \tilde{\lambda} + \frac{1}{2} \tilde{\sigma}^2 \partial_{xx}\tilde{\lambda} \right)(s,S_s^n) \,\mathrm{d}s} \right] \\
&\quad + \mathbb{E}\left[\int_{t_{i-1}^n}^{t}{\bigg(\tilde{\sigma}(s,S_s) - \tilde{\sigma}(s,S_s^n) + \alpha_s \big(\tilde{\lambda}(s,S_s) - \tilde{\lambda}(s,S_s^n)\big) - (\delta_s - \delta_{t_{i-1}^n})(\tilde{\sigma} \partial_x \tilde{\lambda})(s,S_s^n)\bigg)^2 \,\mathrm{d}s}\right] \\
\end{align*}
leading to for all $t \in [t_{i-1}^n,t_i^n)$,
\begin{align*}
\mathbb{E}[|\Delta \tilde{S}_t^n|^2] &\leq \mathbb{E}[|\Delta S_{t_{i-1}^n}^n|^2] \\
&\leq \mathbb{E}\left[\int_{t_{i-1}^n}^{t}{|\Delta\tilde{S}_s^n| |\tilde{\nu}(s,S_s) - \tilde{\nu}(s,S_s^n)| \,\mathrm{d}s}\right] \\
&\quad + 2 \mathbb{E}\left[\int_{t_{i-1}^n}^{t}{|\Delta\tilde{S}_s^n| |\alpha_s| |(\tilde{\sigma} \partial_x \tilde{\lambda})(s,S_s) - (\tilde{\sigma} \partial_x \tilde{\lambda})(s,S_s^n)|\,\mathrm{d}s}\right] \\
&\quad + 2 \mathbb{E}\left[\int_{t_{i-1}^n}^{t}{|\Delta\tilde{S}_s^n| |\beta_s| |\tilde{\lambda}(s,S_s) - \tilde{\lambda}(s,S_s^n)|\,\mathrm{d}s}\right] \\
&\quad + \mathbb{E}\left[\int_{t_{i-1}^n}^{t}{\left(|\Delta\tilde{S}_s^n|^2 + |\delta_s - \delta_{t_{i-1}^n}|^2 \left|\left(\partial_t \tilde{\lambda} + \tilde{\nu} \partial_x \tilde{\lambda} + \frac{1}{2} \tilde{\sigma}^2 \partial_{xx}\tilde{\lambda} \right)(s,S_s^n)\right|^2 \right) \,\mathrm{d}s} \right] \\
&\quad + 3 \mathbb{E}\left[\int_{t_{i-1}^n}^{t}{|\tilde{\sigma}(s,S_s) - \tilde{\sigma}(s,S_s^n)|^2 \,\mathrm{d}s}\right] \\
&\quad + 3 \mathbb{E}\left[\int_{t_{i-1}^n}^{t}{|\alpha_s^2| |\tilde{\lambda}(s,S_s) - \tilde{\lambda}(s,S_s^n)|^2 \,\mathrm{d}s}\right] \\
&\quad + 3 \mathbb{E}\left[\int_{t_{i-1}^n}^{t}{|\delta_s - \delta_{t_{i-1}^n}|^2 |(\tilde{\sigma} \partial_x \tilde{\lambda})(s,S_s^n)|^2 \,\mathrm{d}s}\right]
\end{align*}
hence there exists a constant $C_3 > 0$ such that for all $t \in [t_{i-1}^n,t_i^n)$,
\begin{equation}
\label{proof discrete trading strategy spot 4}
\mathbb{E}[|\Delta \tilde{S}_t^n|^2] \leq \mathbb{E}[|\Delta S_{t_{i-1}^n}^n|^2] + C_3 \int_{t_{i-1}^n}^{t}{\left(\mathbb{E}[|\delta_s - \delta_{t_{i-1}^n}|^2] + \mathbb{E}[|\Delta S_s^n|^2] + \mathbb{E}[|\Delta \tilde{S}_s^n|^2]\right)\,\mathrm{d}s}.
\end{equation}
Besides by Proposition \ref{prop discrete trading strategy delta} there exists a constant $C_4 > 0$ such that for all $t \in [t_{i-1}^n, t_i^n)$,
\begin{equation}
\label{proof discrete trading strategy spot 5}
\mathbb{E}[|\delta_t - \delta_{t_{i-1}^n}|^2] \leq C_4 \Delta_n,
\end{equation}
therefore by plugging (\ref{proof discrete trading strategy spot 3}) and (\ref{proof discrete trading strategy spot 5})
in the inequality (\ref{proof discrete trading strategy spot 4}) we obtain that for all $t \in [t_{i-1}^n, t_i^n)$,
\begin{align*}
\mathbb{E}[|\Delta \tilde{S}_t^n|^2] &\leq \mathbb{E}[|\Delta S_{t_{i-1}^n}^n|^2] + C_3 C_4 \Delta_n^2 + C_3 \Delta_n \left(\mathbb{E}[|\Delta S_{t_{i-1}^n}^n|^2](1 + C_2 \Delta_n) + C_2 \Delta_n \right) + C_3 \int_{t_{i-1}^n}^{t}{\mathbb{E}[|\Delta \tilde{S}_s^n|^2]\,\mathrm{d}s} \\
&\leq \mathbb{E}[|\Delta S_{t_{i-1}^n}^n|^2]\left(1 + \underbrace{C_3 \Delta_n + C_2 C_3 \Delta_n^2}_{= O(\Delta_n)}\right) + \underbrace{(C_2 + C_4) C_3 \Delta_n^2}_{= O(\Delta_n^2)} + C_3 \int_{t_{i-1}^n}^{t}{\mathbb{E}[|\Delta \tilde{S}_s^n|^2]\,\mathrm{d}s}.
\end{align*}
Hence there exists a constant $C_5 > 0$ such that for all $t \in [t_{i-1}^n, t_i^n)$,
\begin{equation}
\label{proof discrete trading strategy spot 6}
\mathbb{E}[|\Delta \tilde{S}_t^n|^2] \leq \mathbb{E}[|\Delta S_{t_{i-1}^n}^n|^2](1 + C_5 \Delta_n) + C_5 \Delta_n^2 + C_5 \int_{t_{i-1}^n}^{t}{\mathbb{E}[|\Delta \tilde{S}_s^n|^2]\,\mathrm{d}s},
\end{equation}
which implies by Gronwall's lemma that for all $t \in [t_{i-1}^n, t_i^n)$,
\begin{align*}
\mathbb{E}[|\Delta \tilde{S}_t^n|^2] &\leq \exp({C_5(t-t_{i-1}^n)})\left(\mathbb{E}[|\Delta S_{t_{i-1}^n}^n|^2](1 + C_5 \Delta_n) + C_5 \Delta_n^2\right) \\
&\leq \exp(C_5 \Delta_n)\left(\mathbb{E}[|\Delta S_{t_{i-1}^n}^n|^2](1 + C_5 \Delta_n) + C_5 \Delta_n^2\right) \numberthis \label{proof discrete trading strategy spot 7}.
\end{align*}
By plugging the inequality (\ref{proof discrete trading strategy spot 7}) in (\ref{proof discrete trading strategy spot 6}), we get
\begin{align*}
\mathbb{E}[|\Delta \tilde{S}_t^n|^2] &\leq \mathbb{E}[|\Delta S_{t_{i-1}^n}^n|^2](1 + C_5 \Delta_n) + C_5 \Delta_n^2 + C_5 \Delta_n \exp(C_5 \Delta_n)\left(\mathbb{E}[|\Delta S_{t_{i-1}^n}^n|^2](1 + C_5 \Delta_n) + C_5 \Delta_n^2\right) \\
&\leq \mathbb{E}[|\Delta S_{t_{i-1}^n}^n|^2]\left(1 + \underbrace{C_5 \Delta_n + C_5 \Delta_n \exp(C_5 \Delta_n) + C_5^2 \Delta_n^2 \exp(C_5 \Delta_n)}_{= O(\Delta_n)}\right) \\
&\quad + \underbrace{C_5 \Delta_n^2 + C_5^2 \Delta_n^3 \exp(C_5 \Delta_n)}_{= O(\Delta_n^2)}.
\end{align*}
Therefore there exists a constant $C_6 > 0$ such that for all $t \in [t_{i-1}^n, t_i^n)$,
$$ \mathbb{E}[|\Delta \tilde{S}_t^n|^2] \leq \mathbb{E}[|\Delta S_{t_{i-1}^n}^n|^2](1 + C_6 \Delta_n) + C_6 \Delta_n^2. $$
Let $(\tau_k^{i,n})_{k \in \mathbb{N}}$ a non-decreasing sequence such that $\lim\limits_{k \rightarrow +\infty}{\tau_k^{i,n}} = t_i^n$, by Fatou's lemma we have
$$ \mathbb{E}\left[\liminf\limits_{k \rightarrow +\infty}{|\Delta \tilde{S}_{\tau_k^{i,n}}^n|^2}\right] \leq \liminf\limits_{k \rightarrow +\infty}{\mathbb{E}[|\Delta \tilde{S}_{\tau_k^{i,n}}^n|^2]} \leq \mathbb{E}[|\Delta S_{t_{i-1}^n}^n|^2](1 + C_6 \Delta_n) + C_6 \Delta_n^2, $$
and
$$ \mathbb{E}\left[\liminf\limits_{k \rightarrow +\infty}{|\Delta \tilde{S}_{\tau_k^{i,n}}^n|^2}\right] = \mathbb{E}\left[\lim\limits_{k \rightarrow +\infty}{|\Delta \tilde{S}_{\tau_k^{i,n}}^n|^2}\right] = \mathbb{E}[|\Delta S_{t_i^n}^n|^2] $$
since $\lim\limits_{t \rightarrow (t_i^n)^-}{\Delta\tilde{S}_t^n} = \Delta S_{t_i^n}^n$. Therefore we have for all $n \in \mathbb{N}^*, i \in \llbracket 1,n \rrbracket$,
\begin{align*}
\mathbb{E}[|\Delta S_{t_i^n}^n|^2] &\leq \mathbb{E}[|\Delta S_{t_{i-1}^n}^n|^2](1 + C_6 \Delta_n) + C_6 \Delta_n^2 \\
&\leq \mathbb{E}[|\Delta S_{t_{i-2}^n}^n|^2](1 + C_6 \Delta_n)^2 + C_6 \Delta_n^2(1 + C_6 \Delta_n) + C_6 \Delta_n^2 \\
&\ldots \\
&\leq \mathbb{E}[|\Delta S_{t_0^n}^n|^2](1 + C_6 \Delta_n)^i + C_6 \Delta_n^2\sum_{k=0}^{i-1}{(1 + C_6 \Delta_n)^k}.
\end{align*}
Considering the fact that $\Delta S_{t_0^n}^n = 0$, we have for all $n \in \mathbb{N}^*, i \in \llbracket 1,n \rrbracket$,
\begin{align*}
\mathbb{E}[|\Delta S_{t_i^n}^n|^2] &\leq C_6 \Delta_n^2\frac{(1 + C_6 \Delta_n)^i - 1}{C_6 \Delta_n} \\
&\leq \Delta_n (1 + C_6 \Delta_n)^n \\
&\leq \Delta_n \exp\big(n \ln(1 + C_6 \Delta_n)\big) \\
&\leq \underbrace{\Delta_n \exp(C_6 n \Delta_n)}_{O(\Delta_n)}.
\end{align*}
Hence there exists a constant $C_7 > 0$ such that for all $n \in \mathbb{N}^*, i \in \llbracket 0,n \rrbracket$,
$$ \mathbb{E}[|\Delta S_{t_i^n}^n|^2] \leq C_7 \Delta_n $$
and for all $n \in \mathbb{N}^*, i \in \llbracket 0,n \rrbracket, t \in [t_{i-1}^n, t_i^n),$
\begin{align*}
\mathbb{E}[|\Delta S_t^n|^2] &\leq \mathbb{E}[|\Delta S_{t_{i-1}^n}^n|^2](1 + C_2 \Delta_n) + C_2 \Delta_n^2 \\
&\leq \underbrace{C_7 \Delta_n (1 + C_2 \Delta_n) + C_2 \Delta_n^2}_{= O(\Delta_n)}
\end{align*}
which implies that there exists a constant $C_8 > 0$ such that for all $n \in \mathbb{N}^*, i \in \llbracket 1,n \rrbracket$,
$$ \sup\limits_{t \in [t_{i-1}^n,t_i^n]}{\mathbb{E}[|\Delta S_t^n|^2]} \leq C_8 \Delta_n. $$
This gives that for all $n \in \mathbb{N}^*$,
$$ \sup\limits_{t \in [0,T]}{\mathbb{E}[|\Delta S_t^n|^2]} = \max\limits_{1 \leq i \leq n}{\sup\limits_{t \in [t_{i-1}^n,t_i^n]}{\mathbb{E}[|\Delta S_t^n|^2]}} \leq C_8 \Delta_n. $$
\end{proof}

\begin{prop_appendix}
\label{prop discrete trading strategy V}
Suppose the assumptions of Theorem \ref{existence and uniqueness of a strong solution} hold. Let $(V^n)_{n \in \mathbb{N}}$ a sequence of processes on $[0,T]$ such that for all $t \in [0,T]$,
\begin{align*}
V_t^n &:= V_0 + \sum_{k=1}^{n}{\mathbbm{1}_{\{t_{k-1}^n \leq t \leq T\}}\delta_{t_{k-1}^n}(S_{t \wedge t_k^n -}^n - S_{t_{k-1}^n}^n)} \\
&\quad + \sum_{k=1}^{n}{\mathbbm{1}_{\{t_k^n \leq t \leq T\}} \left[\frac{1}{2} (\delta_{t_k^n} - \delta_{t_{k-1}^n})^2 \tilde{\lambda}(t_k^n,S_{t_k^n -}^n) + \delta_{t_{k-1}^n}(\delta_{t_k^n} - \delta_{t_{k-1}^n})\tilde{\lambda}(t_k^n,S_{t_k^n -}^n) \right]}. \numberthis \label{def discrete trading strategy V}
\end{align*}
Assume that $\Delta_n = O\left(\displaystyle\frac{1}{n}\right)$. Then there exists a constant $C > 0$ such that for all $n \in \mathbb{N}^*$,
$$ \sup_{t \in [0,T]}{\mathbb{E}[|V_t - V_t^n|^2]} \leq C \Delta_n. $$
\end{prop_appendix}

\begin{proof}
Let $n \in \mathbb{N}^*$, $\Delta V^n := V - V^n$ and $t \in [t_{i-1}^n, t_i^n)$ for $i \in \llbracket 1,n\rrbracket$. We have for all $t \in [t_{i-1}^n, t_i^n)$,
\begin{align*}
\mathrm{d}V_t &= \delta_t \mathrm{d}S_t + \frac{1}{2} \tilde{\lambda}(t,S_t) \alpha_t^2 \mathrm{d}t \\
\mathrm{d}V_t^n &= \delta_{t_{i-1}^n} \mathrm{d}S_t^n,
\end{align*}
which leads to
\begin{align*}
\mathrm{d}\Delta V_t^n &= \mathrm{d}V_t - \mathrm{d}V_t^n \\
&= \big(\delta_t \tilde{\sigma}(t,S_t) - \delta_{t_{i-1}^n} \tilde{\sigma}(t,S_t^n) + \delta_t\alpha_t \tilde{\lambda}(t,S_t)\big) \mathrm{d}W_t \\
&\quad + \left(\delta_t \tilde{\nu}(t,S_t) - \delta_{t_{i-1}^n} \tilde{\nu}(t,S_t^n) + \frac{1}{2} \alpha_t^2 \tilde{\lambda}(t,S_t) + \delta_t\beta_t \tilde{\lambda}(t,S_t) + \delta_t \alpha_t (\tilde{\sigma}\partial_x \tilde{\lambda})(t,S_t)\right) \mathrm{d}t.
\end{align*}
By Ito's lemma we have
$$ \mathrm{d}|\Delta V_t^n|^2 = 2 \Delta V_t^n \mathrm{d}\Delta V_t^n + \mathrm{d}\langle \Delta V^n, \Delta V^n \rangle_t, $$
which gives for all $t \in [t_{i-1}^n,t_i^n)$,
\begin{align*}
\mathbb{E}[|\Delta V_t^n|^2] &= \mathbb{E}[|\Delta V_{t_{i-1}^n}^n|^2] \\
& \quad + 2\mathbb{E}\left[\int_{t_{i-1}^n}^{t}{\Delta V_s^n \big(\delta_s \tilde{\nu}(s,S_s) - \delta_{t_{i-1}^n} \tilde{\nu}(s,S_s^n) \big)\,\mathrm{d}s}\right] \\
& \quad + 2\mathbb{E}\left[\int_{t_{i-1}^n}^{t}{\Delta V_s^n \left(\frac{1}{2} \alpha_s^2 \tilde{\lambda}(s,S_s) + \delta_s\beta_s \tilde{\lambda}(s,S_s) + \delta_s \alpha_s (\tilde{\sigma}\partial_x \tilde{\lambda})(s,S_s) \right)\,\mathrm{d}s}\right] \\
& \quad + \mathbb{E}\left[\int_{t_{i-1}^n}^{t}{\big(\delta_s \tilde{\sigma}(s,S_s) - \delta_{t_{i-1}^n} \tilde{\sigma}(s,S_s^n) + \delta_s\alpha_s \tilde{\lambda}(s,S_s)\big)^2\,\mathrm{d}s}\right] \\
&\leq \mathbb{E}[|\Delta V_{t_{i-1}^n}^n|^2] \\
& \quad + \mathbb{E}\left[\int_{t_{i-1}^n}^{t}{\big(|\Delta V_s^n|^2 + |\delta_s \tilde{\nu}(s,S_s) - \delta_{t_{i-1}^n} \tilde{\nu}(s,S_s^n)|^2\big) \,\mathrm{d}s}\right] \\
& \quad + \mathbb{E}\left[\int_{t_{i-1}^n}^{t}{\left(|\Delta V_s^n|^2 + \left|\frac{1}{2} \alpha_s^2 \tilde{\lambda}(s,S_s) + \delta_s\beta_s \tilde{\lambda}(s,S_s) + \delta_s \alpha_s (\tilde{\sigma}\partial_x \tilde{\lambda})(t,S_s)\right|^2\right) \,\mathrm{d}s}\right] \\
& \quad + 2 \mathbb{E}\left[\int_{t_{i-1}^n}^{t}{\big(|\delta_s \tilde{\sigma}(s,S_s) - \delta_{t_{i-1}^n} \tilde{\sigma}(s,S_s^n)|^2 + |\delta_s\alpha_s f(s,S_s)|^2\big) \,\mathrm{d}s}\right].
\end{align*}
Besides for all function $h : (t,x) \in \mathbb{R}^2 \mapsto \mathbb{R}$ bounded and Lipschitz in the space variable we have for all $t \in [t_{i-1}^n,t_i^n)$,
\begin{align*}
\mathbb{E}[|\delta_s h(s,S_s) - \delta_{t_{i-1}^n} h(s,S_s^n)|^2] &\leq 2\mathbb{E}[|h(s,S_s^n)|^2 |\delta_s - \delta_s^n|^2] + 2\mathbb{E}[|\delta_{t_{i-1}^n}|^2 |h(s,S_s) - h(s,S_s^n)|^2] \\
&\leq 2\mathbb{E}[|h(s,S_s^n)|^2 |\delta_s - \delta_s^n|^2] + 2\sqrt{\mathbb{E}[|\delta_{t_{i-1}^n}|^4]} \sqrt{\mathbb{E}[|h(s,S_s) - h(s,S_s^n)|^4]}
\end{align*}
which implies that there exists a constant $C_h > 0$ such that for all $t \in [t_{i-1}^n,t_i^n)$,
\begin{equation}
\label{useful result 1}
\mathbb{E}[|\delta_s h(s,S_s) - \delta_{t_{i-1}^n} h(s,S_s^n)|^2] \leq C_h \big(\mathbb{E}[|S_s - S_s^n|^2] + \mathbb{E}[|\delta_s - \delta_s^n|^2]\big).
\end{equation}
Hence by (\ref{useful result 1}) there exists a constant $C_1 > 0$ such that for all $t \in [t_{i-1}^n,t_i^n)$,
\begin{align*}
\mathbb{E}[|\Delta V_t^n|^2] &\leq \mathbb{E}[|\Delta V_{t_{i-1}^n}^n|^2] \\ 
&\quad + C_1 \int_{t_{i-1}^n}^{t}{\big(1 + \mathbb{E}[|\Delta V_s^n|^2] + \mathbb{E}[|S_s - S_s^n|^2] + \mathbb{E}[|\delta_s - \delta_s^n|^2] \big) \,\mathrm{d}s}. \numberthis \label{proof discrete trading strategy V 1}
\end{align*}
By plugging the Propositions \ref{prop discrete trading strategy delta} and \ref{prop discrete trading strategy spot} in the inequality (\ref{proof discrete trading strategy V 1}), we obtain that there exists a constant $C_2 > 0$ such that for all $t \in [t_{i-1}^n,t_i^n)$,
\begin{equation}
\label{proof discrete trading strategy V 2}
\mathbb{E}[|\Delta V_t^n|^2] \leq \mathbb{E}[|\Delta V_{t_{i-1}^n}^n|^2] + C_2 \Delta_n + C_2 \int_{t_{i-1}^n}^{t}{\mathbb{E}[|\Delta V_s^n|^2]\,\mathrm{d}s}.
\end{equation}
By applying Gronwall's lemma, similarly as the inequality (\ref{proof discrete trading strategy spot 1}) leads to the inequality (\ref{proof discrete trading strategy spot 3}), the inequality (\ref{proof discrete trading strategy V 2}) implies also that there exists a constant $C_3 > 0$ such that for all $t \in [t_{i-1}^n, t_i^n)$,
\begin{equation}
\label{proof discrete trading strategy V 3}
\mathbb{E}[|\Delta V_t^n|^2] \leq \mathbb{E}[|\Delta V_{t_{i-1}^n}^n|^2](1 + C_3 \Delta_n) + C_3 \Delta_n.
\end{equation}
We have for all $t \in [t_{i-1}^n,t_i^n)$,
$$ \mathrm{d}\big(|\delta_t - \delta_{t_{i-1}^n}|^2 \tilde{\lambda}(t,S_t^n)\big) = |\delta_t - \delta_{t_{i-1}^n}|^2 \mathrm{d}(\tilde{\lambda}(t,S_t^n)) + \tilde{\lambda}(t,S_t^n) \mathrm{d}\big(|\delta_t - \delta_{t_{i-1}^n}|^2\big) + \mathrm{d}\langle |\delta - \delta_{t_{i-1}^n}|^2, \tilde{\lambda}(.,S^n)\rangle_t $$
with
$$ \mathrm{d}\tilde{\lambda}(t,S_t^n) = (\tilde{\sigma} \partial_x \tilde{\lambda})(t,S_t^n) \mathrm{d}W_t + \left(\partial_t \tilde{\lambda} + \tilde{\nu} \partial_x \tilde{\lambda} + \frac{1}{2} \tilde{\sigma}^2 \partial_{xx}\tilde{\lambda} \right)(t,S_t^n) \mathrm{d}t, $$
and
$$ \mathrm{d}|\delta_t - \delta_{t_{i-1}^n}|^2 = 2(\delta_t - \delta_{t_{i-1}^n})\alpha_t \mathrm{d}W_t + \big(2(\delta_t - \delta_{t_{i-1}^n})\beta_t + \alpha_t^2\big) \mathrm{d}t. $$
Therefore for all $t \in [t_{i-1}^n,t_i^n)$,
\begin{align*}
\mathrm{d}\big(|\delta_t - \delta_{t_{i-1}^n}|^2 \tilde{\lambda}(t,S_t^n)\big) &= |\delta_t - \delta_{t_{i-1}^n}|^2 (\tilde{\sigma} \partial_x \tilde{\lambda})(t,S_t^n) \mathrm{d}W_t \\
&\quad + 2(\delta_t - \delta_{t_{i-1}^n}) \alpha_t \tilde{\lambda}(t,S_t^n) \mathrm{d}W_t \\
&\quad + |\delta_t - \delta_{t_{i-1}^n}|^2 \left(\partial_t \tilde{\lambda} + \tilde{\nu} \partial_x \tilde{\lambda} + \frac{1}{2} \tilde{\sigma}^2 \partial_{xx}\tilde{\lambda} \right)(t,S_t^n) \mathrm{d}t \\
&\quad + \big(2(\delta_t - \delta_{t_{i-1}^n})\beta_t + \alpha_t^2\big) \tilde{\lambda}(t,S_t^n) \mathrm{d}t \\
&\quad + 2(\delta_t - \delta_{t_{i-1}^n}) \alpha_t (\tilde{\sigma} \partial_x \tilde{\lambda})(t,S_t^n) \mathrm{d}t,
\end{align*}
which implies that for all $t \in [t_{i-1}^n, t_i^n)$,
$$ |\delta_t - \delta_{t_{i-1}^n}|^2 \tilde{\lambda}(t,S_t^n) = \mathcal{S}_t^{i,n}, $$
by setting for all $t \in [t_{i-1}^n, t_i^n]$,
\begin{align*}
\mathcal{S}_t^{i,n} &:= \int_{t_{i-1}^n}^{t}{|\delta_s - \delta_{t_{i-1}^n}|^2 (\tilde{\sigma} \partial_x \tilde{\lambda})(s,S_s^n) \,\mathrm{d}W_s} \\
&\quad + \int_{t_{i-1}^n}^{t}{2(\delta_s - \delta_{t_{i-1}^n}) \alpha_s \tilde{\lambda}(s,S_s^n) \,\mathrm{d}W_s} \\
&\quad + \int_{t_{i-1}^n}^{t}{|\delta_s - \delta_{t_{i-1}^n}|^2 \left(\partial_t \tilde{\lambda} + \tilde{\nu} \partial_x \tilde{\lambda} + \frac{1}{2} \tilde{\sigma}^2 \partial_{xx}\tilde{\lambda} \right)(s,S_s^n) \,\mathrm{d}s} \\
&\quad + \int_{t_{i-1}^n}^{t}{\big(2(\delta_s - \delta_{t_{i-1}^n})\beta_s + \alpha_s^2\big) \tilde{\lambda}(s,S_s^n) \,\mathrm{d}s} \\
&\quad + \int_{t_{i-1}^n}^{t}{ 2(\delta_s - \delta_{t_{i-1}^n}) \alpha_s (\tilde{\sigma} \partial_x \tilde{\lambda})(s,S_s^n) \,\mathrm{d}s}.
\end{align*}
We set for all $t \in [t_{i-1}^n, t_i^n)$, $\tilde{V}_t^n := V_t^n + \delta_{t_{i-1}^n}(\delta_t - \delta_{t_{i-1}^n})\tilde{\lambda}(t,S_t^n) + \displaystyle\frac{1}{2} (\delta_t - \delta_{t_{i-1}^n})^2 \tilde{\lambda}(t,S_t^n)$. Hence 
$$ \lim\limits_{t \rightarrow (t_i^n)^-}{\tilde{V}_t^n} = V_{t_i^n -}^n + \delta_{t_{i-1}^n}(\delta_{t_i^n} - \delta_{t_{i-1}^n})\tilde{\lambda}(t_i^n, S_{t_i^n -}^n) + \frac{1}{2} (\delta_{t_i^n} - \delta_{t_{i-1}^n})^2 \tilde{\lambda}(t_i^n, S_{t_i^n -}^n) = V_{t_i^n}^n.$$
Let us consider now for all  $t \in [t_{i-1}^n, t_i^n)$, $\Delta \tilde{V}_t^n := V_t - \tilde{V}_t^n = \Delta V_t^n - \delta_{t_{i-1}^n} \mathcal{R}_t^{i,n} - \displaystyle\frac{1}{2} \mathcal{S}_t^{i,n}$,
where $\mathcal{R}^{i,n}$ is defined, as previously, such that for all $t \in [t_{i-1}^n, t_i^n]$,
\begin{align*}
\mathcal{R}_t^{i,n} &:= \int_{t_{i-1}^n}^{t}{\big((\delta_s - \delta_{t_{i-1}^n})(\tilde{\sigma} \partial_x \tilde{\lambda})(s,S_s^n) + \alpha_s \tilde{\lambda}(s,S_s^n)\big) \,\mathrm{d}W_s} \\
&\quad + \int_{t_{i-1}^n}^{t}{(\delta_s - \delta_{t_{i-1}^n})\left(\partial_t \tilde{\lambda} + \tilde{\nu} \partial_x \tilde{\lambda} + \frac{1}{2} \tilde{\sigma}^2 \partial_{xx}\tilde{\lambda} \right)(s,S_s^n) \,\mathrm{d}s} \\
&\quad + \int_{t_{i-1}^n}^{t}{\big(\beta_s \tilde{\lambda}(s,S_s^n) + \alpha_s (\tilde{\sigma}\partial_x \tilde{\lambda})(s,S_s^n)\big) \,\mathrm{d}s}.
\end{align*}
We have by Ito's lemma
$$ \mathrm{d}|\Delta \tilde{V}_t^n|^2 = 2 \Delta \tilde{V}_t^n \mathrm{d}\Delta \tilde{V}_t^n + \mathrm{d}\langle \Delta \tilde{V}^n, \Delta \tilde{V}^n \rangle_t, $$
hence for all $t \in [t_{i-1}^n, t_i^n)$,
\begin{align*}
\mathrm{d}|\Delta\tilde{V}_t^n|^2 &= 2\Delta\tilde{V}_t^n \big(\delta_t \tilde{\sigma}(t,S_t) - \delta_{t_{i-1}^n} \tilde{\sigma}(t,S_t^n) + \delta_t\alpha_t \tilde{\lambda}(t,S_t)\big) \mathrm{d}W_t \\
&\quad + 2\Delta\tilde{V}_t^n \big(\delta_t \tilde{\nu}(t,S_t) - \delta_{t_{i-1}^n} \tilde{\nu}(t,S_t^n) + \frac{1}{2} \alpha_t^2 \tilde{\lambda}(t,S_t) + \delta_t\beta_t \tilde{\lambda}(t,S_t) + \delta_t \alpha_t (\tilde{\sigma}\partial_x \tilde{\lambda})(t,S_t)\big) \mathrm{d}t \\
&\quad - 2\Delta\tilde{V}_t^n \delta_{t_{i-1}^n} \big((\delta_t - \delta_{t_{i-1}^n})(\tilde{\sigma} \partial_x \tilde{\lambda})(t,S_t^n) + \alpha_t \tilde{\lambda}(t,S_t^n)\big)\mathrm{d}W_t \\
&\quad - 2\Delta\tilde{V}_t^n \delta_{t_{i-1}^n} (\delta_t - \delta_{t_{i-1}^n})\left(\partial_t \tilde{\lambda} + \tilde{\nu} \partial_x \tilde{\lambda} + \frac{1}{2} \tilde{\sigma}^2 \partial_{xx}\tilde{\lambda} \right)(t,S_t^n)\mathrm{d}t \\
&\quad - 2\Delta\tilde{V}_t^n \delta_{t_{i-1}^n} \big(\beta_t \tilde{\lambda}(t,S_t^n) + \alpha_t (\tilde{\sigma}\partial_x \tilde{\lambda})(t,S_t^n)\big)\mathrm{d}t \\
&\quad - 2\Delta\tilde{V}_t^n \frac{1}{2} |\delta_t - \delta_{t_{i-1}^n}|^2 (\tilde{\sigma} \partial_x \tilde{\lambda})(t,S_t^n) \mathrm{d}W_t \\
&\quad - 2\Delta\tilde{V}_t^n (\delta_t - \delta_{t_{i-1}^n}) \alpha_t \tilde{\lambda}(t,S_t^n) \mathrm{d}W_t \\
&\quad - 2\Delta\tilde{V}_t^n \frac{1}{2} |\delta_t - \delta_{t_{i-1}^n}|^2 \left(\partial_t \tilde{\lambda} + \tilde{\nu} \partial_x \tilde{\lambda} + \frac{1}{2} \tilde{\sigma}^2 \partial_{xx}\tilde{\lambda} \right)(t,S_t^n) \mathrm{d}t \\
&\quad - 2\Delta\tilde{V}_t^n \frac{1}{2} \big(2(\delta_t - \delta_{t_{i-1}^n})\beta_t + \alpha_t^2\big) \tilde{\lambda}(t,S_t^n) \mathrm{d}t \\
&\quad - 2\Delta\tilde{V}_t^n (\delta_t - \delta_{t_{i-1}^n}) \alpha_t (\tilde{\sigma} \partial_x \tilde{\lambda})(t,S_t^n) \mathrm{d}t \\
&\quad + \big|\big(\delta_t \tilde{\sigma}(t,S_t) - \delta_{t_{i-1}^n} \tilde{\sigma}(t,S_t^n) + \delta_t\alpha_t \tilde{\lambda}(t,S_t)\big) \\
&\qquad - \delta_{t_{i-1}^n} \big((\delta_t - \delta_{t_{i-1}^n})(\tilde{\sigma} \partial_x \tilde{\lambda})(t,S_t^n) + \alpha_t \tilde{\lambda}(t,S_t^n)\big) \\
&\qquad - \frac{1}{2} |\delta_t - \delta_{t_{i-1}^n}|^2 (\tilde{\sigma} \partial_x \tilde{\lambda})(t,S_t^n) \\
&\qquad - (\delta_t - \delta_{t_{i-1}^n}) \alpha_t \tilde{\lambda}(t,S_t^n)\big|^2 \mathrm{d}t.
\end{align*}
Therefore, for all $t \in [t_{i-1}^n,t_i^n)$,
\begin{align*}
    \mathbb{E}[|\Delta\tilde{V}_t^n|^2] &= \mathbb{E}[|\Delta V_{t_{i-1}^n}^n|^2] \\
    &\quad + 2\mathbb{E}\left[\int_{t_{i-1}^n}^{t}{\Delta\tilde{V}_s^n \big(\delta_s \tilde{\nu}(s,S_s) - \delta_{t_{i-1}^n} \tilde{\nu}(s,S_s^n) \big)\,\mathrm{d}s}\right] \\
    &\quad + 2\mathbb{E}\left[\int_{t_{i-1}^n}^{t}{\Delta\tilde{V}_s^n \frac{1}{2}\alpha_s^2 \big( \tilde{\lambda}(s,S_s) - \tilde{\lambda}(s,S_s^n) \big)\,\mathrm{d}s}\right] \\
    &\quad + 2\mathbb{E}\left[\int_{t_{i-1}^n}^{t}{\Delta\tilde{V}_s^n \beta_s \big( \delta_s \tilde{\lambda}(s,S_s) - \delta_{t_{i-1}^n} \tilde{\lambda}(s,S_s^n) \big)\,\mathrm{d}s}\right] \\
    &\quad + 2\mathbb{E}\left[\int_{t_{i-1}^n}^{t}{\Delta\tilde{V}_s^n \alpha_s \big( \delta_s (\tilde{\sigma}\partial_x \tilde{\lambda})(s,S_s) - \delta_{t_{i-1}^n} (\tilde{\sigma}\partial_x \tilde{\lambda})(s,S_s^n) \big)\,\mathrm{d}s}\right] \\
    &\quad - 2\mathbb{E}\left[\int_{t_{i-1}^n}^{t}{\Delta\tilde{V}_s^n \delta_{t_{i-1}^n} (\delta_s - \delta_{t_{i-1}^n})\left(\partial_t \tilde{\lambda} + \tilde{\nu} \partial_x \tilde{\lambda} + \frac{1}{2} \tilde{\sigma}^2 \partial_{xx}\tilde{\lambda} \right)(s,S_s^n)\,\mathrm{d}s}\right] \\
    &\quad - 2\mathbb{E}\left[\int_{t_{i-1}^n}^{t}{\Delta\tilde{V}_s^n \frac{1}{2} |\delta_s - \delta_{t_{i-1}^n}|^2 \left(\partial_t \tilde{\lambda} + \tilde{\nu} \partial_x \tilde{\lambda} + \frac{1}{2} \tilde{\sigma}^2 \partial_{xx}\tilde{\lambda} \right)(s,S_s^n)\,\mathrm{d}s}\right] \\
    &\quad - 2\mathbb{E}\left[\int_{t_{i-1}^n}^{t}{\Delta\tilde{V}_s^n (\delta_s - \delta_{t_{i-1}^n}) \big(\beta_s \tilde{\lambda}(s,S_s^n) + \alpha_s (\tilde{\sigma}\partial_x \tilde{\lambda})(s,S_s^n) \big)\,\mathrm{d}s}\right] \\
    &\quad + \mathbb{E}\left[\int_{t_{i-1}^n}^{t}{|\gamma_s^{i,n}|^2 \,\mathrm{d}s}\right],
\end{align*}
with $\gamma^{i,n}$ defined such as for all $t \in [t_{i-1}^n,t_i^n)$,
\begin{align*}
    \gamma_t^{i,n} &:= \delta_t \tilde{\sigma}(t,S_t) - \delta_{t_{i-1}^n} \tilde{\sigma}(t,S_t^n) \\
    &\quad + \alpha_t \big(\delta_t \tilde{\lambda}(t,S_t) - \delta_{t_{i-1}^n} \tilde{\lambda}(t,S_t^n)\big) \\
    &\quad - (\delta_t - \delta_{t_{i-1}^n})\big(\delta_{t_{i-1}^n} (\tilde{\sigma} \partial_x \tilde{\lambda})(t,S_t^n) + \frac{1}{2}(\delta_t - \delta_{t_{i-1}^n})(\tilde{\sigma} \partial_x \tilde{\lambda})(t,S_t^n) + \alpha_t \tilde{\lambda}(t,S_t^n)\big).
\end{align*}
Hence, this gives that for all $t \in [t_{i-1}^n,t_i^n)$,
\begin{align*}
    \mathbb{E}[|\Delta\tilde{V}_t^n|^2] &\leq \mathbb{E}[|\Delta V_{t_{i-1}^n}^n|^2] \\
    &\quad + \mathbb{E}\left[\int_{t_{i-1}^n}^{t}{\big(|\Delta\tilde{V}_s^n|^2 + |\delta_s \tilde{\nu}(s,S_s) - \delta_{t_{i-1}^n} \tilde{\nu}(s,S_s^n)|^2\big)\,\mathrm{d}s}\right] \\
    &\quad + \mathbb{E}\left[\int_{t_{i-1}^n}^{t}{\big(|\Delta\tilde{V}_s^n|^2 + \left|\frac{1}{2}\alpha_s^2\right|^2 |\tilde{\lambda}(s,S_s) - \tilde{\lambda}(s,S_s^n)|^2\big)\,\mathrm{d}s}\right] \\
    &\quad + \mathbb{E}\left[\int_{t_{i-1}^n}^{t}{\big(|\Delta\tilde{V}_s^n|^2 + |\beta_s|^2 |\delta_s \tilde{\lambda}(s,S_s) - \delta_{t_{i-1}^n} \tilde{\lambda}(s,S_s^n)|^2 \big)\,\mathrm{d}s}\right] \\
    &\quad + \mathbb{E}\left[\int_{t_{i-1}^n}^{t}{\big(|\Delta\tilde{V}_s^n|^2 + |\alpha_s|^2 |\delta_s (\tilde{\sigma}\partial_x \tilde{\lambda})(s,S_s) - \delta_{t_{i-1}^n} (\tilde{\sigma}\partial_x \tilde{\lambda})(s,S_s^n)|^2 \big)\,\mathrm{d}s}\right] \\
    &\quad + \mathbb{E}\left[\int_{t_{i-1}^n}^{t}{\left(|\Delta\tilde{V}_s^n|^2 + \left|\left(\partial_t \tilde{\lambda} + \tilde{\nu} \partial_x \tilde{\lambda} + \frac{1}{2} \tilde{\sigma}^2 \partial_{xx}\tilde{\lambda} \right)(s,S_s^n)\right|^2 |\delta_{t_{i-1}^n}(\delta_s - \delta_{t_{i-1}^n})|^2\right)\,\mathrm{d}s}\right] \\
    &\quad + \mathbb{E}\left[\int_{t_{i-1}^n}^{t}{\left(|\Delta\tilde{V}_s^n|^2 + \left|\frac{1}{2} \left(\partial_t \tilde{\lambda} + \tilde{\nu} \partial_x \tilde{\lambda} + \frac{1}{2} \tilde{\sigma}^2 \partial_{xx}\tilde{\lambda} \right)(s,S_s^n)\right|^2 |\delta_s - \delta_{t_{i-1}^n}|^4\right)\,\mathrm{d}s}\right] \\
    &\quad + \mathbb{E}\left[\int_{t_{i-1}^n}^{t}{\left(|\Delta\tilde{V}_s^n|^2 +  \left|\big(\beta_s \tilde{\lambda}(s,S_s^n) + \alpha_s (\tilde{\sigma}\partial_x \tilde{\lambda})(s,S_s^n)\big)\right|^2 |\delta_s - \delta_{t_{i-1}^n}|^2\right)\,\mathrm{d}s}\right] \\
    &\quad + 5\mathbb{E}\left[\int_{t_{i-1}^n}^{t}{|\delta_s \tilde{\sigma}(s,S_s) - \delta_{t_{i-1}^n} \tilde{\sigma}(s,S_s^n)|^2 \,\mathrm{d}s}\right] \\
    &\quad + 5\mathbb{E}\left[\int_{t_{i-1}^n}^{t}{|\alpha_s|^2 |\delta_s \tilde{\lambda}(s,S_s) - \delta_{t_{i-1}^n} \tilde{\lambda}(s,S_s^n)|^2 \,\mathrm{d}s}\right] \\
    &\quad + 5\mathbb{E}\left[\int_{t_{i-1}^n}^{t}{|(\tilde{\sigma} \partial_x \tilde{\lambda})(s,S_s^n)|^2 |\delta_{t_{i-1}^n}(\delta_s - \delta_{t_{i-1}^n})|^2\,\mathrm{d}s}\right] \\
    &\quad + 5\mathbb{E}\left[\int_{t_{i-1}^n}^{t}{\left|\frac{1}{2}(\tilde{\sigma} \partial_x \tilde{\lambda})(s,S_s^n)\right|^2 |\delta_s - \delta_{t_{i-1}^n}|^4\,\mathrm{d}s}\right] \\
    &\quad + 5\mathbb{E}\left[\int_{t_{i-1}^n}^{t}{|\alpha_s \tilde{\lambda}(s,S_s^n)|^2 |\delta_s - \delta_{t_{i-1}^n}|^2\,\mathrm{d}s}\right].
\end{align*}
which implies that there exists a constant $C_4 > 0$ such that for all $t \in [t_{i-1}^n,t_i^n)$,
\begin{align*}
\mathbb{E}[|\Delta\tilde{V}_t^n|^2] &\leq \mathbb{E}[|\Delta V_{t_{i-1}^n}^n|^2] \\
&\quad + C_4 \int_{t_{i-1}^n}^{t}{\big(\mathbb{E}[|\Delta\tilde{V}_s^n|^2] + \mathbb{E}[|S_s - S_s^n|^2] + \mathbb{E}[|\delta_s - \delta_s^n|^2]\big) \,\mathrm{d}s} \\
&\quad + C_4 \int_{t_{i-1}^n}^{t}{\big(\mathbb{E}[|\delta_{t_{i-1}^n}(\delta_s - \delta_{t_{i-1}^n})|^2] + \mathbb{E}[|\delta_s - \delta_{t_{i-1}^n}|^4]\big) \,\mathrm{d}s} \\
&\leq \mathbb{E}[|\Delta V_{t_{i-1}^n}^n|^2] \\
&\quad + C_4 \int_{t_{i-1}^n}^{t}{\big(\mathbb{E}[|\Delta\tilde{V}_s^n|^2] + \mathbb{E}[|S_s - S_s^n|^2] + \mathbb{E}[|\delta_s - \delta_s^n|^2]\big) \,\mathrm{d}s} \\
&\quad + C_4 \int_{t_{i-1}^n}^{t}{\left(\sqrt{\mathbb{E}[|\delta_{t_{i-1}^n}|^4]} \sqrt{\mathbb{E}[|\delta_s - \delta_{t_{i-1}^n}|^4]} + \mathbb{E}[|\delta_s - \delta_{t_{i-1}^n}|^4]\right) \,\mathrm{d}s}. \numberthis \label{proof discrete trading strategy V 4}
\end{align*}
Besides by the generalized Minkowski inequality and the Burkholder-Davis-Gundy inequality (see for instance Section 7.8.1 in \cite{pages2018numerical}) we have that there exists $K > 0$ such that for all $t \in [t_{i-1}^n,t_i^n)$,
$$ \big\Vert\delta_t - \delta_{t_{i-1}^n}\big\Vert_{4} \leq K \sqrt{\Delta_n},$$
hence
\begin{equation}
\label{useful result 2}
\sqrt{\mathbb{E}[|\delta_s - \delta_{t_{i-1}^n}|^4]} = \big\Vert\delta_t - \delta_{t_{i-1}^n}\big\Vert_{4}^2 \leq K^2 \Delta_n.
\end{equation}
By plugging (\ref{useful result 1}), (\ref{useful result 2}), Propositions \ref{prop discrete trading strategy delta} and \ref{prop discrete trading strategy spot} in the inequality (\ref{proof discrete trading strategy V 4}), we obtain that there exists a constant $C_5 > 0$ such that for all $t \in [t_{i-1}^n,t_i^n)$,
\begin{equation}
\label{proof discrete trading strategy V 5}
\mathbb{E}[|\Delta\tilde{V}_t^n|^2] \leq \mathbb{E}[|\Delta V_{t_{i-1}^n}^n|^2] + C_5 \Delta_n^2 + C_5 \int_{t_{i-1}^n}^{t}{\mathbb{E}[|\Delta \tilde{V}_s^n|^2]\,\mathrm{d}s}.
\end{equation}
Hence by Gronwall's lemma we have that for all $t \in [t_{i-1}^n,t_i^n)$,
\begin{align*}
\mathbb{E}[|\Delta\tilde{V}_t^n|^2] &\leq \exp(C_5(t-t_{i-1}^n))\left(\mathbb{E}[|\Delta V_{t_{i-1}^n}^n|^2] + C_5 \Delta_n^2\right) \\
&\leq \exp(C_5 \Delta_n)\left(\mathbb{E}[|\Delta V_{t_{i-1}^n}^n|^2] + C_5 \Delta_n^2\right), \numberthis \label{proof discrete trading strategy V 6}
\end{align*}
wich leads to, by plugging the inequality (\ref{proof discrete trading strategy V 6}) in (\ref{proof discrete trading strategy V 5}), that for all $t \in [t_{i-1}^n,t_i^n)$,
\begin{align*}
\mathbb{E}[|\Delta\tilde{V}_t^n|^2] &\leq \mathbb{E}[|\Delta V_{t_{i-1}^n}^n|^2] + C_5 \Delta_n^2 + C_5 \Delta_n\exp(C_5 \Delta_n)\left(\mathbb{E}[|\Delta V_{t_{i-1}^n}^n|^2] + C_5 \Delta_n^2\right) \\
&\leq \mathbb{E}[|\Delta V_{t_{i-1}^n}^n|^2]\left(1 + \underbrace{C_5 \Delta_n \exp(C_5 \Delta_n)}_{= O(\Delta_n)}\right) + \underbrace{C_5 \Delta_n^2 + C_5^2 \Delta_n^3 \exp(C_5 \Delta_n)}_{= O(\Delta_n^2)}.
\end{align*}
Therefore there exists a constant $C_6 > 0$ such that for all $t \in [t_{i-1}^n,t_i^n)$,
\begin{equation}
\label{proof discrete trading strategy V 7}
\mathbb{E}[|\Delta\tilde{V}_t^n|^2] \leq \mathbb{E}[|\Delta V_{t_{i-1}^n}^n|^2](1 + C_6 \Delta_n) + C_6 \Delta_n^2.
\end{equation}
Let $(\tau_k^{i,n})_{k \in \mathbb{N}}$ a non-decreasing sequence such that $\lim\limits_{k \rightarrow +\infty}{\tau_k^{i,n}} = t_i^n$, by Fatou's lemma we have
$$ \mathbb{E}\left[\liminf\limits_{k \rightarrow +\infty}{|\Delta \tilde{V}_{\tau_k^{i,n}}^n|^2}\right] \leq \liminf\limits_{k \rightarrow +\infty}{\mathbb{E}[|\Delta \tilde{V}_{\tau_k^{i,n}}^n|^2]} \leq \mathbb{E}[|\Delta V_{t_{i-1}^n}^n|^2](1 + C_6 \Delta_n) + C_6 \Delta_n^2, $$
and
$$ \mathbb{E}\left[\liminf\limits_{k \rightarrow +\infty}{|\Delta \tilde{V}_{\tau_k^{i,n}}^n|^2}\right] = \mathbb{E}\left[\lim\limits_{k \rightarrow +\infty}{|\Delta \tilde{V}_{\tau_k^{i,n}}^n|^2}\right] = \mathbb{E}[|\Delta V_{t_i^n}^n|^2] $$
since $\lim\limits_{t \rightarrow (t_i^n)^-}{\Delta\tilde{V}_t^n} = \Delta V_{t_i^n}^n$. Therefore we have for all $n \in \mathbb{N}^*, i \in \llbracket 1,n \rrbracket$,
\begin{align*}
\mathbb{E}[|\Delta V_{t_i^n}^n|^2] &\leq \mathbb{E}[|\Delta V_{t_{i-1}^n}^n|^2](1 + C_6 \Delta_n) + C_6 \Delta_n^2 \\
&\leq \mathbb{E}[|\Delta V_{t_{i-2}^n}^n|^2](1 + C_6 \Delta_n)^2 + C_6 \Delta_n^2(1 + C_6 \Delta_n) + C_6 \Delta_n^2 \\
&\ldots \\
&\leq \mathbb{E}[|\Delta V_{t_0^n}^n|^2]\left(1 + C_6 \Delta_n\right)^i + C_6 \Delta_n^2\sum_{k=0}^{i-1}{(1 + C_6 \Delta_n)^k}.
\end{align*}
Considering the fact that $\Delta V_{t_0^n}^n = 0$, we have for all $n \in \mathbb{N}^*, i \in \llbracket 1,n \rrbracket$,
\begin{align*}
\mathbb{E}[|\Delta V_{t_i^n}^n|^2] &\leq C_6 \Delta_n^2\frac{(1 + C_6 \Delta_n)^i - 1}{C_6 \Delta_n} \\
&\leq \Delta_n (1 + C_6 \Delta_n)^n \\
&\leq \Delta_n \exp(n \ln(1 + C_6 \Delta_n)) \\
&\leq \underbrace{\Delta_n \exp(C_6 n \Delta_n)}_{O(\Delta_n)}.
\end{align*}
Hence there exists a constant $C_7 > 0$ such that for all $n \in \mathbb{N}^*, i \in \llbracket 0,n \rrbracket$,
$$ \mathbb{E}[|\Delta V_{t_i^n}^n|^2] \leq C_7 \Delta_n $$
and for all $n \in \mathbb{N}^*, i \in \llbracket 0,n \rrbracket, t \in [t_{i-1}^n, t_i^n),$
\begin{align*}
\mathbb{E}[|\Delta V_t^n|^2] &\leq \mathbb{E}[|\Delta V_{t_{i-1}^n}^n|^2](1 + C_3 \Delta_n) + C_3 \Delta_n \\
&\leq \underbrace{C_7 \Delta_n (1 + C_3 \Delta_n) + C_3 \Delta_n}_{= O(\Delta_n)}
\end{align*}
which implies that there exists a constant $C_8 > 0$ such that for all $n \in \mathbb{N}^*, i \in \llbracket 1,n \rrbracket$,
$$ \sup\limits_{t \in [t_{i-1}^n,t_i^n]}{\mathbb{E}[|\Delta V_t^n|^2]} \leq C_8 \Delta_n. $$
This gives that for all $n \in \mathbb{N}^*$,
$$ \sup\limits_{t \in [0,T]}{\mathbb{E}[|\Delta V_t^n|^2]} = \max\limits_{1 \leq i \leq n}{\sup\limits_{t \in [t_{i-1}^n,t_i^n]}{\mathbb{E}[|\Delta V_t^n|^2]}} \leq C_8 \Delta_n. $$
\end{proof}

\end{document}